\documentclass[11pt,twoside]{atmp}
\usepackage{graphicx} 

\newtheorem{theorem}{Theorem}[section]
\newtheorem{lemma}[theorem]{Lemma}

\theoremstyle{definition}
\newtheorem{definition}[theorem]{Definition}

\theoremstyle{remark}
\newtheorem{remark}[theorem]{Remark}

\numberwithin{equation}{section}

\input{tcilatex}
\begin{document}
\title[Ricci Flow Conjugated Initial data sets  ...]
{Ricci Flow Conjugated Initial Data Sets for Einstein Equations}
\author{MAURO CARFORA}
\address{Dipartimento di Fisica Nucleare e Teorica, Universita` degli Studi di Pavia\\
and\\ 
Istituto Nazionale di Fisica Nucleare, Sezione di Pavia\\
via A. Bassi 6, I-27100 Pavia, Italy}
\email{mauro.carfora@pv.infn.it}
\thanks{Research supported in part by PRIN Grant $\# 20082K9KXZ-004$.}

\begin{abstract}
We discuss a natural form of Ricci--flow conjugation between two distinct general relativistic data sets given on a compact $n\geq 3$-dimensional manifold $\Sigma$. We establish the existence of the relevant entropy functionals for the matter and geometrical variables, their monotonicity properties, and the associated convergence in the appropriate sense. We show that in such a framework there is a natural mode expansion generated by the spectral resolution of the Ricci conjugate Hodge--DeRham operator. This mode expansion allows to compare the two distinct data sets and gives rise to a computable heat kernel expansion of the fluctuations among the fields defining the data.
In particular this shows that Ricci flow conjugation entails a natural form of $L^2$ parabolic averaging of one data set  with respect to the other 
with a number of desiderable properties: \emph{(i)} It preserves the dominant energy condition; \emph{(ii)} It is  localized by a heat kernel whose support sets the scale of averaging; \emph{(iii)} It is characterized by a set of balance functionals which allow the analysis of its entropic stability. 
\end{abstract}

\maketitle
\vfill \eject

\section{Introduction}

The study of initial data sets for Einstein equations\cite{Jim} is a well--developed part of mathematical relativity with  seminal interactions with geometric analysis. Examples abound, and a fine selection is provided by  the interplay between minimal surfaces, mean curvature flow, and asymptotically flat data\cite{HuIlm,HuIlm2}, by the connection between the positive mass theorem and the proof of the Yamabe problem\cite{Au,S1}, and most recently by the engineering of new data sets out of sophisticated gluing techniques\cite{chru, corvino1, corvino2}.  These results often go beyond the motivating physics, and clearly indicate that being the carrier of an  Einstein initial data set $\mathcal{C}_g(\Sigma)$ is an important geometrical characterization for  an $n$--dimensional Riemannian manifold $\Sigma $.
\begin{figure}[h]
\includegraphics[scale=0.8]{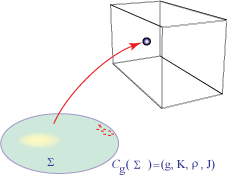}
\caption{The initial data set $\mathcal{C}_g(\Sigma)$  as a point in the space $\mathcal{T}\,\mathcal{M}et(\Sigma)\times C^{\infty }(\Sigma ,\mathbb{R}_+)\times C^{\infty }(\Sigma ,T\Sigma)$.}
\label{fig:1}      
\end{figure}
\\
If we denote by $\mathcal{T}_{(\Sigma,g)}\,\mathcal{M}et(\Sigma)$ the tangent space (at $(\Sigma,g)$) to the manifold $\mathcal{M}et(\Sigma)$ of Riemannian metrics $g$ on $\Sigma$, by $C^{\infty }(\Sigma ,\mathbb{R}_+)$ the space of smooth non--negative functions, and by 
$C^{\infty }(\Sigma ,T\Sigma)$ the space of smooth vector fields on $\Sigma$, then  
a (generalized) Einstein initial data set,
\begin{equation}
\label{dtst}
\mathcal{C}_g(\Sigma):=(g,\,K,\,\varrho,\,J)\in\mathcal{T}\,\mathcal{M}et(\Sigma)\times C^{\infty }(\Sigma ,\mathbb{R}_+)\times C^{\infty }(\Sigma ,T\Sigma)\;,
\end{equation}
is defined by a Riemannian metric  $g\in \mathcal{M}et(\Sigma)$, a symmetric bilinear form 
$K\in \mathcal{T}_{(\Sigma,g)}\,\mathcal{M}et(\Sigma) $, \footnote{We can think of the pair $(g,K)$ as a point of the tangent bundle $\mathcal{T}\,\mathcal{M}et(\Sigma)$ to $\mathcal{M}et(\Sigma)$.}  \, a scalar field $\varrho\in C^{\infty }(\Sigma ,\mathbb{R}_+)$, and a vector field  
${J}\in C^{\infty }(\Sigma ,T\Sigma)$, constrained by the dominant energy condition $\varrho \geq |J|$, and by
the Hamiltonian and the divergence constraints 
\begin{eqnarray}
\mathcal{R}(g)-\left(2\Lambda + |K|^{2}_{g} -(tr_{g}\,K)^{2}\right) &=&16\pi \varrho 
\;,  \label{constraint0} \\
2\,\nabla _{a}\left(K^{ab}-g^{ab}\,(tr_{g}\,K)\right) &=&16\pi J^{b}\;.  \label{constraints0}
\end{eqnarray}
Here $\Lambda $ is a constant (the cosmological constant), $|K|^{2}_{g}:=K_{\;\,b}^{a}K_{\;\,a}^{b}$,
$tr_{g}\,K:=g^{ab}K_{ab}$, and  $\mathcal{R}(g)$ \ is the scalar curvature of the Riemannian metric $g$. This characterization of $\mathcal{C}_g(\Sigma)$ is a natural generalization of  the notion of  $3$--dimensional initial data set for Einstein equations  where  the
symmetric tensor field $K$ can be interpreted as the extrinsic
curvature of the embedding $i_{t}:\Sigma
\rightarrow M^{(4)}$ of $(\Sigma ,g)$ in the spacetime $(M^{(4)}\simeq
\Sigma \times \mathbb{R},g^{(4)})$ resulting from the Einstein evolution of $\mathcal{C}_g(\Sigma)$, whereas $\varrho $ and $J$ can be respectively identified with the mass density and the momentum density of
the material self--gravitating sources on $(\Sigma ,g)$. In full generality, to (\ref{constraint0}) and (\ref{constraints0}) one should also add the set of additional constraints of non--gravitational origin associated with the dynamics of the  sources. 
In order to avoid specifying the precise nature of the matter fields, here we  represents these fields in the initial data with the pair $(\rho,J)$,  only requiring that the  dominant energy condition $\rho\geq\,|J|$ holds. From a geometric perspective it is worthwhile recalling that the set of solutions to the constraint equations (\ref{constraint0}) and (\ref{constraints0}) is, under suitable conditions, an $\infty $--dimensional submanifold of the configurational space 
$\mathcal{T}\,\mathcal{M}et(\Sigma)\times C^{\infty }(\Sigma ,\mathbb{R}_+)\times C^{\infty }(\Sigma ,T\Sigma)$, \cite{Bart, Chrusc,Fisch}. This is related to the fact that, from a geometric analysis point of view, the constraints 
(\ref{constraint0}) and (\ref{constraints0}) provide an undetermined system of coupled (elliptic) PDEs \cite{Jim, confChoq, confLich} . It is precisely such a property that allows for a subtle interaction with the additional geometrical structures the manifold $\Sigma$ may be endowed with,  and  is responsible for the geometrical richness of the notion of Einstein initial data sets alluded above.
\begin{figure}[h]
\includegraphics[scale=0.8]{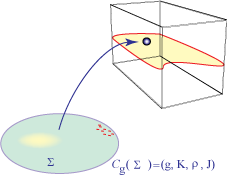}
\caption{Under suitable conditions initial data sets $\mathcal{C}_g(\Sigma)$ can be considered as  points of a submanifold in the constraints configurational space $\mathcal{T}\,\mathcal{M}et(\Sigma)\times C^{\infty }(\Sigma ,\mathbb{R}_+)\times C^{\infty }(\Sigma ,T\Sigma)$.}
\label{fig:2}   
\end{figure}
\\
Among the additional structures that may decorate the constraints configurational space $\mathcal{T}\mathcal{M}et(\Sigma)\times C^{\infty }(\Sigma ,\mathbb{R}_+)\times C^{\infty }(\Sigma ,T\Sigma)$ the one that interests us in this article is related to the Ricci flow introduced by R. Hamilton \cite{11}
\begin{equation} 
\begin{tabular}{l}
$\frac{\partial }{\partial \beta }g_{ab}(\beta )=-2\mathcal{R}_{ab}(\beta ),$ \\ 
\\ 
$g_{ab}(\beta =0)=g_{ab}$\, ,\;\; $0\leq \beta <T_{0}$\;,%
\end{tabular}
   \label{Imflow}
\end{equation}
where $\mathcal{R}_{ab}(\beta )$ is the Ricci tensor of the metric $g_{ik}(\beta )$. This weakly-parabolic diffusion--reaction PDE defines a  flow on the space of Riemannian metrics $\mathcal{M}et(\Sigma)$ \footnote{The Ricci flow is not the only natural (\emph{i.e.} $\mathcal{D}iff(\Sigma)$--equivariant) geometric flow on $\mathcal{M}et(\Sigma)$, other examples that may come to mind are the Yamabe flow and the Cross-curvature flow. In our setting the Ricci flow comes to the fore because it interacts with the constraints configurational space in a very natural way, as will be evident from the analysis presented here.},\, which extends naturally to the constraints configurational space  via linearization and via the (Ricci flow induced) scalar and vector heat flows on $C^{\infty }(\Sigma ,\mathbb{R}_+)\times C^{\infty }(\Sigma ,T\Sigma)$.
In particular, the linearized Ricci flow and  the backward conjugated linearized Ricci flow induce on $\mathcal{T}\,\mathcal{M}et(\Sigma)\times C^{\infty }(\Sigma ,\mathbb{R}_+)\times C^{\infty }(\Sigma ,T\Sigma)$ a parabolic conjugation with strong averaging properties \cite{carfback}. 
This suggests that there may be a non--trivial interplay between Ricci flow and Einstein initial data sets. The key idea in such a scenario is that the Ricci flow, even if it cannot evolve along the constraint manifold \footnote{This is a consequence of the weak--parabolicity of the Ricci flow and of the fact that the  constraints are in involution with respect to the hyperbolic evolution associated with the evolutive part of the Einstein equations.},  may interpolate in the configurational space $\mathcal{T}\,\mathcal{M}et(\Sigma)\times C^{\infty }(\Sigma ,\mathbb{R}_+)\times C^{\infty }(\Sigma ,T\Sigma)$  between distinct initial data sets.  The above remarks suggest that when a Ricci flow interpolation exists it is a form of parabolic conjugation and as such it may provide a natural geometrical way of comparing an Einstein data set $\mathcal{C}_g(\Sigma):=(g,\,K,\,\varrho,\,J)$ with  a given reference data set $\overline{\mathcal{C}}_{\overline{g}}(\Sigma):=(\overline{g},\,\overline{K},\,\overline{\varrho},\,\overline{J})$. \\
\\
\noindent
\begin{figure}[h]
\includegraphics[scale=0.4]{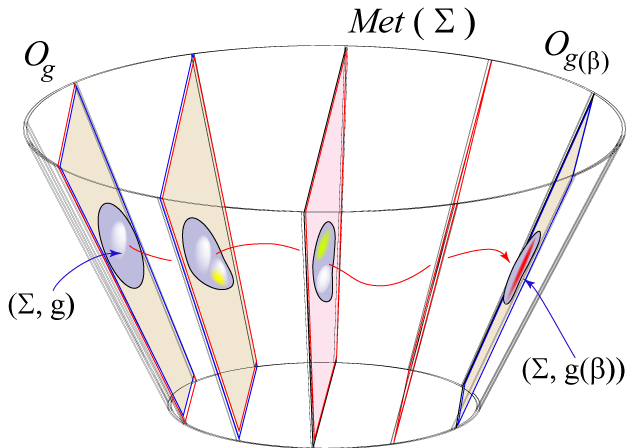}
\caption{The Ricci flow in the space of Riemannian metrics $\mathcal{M}et(\Sigma)$.}
\label{fig:3}   
\end{figure}
\\
In order to make such heuristic remarks more precise we introduce the following set of definitions characterizing Ricci flow conjugation.
\begin{definition} (\emph{Physical data vs. Reference data})\\
\label{refdefn}
\noindent Let $\mathcal{C}_g(\Sigma):=(g,\,K,\,\varrho,\,J)$ and $\overline{\mathcal{C}}_{\overline{g}}(\Sigma):=(\overline{g},\,\overline{K},\,\overline{\varrho},\,\overline{J})$ denote two distinct initial data sets on a $C^{\infty }$ compact $n\geq3$--dimensional manifold without boundary  $\Sigma$. 
Both $\mathcal{C}_g(\Sigma)$ and $\overline{\mathcal{C}}_{\overline{g}}(\Sigma)$ are supposed to satisfy the corresponding dominant energy condition and the Hamiltonian and divergence constraints. For $\mathcal{C}_g(\Sigma):=(g,\,K,\,\varrho,\,J)$ these are provided by (\ref{constraint0}) and (\ref{constraints0}), while for  $\overline{\mathcal{C}}_{\overline{g}}(\Sigma):=(\overline{g},\,\overline{K},\,\overline{\varrho},\,\overline{J})$ they are explicitly given 
by $\overline{\varrho} \geq |\overline{J}|$,  and 
\begin{eqnarray}
\mathcal{R}(\overline{g})-\left(2\overline{\Lambda} + |\overline{K}|^{2}_{\overline{g}} -(tr_{\overline{g}}\,\overline{K})^{2}\right) &=&16\pi \overline{\varrho} 
\;,  \label{constraint0bar} \\
2\,\overline{\nabla} _{a}\left(\overline{K}^{\;ab}-\overline{g}^{\;ab}\,(tr_{\overline{g}}\,\overline{K})\right) &=&16\pi \overline{J}^{\;b}\;.  \label{constraints0bar}
\end{eqnarray}
 where $\overline{\Lambda}$ is a cosmological constant, (possibly distinct from $\Lambda $), and  $\overline{\nabla}$ denotes \footnote{In what follows, we will often omit the overline over $\nabla$ since the meaning will be clear from the geometrical context.} the Levi--Civita connection associated with the Riemannian manifold $(\Sigma,\overline{g})$. The data set $\mathcal{C}_g(\Sigma):=(g,\,K,\,\varrho,\,J)$ will be conventionally referred to as the \emph{Physical Data} on $\Sigma$, whereas 
$\overline{\mathcal{C}}_{\overline{g}}(\Sigma):=(\overline{g},\,\overline{K},\,\overline{\varrho},\,\overline{J})$ will be called the \emph{Reference Data}. $\square$
\end{definition} 
 
\noindent 
In such a general setting, Ricci flow conjugation can be defined whenever the metric tensors $g$ and $\bar{g}$, associated with the two distinct data sets $\mathcal{C}_g(\Sigma)$ and $\bar{\mathcal{C}}_{\bar{g}}(\Sigma)$, are  connected by a fiducial, non--collapsing, Ricci flow $\beta\longmapsto g(\beta)$ of bounded geometry on $\Sigma\times [0,\beta^*]$.
\begin{definition}(\emph{Interpolating fiducial Ricci flow})\\
\label{Fidudefn}
A fiducial Ricci flow of bounded geometry interpolating between the two Riemannian manifolds $(\Sigma,g)$ and $(\Sigma, \overline{g})$  is  a non--collapsing  \footnote{The assumption of non--collapsing is necessary since torus bundles over the circle admit smooth Ricci flows with bounded geometry
which exist for all $\beta\in [0,\infty )$, and collapse as $\beta\rightarrow \infty $ \cite{HJim}.} solution of 
 the weakly--parabolic initial value problem
\begin{equation} 
\begin{tabular}{l}
$\frac{\partial }{\partial \beta }\,\,g_{ab}(\beta )=-2\,\mathcal{R}_{ab}(\beta ),$ \\ 
\\ 
$g_{ab}(\beta =0)=g_{ab}$\, ,\;\; $0\leq \beta \leq \beta^*$\;,%
\end{tabular}
   \label{mflow}
\end{equation}
such that $g_{ab}(\beta^*)=\bar{g}_{ab}$, and such that there exists constants $C_{k}>0$ for which $\left|\nabla ^{k}\,Rm(\beta )\right|\leq C_{k}$, \, $k=0,1,\ldots$, \, for $0\leq \beta \leq \beta ^{*}$.  We assume that any such a flow is contained in a corresponding maximal solution, $\beta \rightarrow g_{ab}(\beta)$, $0\leq \beta \leq \beta^*<T_0\leq \infty$,  with the same initial metric $g_{ab}(\beta =0)=g_{ab}$.\; $\square$
\end{definition}
\noindent
Recall that the maximal interval of existence, $[0,T_{0})$, for the flow (\ref{mflow}) is either $T_0\rightarrow \infty $ or  $\lim_{\beta \nearrow T_{0}}\, [\sup_{x\in \Sigma }\,|Rm(x,\beta )|]=\infty $,  whenever $T_{0}<\infty$, \cite{11,12, sesum} where $Rm(x,\beta )$  denotes the Riemann tensor of $(\Sigma ,g(\beta ))$ evaluated at the generic point   $x\in \Sigma $.
\begin{figure}[h]
\includegraphics[scale=0.5]{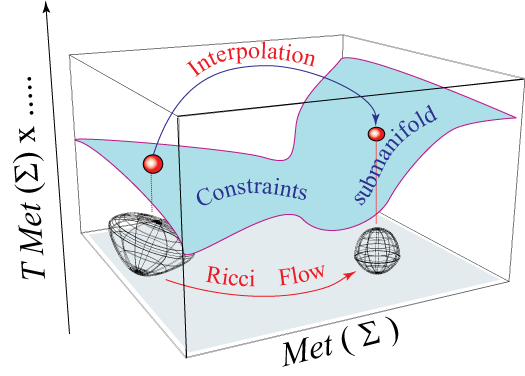}
\caption{The Ricci flow in the space of Riemannian metrics $\mathcal{M}et(\Sigma)$ may induce a flow 
in the configuration space $\mathcal{T}\,\mathcal{M}et(\Sigma)\times C^{\infty }(\Sigma ,\mathbb{R}_+)\times C^{\infty }(\Sigma ,T\Sigma)$ which interpolates between distinct initial data sets.}
\label{fig:4}   
\end{figure}
\\
\noindent
If between the two Riemannian manifolds $(\Sigma,g)$ and $(\Sigma, \overline{g})$ supporting the Einstein data $\mathcal{C}_g(\Sigma)$ and $\bar{\mathcal{C}}_{\bar{g}}(\Sigma)$ there exists an interpolating Ricci flow of bounded geometry then the  conjugation between $\mathcal{C}_g(\Sigma)$ and $\bar{\mathcal{C}}_{\bar{g}}(\Sigma)$ is characterized by the 
\begin{definition} (\emph{Ricci flow conjugation}\cite{carfback})\\
\label{theDEF}
\noindent Two distinct initial data set $\mathcal{C}_g(\Sigma)$ and $\bar{\mathcal{C}}_{\bar{g}}(\Sigma)$ are said to be Ricci flow conjugated on $\Sigma\times [0,\beta^*]$ along an interpolating Ricci flow of bounded geometry $\beta\longmapsto g(\beta)$, $0\leq\beta\leq\beta^*$, if they are connected by the flows 
$\mathcal{C}(\beta)\in C^{\infty }(\Sigma \times [0,\beta^*],\,\otimes ^p_S\,T^*\Sigma)$ and $\bar{\mathcal{C}}^{\;\sharp }(\eta)\in C^{\infty }(\Sigma \times [0,\beta^*],\,\otimes ^p_S\,T\Sigma)$, $p=0,1,2$, $\eta:=\beta^*-\beta$,

\begin{equation}
\beta\mapsto \mathcal{C}(\beta):=\,\left( \begin{array}{ll}
        \varrho(\beta)  \\
        J_{i}(\beta)  \\
        K_{ab}(\beta)   \end{array} \right)\,,\;\; \eta\mapsto \overline{\mathcal{C}}^{\;\sharp }(\eta):=\,\left( \begin{array}{ll}
        \overline{\varrho}(\eta)  \\
        \overline{J}^{i}(\eta)  \\
        \overline{K}^{ab}(\eta)   \end{array} \right)\;,
        \label{heatDef}
\end{equation}

\noindent respectively defined, along $\beta\longmapsto g(\beta)$, by the solutions of the Hodge--DeRham--Lichnerowicz  heat equation
 
\begin{equation} 
\begin{tabular}{l}
$\frac{\partial  }{\partial  \beta   }\,\,\mathcal{C}(\beta)=\Delta _{d}\, \mathcal{C}(\beta)\;,$ \\ 
\\ 
$\mathcal{C}(\beta=0)=\,\mathcal{C}_{g}(\Sigma)$\, ,\;\; $0\leq \beta \leq \beta^*$\;,%
\end{tabular}
  \label{linDT0comp0}
\end{equation}

\noindent and of the corresponding backward conjugated heat flow

\begin{equation} 
\begin{tabular}{l}
$\frac{\partial  }{\partial  \eta   }\,\,\bar{\mathcal{C}}^{\;\sharp }(\eta)=\Delta _{d}\, \overline{\mathcal{C}}^{\;\sharp }(\eta)-\mathcal{R}(g(\eta))\bar{\mathcal{C}}^{\;\sharp }(\eta)\;,$ \\ 
\\ 
$\bar{\mathcal{C}}^{\;\sharp }(\eta=0)=\,\bar{\mathcal{C}}_{g}^{\;\sharp }(\Sigma)$\, ,\;\; $0\leq \eta \leq \beta^*$\,,\;\;\;$\eta:=\beta^*-\beta$\;,%
\end{tabular}
  \label{conlinDT0comp0}
\end{equation}

\noindent along the time--reversed Ricci evolution $\eta\mapsto g(\eta)$.\; $\square$
\end{definition}

\begin{remark} Here  $\Delta _{d}:= -(d\,\delta_{g(\beta)}+\delta_{g(\beta)}\,d)$ is the Hodge Laplacian, with respect to Ricci evolving metric $\beta\longmapsto g(\beta)$, thought of as acting on $C^{\infty }(\Sigma \times [0,\beta^*],\,\otimes ^p_S\,T^*\Sigma)$, $p=0,1,2$, (recall that formally the  Hodge Laplacian on symmetric bilinear forms acts as the Lichnerowicz--DeRham Laplacian; see below for notation). It is also worthwhile to stress that the  flows (\ref{linDT0comp0}) and (\ref{conlinDT0comp0}) directly arise from the linearization of the Ricci flow along a metric perturbation $g_{ab}^{(\epsilon )}(\beta):=g_{ab}(\beta)+\epsilon\,h_{ab}(\beta)$, $\epsilon >0$,
\begin{equation} 
\begin{tabular}{l}
$\frac{\partial }{\partial \beta }\,h_{ab}(\beta)=-2\,\frac{d}{d\epsilon }\mathcal{R}^{(\epsilon )}_{ab}(\beta )|_{\epsilon =0}\,\doteq -2\,D\,\mathcal{R}ic(g(\beta))\circ \,h_{ab}(\beta)$\, \\ 
\\ 
$h_{ab}(\beta =0)=h_{ab}$\, ,\;\; $0\leq \beta \leq \beta^*$\;.%
\end{tabular}
   \label{linINTRO}
\end{equation}
Indeed, by considering scalar induced perturbations $h_{ab}(\beta):=2\,\nabla _a\nabla_b\,\varrho (\beta)$, vector induced perturbations $h_{ab}(\beta):=\nabla _a J_b(\beta)+\nabla _b J_a(\beta)$, and tensor perturbations $h_{ab}(\beta):=K_{ab}(\beta)$, one easily reduces (\ref{linINTRO}) to (\ref{linDT0comp0}) by naturally
fixing the action of the group of diffeomorphisms $\mathcal{D}iff(\Sigma)$, ( \cite{01}, (see also Chap.2 of \cite{chowluni}). A similar procedure generates the conjugate flow (\ref{conlinDT0comp0}). 
\end{remark}
\begin{figure}[h]
\includegraphics[scale=0.4]{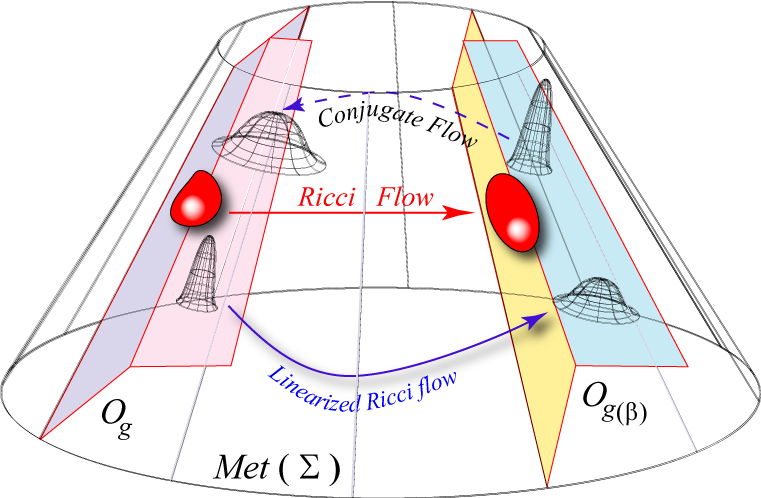}
\caption{The linearized Ricci flow induces a heat flow conjugation 
in the configuration space $\mathcal{T}\,\mathcal{M}et(\Sigma)\times C^{\infty }(\Sigma ,\mathbb{R}_+)\times C^{\infty }(\Sigma ,T\Sigma)$. In this picture the planes denote the $\mathcal{D}iff(\Sigma)$ orbits $\mathcal{O}_g$ and  $\mathcal{O}_{g(\beta)}$ of the manifolds $(\Sigma,g)$ and $(\Sigma,g(\beta))$. The orthogonal planes depict slices $L^2$--orthogonal to the orbits. They parametrize infinitesimally close orbits sampled by a (potentially) non trivial linearized Ricci flow and its conjugate flow. The bell shaped curves emphasize that the linearized Ricci flow and the  conjugated linearized Ricci flow can be reduced to geometrical heat flows.}
\label{fig:5}   
\end{figure}
\noindent
As already stressed, the parabolic nature of the interpolating flows $\beta\mapsto (g(\beta),{\mathcal{C}}(\beta))$ and $\eta\mapsto (g(\eta),\bar{\mathcal{C}}^{\;\sharp }(\eta))$ implies that they do not pointwise satisfy the constraints, (\ref{constraint0}) and (\ref{constraints0}), for $0<(\beta,\eta)<\beta^*$. Rather, they  entail a form of parabolic $L^2$--averaging of one data set $\mathcal{C}_g(\Sigma)$  with respect to the other $\bar{\mathcal{C}}_{\bar{g}}(\Sigma)$ (this latter taken, according to definition \ref{refdefn}, as the reference data set). In this paper we discuss these averaging properties in full detail. In particular we prove the existence of the relevant entropy functionals for the matter variables, their monotonicity property, and the associated convergence in the appropriate $L^2$--sense. We also show that in such a framework Perelman's energy $\mathcal{F}$ characterizes the energy--increasing and energy--decreasing reference trajectories associated with the divergence--free part of $\overline{K}^{\;ab}(\eta)$.

\subsection{Outline of the paper}
We end this introductory section by presenting a commented list of the main results proved and discussed in this work. Throughout the subsection we let $\beta\mapsto (g(\beta),{\mathcal{C}}(\beta))$ and $\eta\mapsto (g(\eta),\overline{\mathcal{C}}^{\;\sharp }(\eta))$ be the flows solution of 
(\ref{linDT0comp0}) and (\ref{conlinDT0comp0}) defining the conjugation between the physical data $\mathcal{C}_g(\Sigma)$ and the reference data $\overline{\mathcal{C}}_{\overline{g}}(\Sigma)$ on $\Sigma\times\,[0,\beta^*]$.\\
\\
\noindent
A first group of results concerns the localization properties of the distribution of the physical matter variables $(\varrho, J)$ with respect to the reference  data $(\overline{\varrho},\,\overline{J})$.  We have the
\begin{theorem}
The relative entropy functional
\begin{equation}
\mathcal{S}[d\Pi(\beta)| d\overline{\varrho}(\beta)]:=\int_{\Sigma }\,{\varrho }(\beta )\,\ln{\varrho  }(\beta )\,d\overline{\varrho}(\beta)\;,
\end{equation} 
where $d\overline{\varrho}(\beta):={\varrho }(\beta )d\mu_{g(\beta)}$ and $d\Pi(\beta):={\varrho }(\beta )\,d\overline{\varrho}(\beta)$,
is monotonically non--increasing along the 
flow $\beta \mapsto (g(\beta),d\overline{\varrho}(\beta))$.
Moreover, the matter distribution $\varrho (\beta)$ is localized in the entropy sense around  $\overline{\varrho}(\beta)$ according to
\begin{equation}
\frac{1}{2}\left\| d\Pi(\beta)
-d\overline{\varrho}(\beta)  \right\| _{\rm var}^{2}\leq\mathcal{S} [d\Pi(\beta)| d\overline{\varrho}(\beta)]\leq\;e^{-\,2\,\int_0^{\beta}\tau\,(t)\,dt}\;\mathcal{S}_0 [d\Pi| d\overline{\varrho}]\;,
\end{equation}
where ${S}_0 [d\Pi| d\overline{\varrho}]:=\mathcal{S} [d\Pi(\beta)| d\overline{\varrho}(\beta)]_{\beta=0}$, $\tau(\beta)>0$ is a $\beta$--dependent log-Sobolev constant, and where $\left\|\;\;   \right\| _{\rm var}^{2}$ denotes the total variation norm  defined by
\begin{equation}
\left\| d\Pi(\beta) -d\overline{\varrho}(\beta)  \right\|_{\rm var}\doteq
\sup_{\left\| \phi \right\| _{b}\leq 1}\left\{ \left| \int_{\Sigma
}\phi d\Pi(\beta) -\int_{\Sigma }\phi d\overline{\varrho}(\beta)  \right| \right\}\;.
\end{equation}
\\
\noindent Finally,  the dominant energy conditions 
\begin{equation}
\varrho (\beta)\;\geq \;|J(\beta)|\;,\;\;\;\;\;\;\overline{\varrho} (\eta)\;\geq \;|\bar{J}(\eta)|\;,
\end{equation}
hold along the flows $\beta\rightarrow (\varrho (\beta),J(\beta))$, $0\leq\beta\leq\beta^*$ and $\eta\rightarrow (\bar{\varrho} (\eta),\bar{J}(\eta))$, $0\leq\eta\leq\beta^*$.
\label{theorSumm1}
\end{theorem}
\noindent
Let us note that, following a standard notational idiosyncrasy, the relative entropy  $\mathcal{S} [d\Pi(\beta)| d\overline{\varrho}(\beta)]$ is minus the physical relative entropy, thus the above result states that, as expected, $-\,\mathcal{S} [d\Pi(\beta)| d\overline{\varrho}(\beta)]$  is non decreasing along the forward flow  $\beta\rightarrow \varrho (\beta)$. Also, we should stress that we are emphasizing matter rather than volume preservation and consequently  the  fiducial Ricci flow is not volume--normalized.  As follows from the above result, this strategy is a posteriori justified since it provides a good control on the localization properties of the  flow $\beta\mapsto \varrho(\beta)$ with respect to the reference backward flow $\eta\mapsto \overline{\varrho}(\eta)$. The proof of Theorem \ref{theorSumm1} and a number of related properties are discussed at length in Section \ref{matterfieldconj}.\\
\begin{figure}[h]
\includegraphics[scale=0.5]{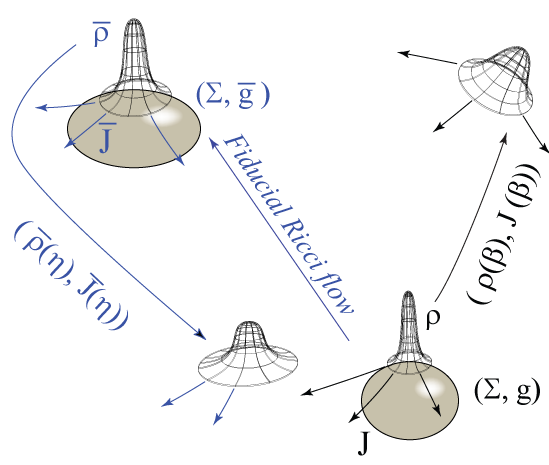}
\caption{Under Ricci flow conjugation we evolve the matter variables $(\varrho, J)$ with respect to the reference distribution $(\overline{\varrho}, \overline{J})$. Parabolic conjugation implies that the associated relative entropy is monotonic and that the dominant energy condition is preserved.}
\label{fig:6}   
\end{figure}
\\
\noindent
A second group of results concerns the general properties of the geometric flows conjugating the extrinsic curvature tensors $K$ and $\overline{K}$. These properties are a direct offspring of the nature of the Lichnerowicz heat equation 
$\bigcirc_d\,{K_{ab}}(\beta)=(\frac{\partial }{\partial \beta }-\Delta_d)\,{K_{ab}}(\beta) =0$  and of its Ricci-flow conjugate $\bigcirc_d^*\,\overline{K}^{ab}(\eta)=(\frac{\partial }{\partial \eta }-\Delta_d +\mathcal{R})\,\overline{K}^{ab}(\eta) =0$. The resulting flows $\beta\mapsto {K_{ab}}(\beta)$ and $\eta\mapsto \overline{K}^{ab}(\eta)$ naturally give rise to a perturbation of the forward and of the backward fiducial Ricci flow
\begin{eqnarray}
\beta &\mapsto& g_{ab}(\beta)\,+\,\epsilon  {K_{ab}}(\beta)\,,\\
\eta &\mapsto& g^{ab}(\beta)\,+\,\epsilon \overline{K}^{\;ab}(\eta)\;,\;\;\;\;\;\;\;\; \epsilon >0\;,\nonumber
\end{eqnarray}
and a basic issue in Ricci flow conjugation is to understand the relation between these perturbed flows and the underlying Ricci flow geometry. This is discussed in general terms in subsection \ref{HDRsubsect} also in relation with the properties of the Berger--Ebin splitting of $\mathcal{T}_{(\Sigma,g(\beta))}\,\mathcal{M}et(\Sigma)$ along the Ricci flow. An important characterization of the reference backward flow $\eta\mapsto \overline{K}^{ab}(\eta)$ is contained in the, (see theorem \ref{nondiss} of the paper, and \cite{carfback}),
\begin{theorem} 
\label{nondissSUM}
Let $\eta \mapsto {\overline{K}}^{\;ab}(\eta )$ be the solution of the Ricci flow conjugate Lichnerowicz heat equation $\bigcirc_d^*\,\overline{K}^{ab}(\eta)=0$ on $\Sigma\times[0,\beta ^{*}]$, then along $\eta\mapsto (g(\eta), {\overline{K}}^{\;ab}(\eta ))$,
\begin{equation}
\label{bello1SUM}
\frac{d}{d\eta }\,\int_{\Sigma }R_{ab}(\eta ){\overline{K}}^{\;ab}(\eta )d\mu _{g(\eta )}=0\;,
\end{equation}

\begin{equation}
\label{bello2SUM}
\frac{d}{d\eta }\,\int_{\Sigma }\left(g_{ab}(\eta)-2\eta\,R_{ab}(\eta )\right){\overline{K}}^{\;ab}(\eta )d\mu _{g(\eta )}=0\;.
\end{equation} 
\end{theorem}
\noindent
To set this result in a proper perspective let us recall that the evolution  of the reference matter density $\eta\mapsto \overline{\varrho }(\eta)$ along the fiducial Ricci flow  is strictly related to   
the Perelman energy functional \cite{18} $\mathcal{F}:\mathcal{M}et(\Sigma)\times C^{\infty }(\Sigma,R)\rightarrow R$ associated with the pair $(g(\beta), \overline{\varrho}(\eta))$ and   defined by
\begin{eqnarray}
\label{SummF-funct}
\mathcal{F}[{g(\eta)},\,\overline{\varrho}(\eta)] &\doteq&  \int_{\Sigma }(\mathcal{R}+|\nabla \ln\overline{\varrho }|^{2})\;
d\overline{\varrho}(\eta)\\
\nonumber\\
&=& \, \int_{\Sigma }(\mathcal{R}+|\nabla f|^{2}){\rm
e}^{-f}\,d\mu _{{g}}= \mathcal{F}[{g(\eta)},\,f(\eta)]\;, \nonumber
\end{eqnarray}
for $g$ evolving along the fiducial Ricci flow $\beta\mapsto g(\beta)$, and  $\eta\mapsto f=-\ln\overline{\varrho}(\eta)$ evolving backward according to 
$\frac{\partial f(\eta )}{\partial \eta}=\,\triangle_{g(\eta)} f-|\nabla f|_{g(\eta)}^2
+\mathcal{R}(\eta)$.  A well--known property of the Perelman functional \cite{18} implies that   $\mathcal{F}[{g};f]$ is non--decreasing along the defining (forward) flows, and we have
\begin{equation}
\frac{d}{d\beta} \mathcal{F}[{g(\beta)},\,\overline{\varrho}(\beta)]= 2\, \int_{\Sigma }\left|\mathcal{R}_{ik}(\beta )-
\nabla_i \nabla_k\,\ln\overline{\varrho}(\beta)\right|^2\,d\overline{\varrho}(\beta)\geq 0\;,
\end{equation}
where $\overline{\varrho}(\beta):=\overline{\varrho}(\eta=\beta^*-\beta)$.
\begin{figure}[h]
\includegraphics[scale=0.4]{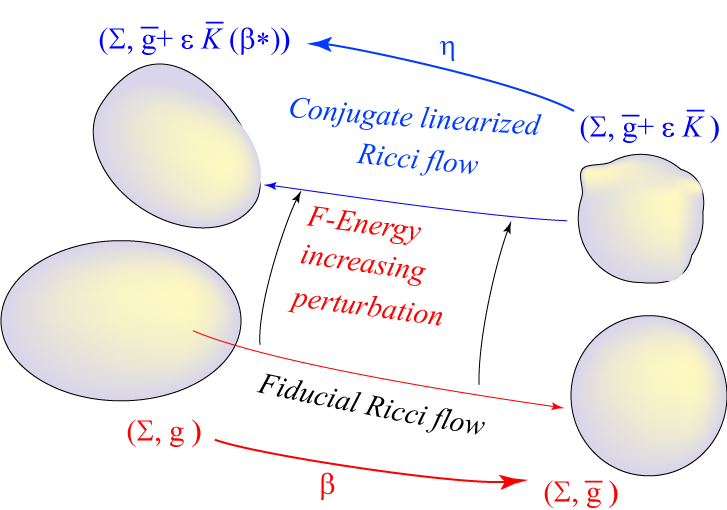}
\caption{The initial  $\overline{K}\in \overline{\mathcal{C}}_{\overline{g}}(\Sigma)$ characterizes the $\mathcal{F}[{g(\eta)},\,\overline{\varrho}(\eta)]$--energy increasing or decreasing nature of the perturbed conjugate flow $\eta\mapsto \overline{K}(\eta)$. Since $\overline{K}$ provides a privileged direction in   $\mathcal{T}\,\mathcal{M}et(\Sigma)$, the linearization
$D\,\mathcal{F}[{g(\eta)},\,\overline{\varrho}(\eta)]\circ \overline{K}$ can be used to define a natural notion of entropic stability of the Ricci flow conjugation between data.}
\label{fig:7}   
\end{figure}
\\
It is natural to discuss how $\mathcal{F}[{g(\eta)},\,\overline{\varrho}(\eta)]$ behaves on the perturbed reference flow 
$\eta \mapsto g^{ab}(\beta)\,+\,\epsilon \overline{K}^{\;ab}(\eta)$, \, as $\epsilon \searrow 0^+$. In particular, one expects that Ricci flow conjugation is a sensible mapping between Einstein data sets if the fiducial Ricci flow interpolating between $(\Sigma ,g)$ and  $(\Sigma ,\overline{g})$ is, in a suitable sense, stable under the perturbation induced by the reference data $\overline{\mathcal{C}}_{\overline{g}}(\Sigma)$ at $\eta=0$. If the fiducial flow is a generalized fixed point of the Ricci flow, \emph{e.g.} a Ricci flat or a shrinking soliton then the problem reduces to the known (second--order)  stability analysis around the given Ricci flow,  (see \emph{e.g.} \cite{Cao},  \cite{guenther}, \cite{sesum2}). More generally, if we interpolate along a generic Ricci flow, we have to consider a form of  first--order entropic stability around the fiducial flow.  We have
\begin{theorem}
\label{capsuleTH}
For $\epsilon >0$ small enough and $0\leq\beta\leq\beta^*$, let $\Omega_{\epsilon }(g(\beta))$
\begin{equation}
\label{epsneigh}
:= \left\{\left. g(\beta)+\, h(\beta)\,\right|\; h\in \mathcal{T}_{(\Sigma,g(\beta))}\,\mathcal{M}et(\Sigma)\;,\;\| h(\beta)\| _{L^{2}(\Sigma,d\mu_{g(\beta)})}<\epsilon    \right\}\;,\nonumber
\end{equation}
denote the (affine) $\epsilon $--tubular neighborhood of the fiducial Ricci flow  $\beta\mapsto g(\beta)$ in $\mathcal{M}et(\Sigma)$. We 
assume that $\beta\mapsto g(\beta)$  is not a Ricci--flat soliton over $\Sigma\times [0,\beta^*]$. If $\overline{K}_{TT}$ is the trace--free and divergence--free part of $\overline{K}\in \overline{\mathcal{C}}_{\overline{g}}(\Sigma)$,\, then the reference flow $\eta\mapsto (g(\eta),\overline{\mathcal{C}}^{\;\sharp }(\eta))$ is  $\mathcal{F}[{g(\eta)},\,\overline{\varrho}(\eta)]$--energy decreasing (increasing) 
in  $\Omega_{\epsilon }(g(\beta))$, i.e. 
\begin{equation}
\left.\frac{d}{d\epsilon }\,\mathcal{F}[{g_{(\epsilon )}(\eta)},\,\overline{\varrho}(\eta)]\right|_{\epsilon =0}\,<\,0\,,\;\;\;\;\; (>0)\;,
\end{equation}
and the Ricci flow conjugation between the two  data sets ${\mathcal{C}}_{{g}}(\Sigma)$ and $\overline{\mathcal{C}}_{\overline{g}}(\Sigma)$ is 
$\mathcal{F}$--stable \;(unstable) in the $\overline{K}$--direction 
if for $\eta=0$ we have
\begin{equation}
\label{entrcond}
\mathfrak{F}(\overline{g},\,\overline{K}):=\int_{\Sigma }\,\left(\overline{\mathcal{R}}_{\;ab}\,\overline{K}_{TT}^{\;ab}+ \frac{1}{n}\,\overline{\mathcal{R}}\,tr_{\overline{g}}\,\overline{K}\right)\,d\mu _{\overline{g}}\,>\,0\,,\;\;(<0)\;.
\end{equation}
\end{theorem}
\noindent 
This theorem, (proved in section \ref{nondissdir}), states that under the \emph{initial} condition (\ref{entrcond}), (it is important to stress that (\ref{entrcond}) is a statement on the reference data $\overline{\mathcal{C}}_{\overline{g}}(\Sigma)$ at $\eta=0$),   the initial data set $\overline{\mathcal{C}}_{\overline{g}}(\Sigma)$ generates, for $\epsilon >0$ small enough, a perturbed  Ricci flow 
$\eta \mapsto g^{ab}(\eta)\,+\,\epsilon \overline{K}^{\;ab}(\eta)$ in  $\Omega_{\epsilon }(g(\beta))$     which is
     $\mathcal{F}[{g(\eta)},\,\overline{\varrho}(\eta)]$--energy decreasing (increasing) with respect to the fiducial (backward) Ricci flow. Thus, if $\mathfrak{F}(\overline{g},\,\overline{K})>0$ we have stability, in $\Omega_{\epsilon }(g(\beta))$, of the  flow under such first--order linear perturbation. Conversely, if  $\mathfrak{F}(\overline{g},\,\overline{K})<0$ the perturbation increases the $\mathcal{F}[{g(\eta)},\,\overline{\varrho}(\eta)]$--energy and the Ricci flow conjugation is  energetically unstable in the reference direction $\overline{K}$. Roughly speaking the direction $\overline{K}$ with respect to which  we compare the forward evolution $\beta\mapsto K_{ab}(\beta)$, (note that $\int_{\Sigma }K_{ab}(\beta)\overline{K}^{\;ab}(\beta)\,d\mu_{g(\beta)}$ is preserved by parabolic conjugation),    generates a perturbed Ricci flow which respect to the fiducial one is energetically more favored and tends to drive the conjugation away from the given $\beta\mapsto g_{ab}(\beta)$.   \\
\\
\noindent 
The actual comparison between the two  data sets ${\mathcal{C}}_{{g}}(\Sigma)$ and $\overline{\mathcal{C}}_{\overline{g}}(\Sigma)$ is realized by exploiting the spectral resolution of the elliptic operator  $ -\Delta _{d}+\mathcal{R}(\bar{g})$  on    $(\Sigma, \overline{g})$, (\emph{i.e.} at $\eta=0$), and a number of  properties which allow to compare  Fourier coefficient along the conjugated flows $\beta\mapsto (g(\beta),{\mathcal{C}}(\beta))$ and $\eta\mapsto (g(\eta),\overline{\mathcal{C}}^{\;\sharp }(\eta))$ solution of 
(\ref{linDT0comp0}) and (\ref{conlinDT0comp0}). We have the
\begin{theorem} (Data comparison)\\
\label{theormode}
\noindent
Along the fiducial backward Ricci flow $\eta\longmapsto g(\eta)$ on $\Sigma\times [0,\beta^*]$,
 let $\eta\mapsto \left\{\Phi ^{\sharp}_{(n)}(\eta)\right\}$\;  denote the flows defined by
\begin{equation}
\bigcirc^*_d\,\Phi ^{\sharp}_{(n)}(\eta)=0\;,\;\;\;\;\Phi ^{\sharp}_{(n)}(\eta=0):=\overline{\Phi} ^{\;\sharp}_{(n)}\;,\;\;\;n\in \mathbb{N}\;,
 \end{equation}
 where  $\left\{\overline{\Phi} ^{\;\sharp}_{(n)},\,\lambda^{(d)} _{(n)}\right\}$ is the discrete spectral resolution of the elliptic operator $ -\Delta _{d}+\mathcal{R}(\bar{g})$  on the reference   $(\Sigma, \overline{g})$, and $\bigcirc^*_d:=\frac{\partial}{\partial\eta}-\Delta_d+\mathcal{R}(g(\eta))$ is the backward Hodge--DeRham-Lichnerowicz heat operator along $\eta\mapsto g(\eta)$.   
If $(\varrho(\beta^*),\,J_i(\beta^*),\,K_{ab}(\beta^*))$ denote  the forward evolution 
of $(\varrho,\,J_i,\,K_{ab})\in \mathcal{C}_{g}(\Sigma )$ along $\beta\rightarrow \mathcal{C}(\beta)$ then we can write
\begin{equation}
\varrho(\beta^*)=\,\sum_n\,\overline{\Phi}^{\;(n)}\,\left[\int_{\Sigma }\,\varrho\;\Phi_{(n)}\,d\mu_{g}\right]_{\beta=0}\;,
\end{equation}
\begin{equation}
J_{a}(\beta^*)=\,\sum_n\,\overline{\Phi} _{a}^{\;(n)}\,\left[\int_{\Sigma }\,J_{i}\;\Phi ^{i}_{(n)}\,d\mu_{g}\right]_{\beta=0}\;,
\end{equation}
\begin{equation}
K_{ab}(\beta^*)=\,\sum_n\,\bar{\Phi} _{ab}^{(n)}\,\left[\int_{\Sigma }\,K_{ij}\,\Phi ^{ij}_{(n)}\,d\mu_{g}\right]_{\beta=0}\;,
\end{equation}
where the integrals appearing on the right hand side are all evaluated at $\beta=0$.
Moreover, under the same hypotheses and notation, we have the mode distribution 
\begin{equation}
\int_{\Sigma }\,\left|\varrho (\beta^*)\right|^2\,d\mu_{\overline{g}}=\sum_n\,\left|\int_{\Sigma }\,\varrho \,\Phi_{(n)}\,d\mu_{g}\right|_{\beta=0}^2\;,
\end{equation}
\begin{equation}
\int_{\Sigma }\,\left|J(\beta^*)\right|^2\,d\mu_{\overline{g}}=\sum_n\,\left|\int_{\Sigma }\,J_{a}\,\Phi ^{a}_{(n)}\,d\mu_{g}\right|_{\beta=0}^2\;,
\end{equation}
\begin{equation}
\int_{\Sigma }\,\left|K(\beta^*)\right|^2\,d\mu_{\overline{g}}=\sum_n\,\left|\int_{\Sigma }\,K_{ij}\,\Phi ^{ij}_{(n)}\,d\mu_{g}\right|_{\beta=0}^2\;.
\end{equation}
\end{theorem}
\begin{figure}[h]
\includegraphics[scale=0.4]{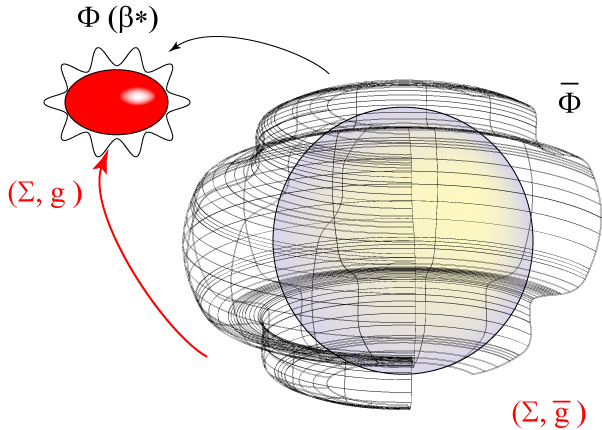}
\caption{The spectral resolution of the elliptic operator  $ -\Delta _{d}+\mathcal{R}(\bar{g})$  on the reference  $(\Sigma, \overline{g})$ gives rise to Fourier coefficients which have remarkable properties under Ricci flow conjugation. These properties generate a mode expansion of the physical fields $(\varrho, J, K)$ with respect to the reference geometry $(\Sigma, \overline{g})$.}
\label{fig:8}   
\end{figure}
\noindent
Such a mode expansion can be applied to the fluctuations of the physical fields $\beta\mapsto (g(\beta),{\mathcal{C}}(\beta))$ with respect to the reference flows $\eta\mapsto (g(\eta),\overline{\mathcal{C}}^{\;\sharp }(\eta))$ so as to get
\begin{theorem} (Fluctuations mode expansion)\\
\label{theorfluc}
\noindent Let 
\begin{eqnarray}
\overline{\varrho}\,&=&\,\sum_n\,\overline{c}_n\left(\overline{\varrho}\right)\;\overline{\Phi}^{(n)}\;,\\
\overline{J}_{\;a}\,&=&\,\sum_n\,\overline{c}_n\left(\overline{J}\,\right)\;\overline{\Phi}_{\;a}^{(n)}\;, \nonumber\\
\overline{K}_{\;ab}\,&=&\,\sum_n\,\overline{c}_n\left(\overline{K}\,\right)\;\overline{\Phi}_{\;ab}^{(n)}\;,\nonumber 
\end{eqnarray}
the mode expansion on $(\Sigma,\,\overline{g})$ of the reference data $\in \overline{\mathcal{C}}_{\overline{g}}\;(\Sigma)$.
Then, if we define
\begin{equation}
\delta \,\varrho_{(n)}:=\left[\int_{\Sigma }\,\varrho\;\Phi_{(n)}\,d\mu_{g}\right]_{\beta=0}-\overline{c}_n\left(\overline{\varrho}\right)\;,
\end{equation}
\begin{equation}
\delta \,J_{(n)}:=\left[\int_{\Sigma }\,J_{i}\;\Phi ^{i}_{(n)}\,d\mu_{g}\right]_{\beta=0}-\overline{c}_n\left(\overline{J}\,\right)\;,
\end{equation}
\begin{equation}
\delta \,K_{(n)}:=\left[\int_{\Sigma }\,K_{ij}\,\Phi ^{ij}_{(n)}\,d\mu_{g}\right]_{\beta=0}-\overline{c}_n\left(\overline{K}\,\right)\;,
\end{equation}
we can write the $\beta$--evolved data $(\varrho(\beta^*),\,J_i(\beta^*),\,K_{ab}(\beta^*))$ as
\begin{equation}
\varrho(\beta^*)=\overline{\varrho}\,+\,\sum_n\,\overline{\Phi}^{\;(n)}\,\delta \,\varrho_{(n)}\;,
\end{equation}
\begin{equation}
J_{a}(\beta^*)=\overline{J}_{\;a}\,+\,\sum_n\,\overline{\Phi} _{a}^{\;(n)}\,\delta \,J_{(n)}\;,
\end{equation}
\begin{equation}
K_{ab}(\beta^*)=\overline{K}_{\;ab}\,+\,\sum_n\,\overline{\Phi} _{ab}^{\;(n)}\,\delta \,K_{(n)}\;.
\end{equation}
\end{theorem}
\noindent
Both Theorem \ref{theormode}  and Theorem \ref{theorfluc} follows from the results of section \ref{conjmodexp}. They can be quite effective when the reference data $\overline{\mathcal{C}}_{\overline{g}}(\Sigma)$ are supported on a manifold of large symmetry \emph{e.g.} a round 3--sphere for which one can have a rather explicit control on the spectral resolution of $ -\Delta _{d}+\mathcal{R}(\bar{g})$ in terms of scalar, vector, and tensor harmonics. In general, this mode expansion and the properties of the associated Ricci flow conjugation suggest that there is an underlying heat--kernel representation governing Ricci flow conjugation. This also implies that, at least for small $\eta$, Ricci flow conjugation is indeed a form of parabolic averaging of the physical data ${\mathcal{C}}_{{g}}(\Sigma)$ with respect to the reference   $\overline{\mathcal{C}}_{\overline{g}}(\Sigma)$. Indeed we have the
\begin{theorem} (Heat kernel representation of the fluctuations)\\
\label{heatrepr}
The  $\beta$--evolved data $(\varrho(\beta^*),\,J_i(\beta^*),\,K_{ab}(\beta^*))$ admit a heat kernel \footnote{See theorem \ref{HeatKernTh} for the definition of the tensorial heat kernel $\mathbb{H}(y,x;\,\eta)$ and the associated notation.} representation in terms of the heat kernel
$\mathbb{H}(y,x;\,\eta)$ 
 of the backward conjugated operator     $\bigcirc^*_d:=\frac{\partial}{\partial\eta}-\Delta _d+\mathcal{R}(g(\eta))$. In particular, the  fluctuations of the data $(\varrho(\beta^*),\,J_i(\beta^*),\,K_{ab}(\beta^*))$ with respect to the reference data $\in \overline{\mathcal{C}}_{\overline{g}}\;(\Sigma)$, admit a computable asymptotic expansion for small $\eta$.  For instance, in the case of the matter density,  we can write (see section 4, for notation) 
\begin{eqnarray}
&&\varrho(\beta^*,y)=\overline{\varrho}(y)+\sum_n\,\overline{\Phi}^{\;(n)}\,\delta \,\varrho_{(n)}\,\\
\nonumber\\
&&=\overline{\varrho}(y)\,+\,\frac{1}{\left(4\pi \,\eta \right)^{\frac{3}{2}}}\,\int_{\Sigma }\exp\left(-\frac{d^{2}_{0}(y,x)}{4\eta } \right)
\,\left[{\varrho }(x,\eta )-\overline{\varrho }(y)\right]\,d\mu _{g(x,\eta )}\nonumber\\
\nonumber\\
&&+\sum_{h=1}^{N}\frac{\eta ^{h}}{\left(4\pi \,\eta \right)^{\frac{3}{2}}}\,\int_{\Sigma }e^{\left(-\frac{d^{2}_{0}(y,x)}{4\eta } \right)}\,
{\Upsilon [h] }(y,x;\eta )\,\left[{\varrho }(x,\eta )-\overline{\varrho }(y)\right]\,d\mu _{g(x,\eta )}\nonumber\\
\nonumber\\
&&+O\left(\eta ^{N-\frac{1}{2}} \right)\nonumber\;,
\end{eqnarray}
where ${\Upsilon [h] }(y,x;\eta )$ are smooth coefficients, depending on the geometry of $(\Sigma ,{g}(\eta))$, characterizing the asymptotics of the heat kernel of $\bigcirc^*_d$.
\end{theorem}
\noindent

This asymptotics takes  a more explicit form when applied to the evaluation of  integral quantities such as 
$\int_{\Sigma }\,\varrho(\beta^*)\,d\,\overline{\varrho}(\eta)$. Note that whereas $\int_{\Sigma }\,\varrho(\beta^*-\eta)\,d\,\overline{\varrho}(\eta)$ is a conserved quantity along the interpolating Ricci flow, the above integral is not. It provides, as $\eta$ varies, the matter content of the Ricci evolved $\varrho(\beta^*)$  with respect to the given reference flow $\overline{\varrho}(\eta)$  and
$\int_{\Sigma }\,[\varrho(\beta^*)-\varrho(\beta^*-\eta)]\,d\,\overline{\varrho}(\eta)$
is a relevant physical quantity which can be used to describe the (small $\eta$) fluctuations of $\varrho (\beta)$ associated with Ricci flow conjugation. Making a parallel with heat propagation, $\int_{\Sigma }\,\varrho(\beta^*)\,d\,\overline{\varrho}(\eta)$ plays the role of the heat content of a system characterized by a  distribution given by $\varrho(\beta^*-\eta)$ and by an $\eta$--dependent \emph{specific heat} proportional to $\overline{\varrho}(\eta)$. We show that  as $\eta \searrow 0^+$ we have the asymptotic expansion
\begin{eqnarray}
&&\int_{\Sigma }\,\varrho(\beta^*)\,\overline{\varrho}(\eta) \,d\mu_{{g}(\eta)}\,=
\int_{\Sigma }\,\varrho(\beta^*)\,\overline{\varrho} \,d\mu_{\overline{g}}\,\\
&&-\,\eta\,\int_{\Sigma }\,\overline{\varrho} \,\overline{\Delta }\,\varrho(\beta^*)\,d\mu_{\overline{g}}
\,+\, \eta^4\, \int_{\Sigma }\,\overline{\mathcal{R}}^{\;ab}\,\overline{\nabla}_a\overline{\nabla}_b\,\varrho(\beta^*)\,\overline{\varrho} \,d\mu_{\overline{g}}\,+\ldots\;.\nonumber
\end{eqnarray}
Similar expansions can be written down for the current content  and the extrinsic curvature content, (see Theorem \ref{Giltheo}).
The proof of these results are discussed in detail in section \ref{asympsec}, and they provide detailed evidence of the non--trivial interaction between Ricci flow and Einstein initial data sets.\\
\\
\noindent
The paper is organized as follows. Section \ref{sezioneDue} introduces Ricci flow conjugation starting with a brief (mainly notational) summary on the Berger-Ebin decomposition of the tangent space to the space of Riemannian metrics and the associated notion of  affine slice. The core of this section is a technical lemma providing the various commutation rules between the Hodge--DeRham--Lichnerowicz heat operator and its Ricci flow conjugate. Some of these commutation rules are well--known whereas other are new and interesting in their own right. Section \ref{conjmodexp} describes a natural decomposition in modes associated with the interaction between Ricci flow conjugation and the  spectral resolution of the elliptic operator $-\Delta _{d}+\mathcal{R}(\overline{g})$. This mode decomposition is then applied to the Ricci conjugated flows   allowing for a comparison between the physical 
data set  ${\mathcal{C}}_{{g}}\;(\Sigma)$ and the reference data $\overline{\mathcal{C}}_{\overline{g}}\;(\Sigma)$ along the lines described in the introduction. In Section \ref{asympsec} we discuss the heat kernels for the Ricci conjugated flows and their small $\eta$ asymptotics. 
Section \ref{matterfieldconj} deals with the Ricci flow conjugation between matter fields. Here we introduce the relevant entropic quantities, discuss the convergence to equilibrium and the associated localization with respect to  the reference data set 
$\overline{\mathcal{C}}_{\overline{g}}\;(\Sigma)$, and prove the preservation of the dominant energy condition. In Section \ref{sezioneSei} we discuss Ricci flow conjugation for the extrinsic curvature flow. Finally in Section  \ref{nondissdir} we analize the role of Perelman $\mathcal{F}$--energy in characterizing entropically the Ricci flow perturbations associated with  the conjugate flow 
$\eta\mapsto \overline{K}^{\;ab}(\eta)$. A few concluding remarks are presented in Section \ref{sezioneOtto}.

\section{Ricci flow conjugation}
\label{sezioneDue}
To set notation, let $\Sigma $ be a $C^{\infty }$ compact\footnote{In the non-compact case our understanding of Ricci flow conjugation is much more limited due to the subtle issue of the appropriate boundary conditions to adopt in dealing with the interpolating flows.} $n$--dimensional manifold, ($n\geq3$), without boundary, and let  $\mathcal{D}iff(\Sigma )$ and $\mathcal{M}et(\Sigma )$ respectively be the group of smooth diffeomorphisms and the open convex cone of all smooth Riemannian metrics over $\Sigma$. For any $g\in \mathcal{M}et(\Sigma )$,  we denote by $\nabla$ the Levi--Civita connection of $g$, and let $\mathcal{R}m(g)=\mathcal{R}^{i}_{klm}\,\partial _i\otimes dx^k\otimes dx^l\otimes dx^m$, $\mathcal{R}ic(g)=\mathcal{R}_{ab}\,dx^a\otimes dx^b$ and $\mathcal{R}(g)$ be the corresponding Riemann,  Ricci and  scalar curvature operators, respectively. The space of smooth $(p,q)$--tensor fields on $\Sigma $,\;
${C}^{\infty }(\Sigma,\otimes ^{p}\, T^{*}\Sigma\otimes ^{q}T\Sigma )$ is  endowed with the pre--Hilbertian  $L^{2}$ inner product  $(U,V)_{L^{2}(\Sigma,d\mu_{g} )}\doteq \int_{\Sigma }\,\langle U,\,V\rangle_g d\mu _{g}$, where $\langle U,\,V \rangle_g$ is the pointwise  $g$--metric in $\otimes ^{p}\, T^{*}\Sigma\otimes ^{q}T\Sigma $. We let $||U||^{2}_{L^{2}}$ and $||U||^{2}_{\mathcal{H}^{s}}\doteq ||U||^{2}_{L^{2}}+\sum_{i=1}^{s}||\nabla ^{(i)}U||^{2}_{L^{2}}$, $s\geq0$, be the corresponding $L^{2}$ and Sobolev norms. The completions of ${C}^{\infty }(\Sigma,\otimes ^{p}\, T^{*}\Sigma\otimes ^{q}T\Sigma )$ in these norms, define the corresponding space of square summable and Sobolev sections $L^{2}(\Sigma,\otimes ^{p}\, T^{*}\Sigma\otimes ^{q}T\Sigma )$ and $\mathcal{H}^{s}(\Sigma,\otimes ^{p}\, T^{*}\Sigma\otimes ^{q}T\Sigma )$, respectively.
In such a setting, the tangent space to $\mathcal{M}et(\Sigma )$ at $(\Sigma,g )$,\, $\mathcal{T}_{(\Sigma ,g)}\mathcal{M}et(\Sigma )$, is identified with the space of smooth symmetric bilinear forms ${C}^{\infty }(\Sigma,\otimes ^{2}_{S}\, T^{*}\Sigma)$ over $\Sigma $, and we shall consider the Riemannian metrics $\mathcal{M}et{\;}^s(\Sigma )$ of Sobolev class $s>\frac{n}{2}$ as an open subset of $\mathcal{H}^{s}(\Sigma,\otimes ^{2}\, T^{*}\Sigma)$. 
The averaging properties of  Ricci flow conjugation between distinct initial data sets depend on the interaction between the linearized and conjugate--linearized Ricci flow and the Berger--Ebin splitting of the space of symmetric bilinear forms. Thus, to place the arguments to follow in a natural context, we start recalling a few basic properties of such a decomposition.

\subsection{Remarks on the affine Berger--Ebin slice theorem}
The (non--linear) space $\mathcal{M}et{\;}^s(\Sigma )$  is acted upon by the (topological) group, $\mathcal{D}iff^{\;s'}(\Sigma )$, defined by the set of diffeomorphisms which, as maps $\Sigma \rightarrow \Sigma $ are an open subset of the Sobolev space of maps  $\mathcal{H}^{\,s'}(\Sigma ,\Sigma )$, with $s'\geq s+1$. In particular, there is a natural projection map $\pi :\mathcal{D}iff^{\;s'}(\Sigma )\rightarrow \mathcal{O}_{g}^{\,s}$, $\pi (\phi )\doteq \phi^*\,g$, where $\mathcal{O}_{g}^{\,s}$ is the $\mathcal{D}iff^{\;s'}(\Sigma )$--orbit of a given metric $g\in\mathcal{M}et{\;}^s(\Sigma )$, and $\phi^*\,g$ is the pull--back under $\phi\in \mathcal{D}iff^{\;s'}(\Sigma )$.  If  $\mathcal{T}_{(\Sigma ,g)}\mathcal{O}_{g}^{\,s}$ denotes the tangent space  to any such an orbit, then   $\mathcal{T}_{(\Sigma ,g)}\mathcal{O}_{g}^{\,s}$ is the image of the injective operator with closed range 
\begin{eqnarray}
\delta_{g} ^{*}: \mathcal{H}^{s+1}(\Sigma,T\Sigma )&\rightarrow& \mathcal{H}^{\,s}(\Sigma,\otimes ^{2}T^{*}\Sigma )\\
w\;\;\;\;\;&\mapsto& \delta_{g}^{*}\,(w)\doteq \frac{1}{2}\,\mathcal{L}_{w}\,g\;,\nonumber
\end{eqnarray}
where  $\mathcal{L}_{w}$ denotes the Lie derivative along the vector field $w$. Standard elliptic theory implies that the $L^{2}$--orthogonal subspace to $Im\;\delta_{g} ^{*}$ in $\mathcal{T}_{(\Sigma ,g)}\mathcal{M}et^{\;s}(\Sigma )$ is spanned by the ($\infty $--dim) kernel of the $L^{2}$ adjoint $\delta_{g}$ of  $\delta_{g} ^{*}$, 
\begin{eqnarray}
\delta_{g}: \mathcal{H}^{\,s }(\Sigma,\otimes ^{2}T^{*}\,\Sigma )&\rightarrow&  \mathcal{H}^{s-1 }(\Sigma,T^{*}\,\Sigma )\\
h\;\;\;\;\;\;\;&\mapsto& \delta_{g} \,h\doteq -\,g^{ij}\,\nabla _{i}h_{jk}\,dx^{k}\;.\nonumber
\end{eqnarray}
This entails the well--known Berger--Ebin ${L^{2}(\Sigma,d\mu_{g} )}$--orthogonal splitting \cite{ebin, cantor} of   the tangent space $\mathcal{T}_{(\Sigma ,g)}\mathcal{M}et^{\,s}(\Sigma )$, 
\begin{equation}
\mathcal{T}_{(\Sigma ,g)}\mathcal{M}et^{\,s}(\Sigma )\cong  \left[\mathcal{T}_{(\Sigma ,g)}\mathcal{M}et^{\,s}(\Sigma )\cap Ker\,\delta_{g} \right]\,\oplus Im\,\delta_{g} ^{*}\left[\mathcal{H}^{s+1}(\Sigma,T\Sigma ) \right]\,,
\label{L0split}
\end{equation}

\noindent according to which, for any given tensor  $h\in \mathcal{T}_{(\Sigma ,g)}\mathcal{M}et^{\,s}(\Sigma )$, we can write $h_{ab}=h_{ab}^T+\mathcal{L}_{w}\,g_{ab}$ where $h_{ab}^T$ denotes the $div$--free part of $h$, ($\nabla ^a\,h_{ab}^T=0$), and where the vector field ${w}$ is characterized as the solution, (unique up to the Killing vectors of $(\Sigma ,g)$), of the elliptic PDE
 $\delta_{g}\,\delta_{g}^*\,w=\delta_{g}\,h$.\\
 \begin{figure}[h]
\includegraphics[scale=0.3]{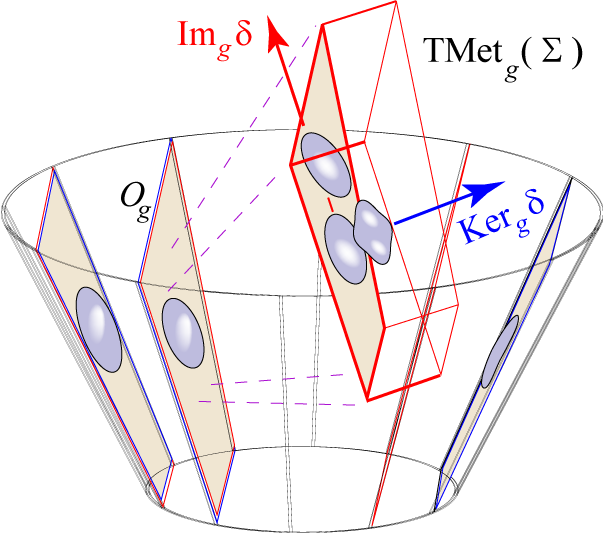}
\caption{The Berger--Ebin decomposition of $\mathcal{T}_{(\Sigma ,g)}\mathcal{M}et^{\,s}(\Sigma )$.}
\label{fig: 9}   
\end{figure}
\\
Let us consider the subset of metrics $\mathcal{M}et^{\,s+1}(\Sigma )\subset \mathcal{M}et^{\,s}(\Sigma )$, and let
\begin{equation}
\mathcal{B}_{\rho }^{s+1}(g)\doteq \left\{\left.h^T\in \mathcal{T}_{(\Sigma ,g)}\mathcal{M}et^{\,s+1}(\Sigma )\cap Ker\;\delta_{g} \,\right|\;\| h^T\| _{L^{2}}<\rho   \right\}\;,
\label{ballmetr}
\end{equation}
be the open ball of radius $\rho $,   $L^{2}$--orthogonal to $Im\;\delta_{g} ^{*}\left(\mathcal{H}^{s+2}(T\Sigma ) \right)$.
According to the Ebin--Palais slice theorem \cite{ebin}, (for a fine survey on slice theorems see \cite{ABesse} and \cite{IseMarsd}), $\mathcal{B}_{\rho }^{s+1}(g)$ 
exponentiates, via the flow induced by $\mathcal{H}^{s+1}$ vector fields \cite{ebinmarsden}, into a submanifold $\mathcal{S}_g^{s+1}$ of $\mathcal{M}et^{s+1}(\Sigma )$  providing a slice to the action of $\mathcal{D}iff^{s+2}(\Sigma )$. Since $\mathcal{M}et^{\,s+1}(\Sigma )$ is an open set in $\mathcal{H}^{\,s+1}(\Sigma,\otimes ^{2}\, T^{*}\Sigma)$,  instead of a local slice obtained by exponentiation, here we shall use the affine slice defined,  for a small enough $\rho$, by the ball $\mathcal{B}_{\rho }^{s+1}(g)$ itself, (affine slice construction is described in \cite{bleecker,IseMarsd}). Explicitly, if we identify 
\begin{equation}
\mathcal{S}_g^{s+1}\simeq \left\{ g\,+\, \mathcal{B}_{\rho }^{s+1}(g)\right \} \subset \mathcal{H}^{\,s+1}(\Sigma,\otimes ^{2}\, T^{*}\Sigma)\;
\label{affslice}
\end{equation}
then, there is a neighborhood $\mathcal{U}_g^{s+1}$ of $g\in\mathcal{O}_g^{\,s+1}(\Sigma )$ and a section $\chi :\mathcal{U}_g^{s+1}\rightarrow \mathcal{D}iff^{s+1}(\Sigma )$, $g'\mapsto \chi (g')\in \mathcal{D}iff^{s+1}(\Sigma )$, with $\pi \circ \chi =id$, such that the map
\begin{eqnarray}
\Upsilon:  \mathcal{U}_g^{s+1}\times \mathcal{S}_g^{s+1}&\longrightarrow& \mathcal{M}et^{s+1}(\Sigma )\;,\label{slicethe}\\
(g',g+h^{\,T})\;\;\;&\longmapsto& \Upsilon (g',g+h^{\,T})\doteq \chi (g')^*\,(g+h^{\,T})\;,\nonumber
\end{eqnarray}
is a local homeomorphism  onto a  neighborhood of $g$ in $\mathcal{M}et^{\,s+1}(\Sigma )$.  Moreover, if $I(\Sigma ,g)\subset \mathcal{M}et^{s+1}(\Sigma )$ denotes the isometry group of $(\Sigma ,g)$ and $\eta \in I(\Sigma ,g)$,  then $\eta^*\,\mathcal{S}_g^{s+1}=\mathcal{S}_g^{s+1}$. Conversely, if $\eta\in \mathcal{D}iff^{s+1}(\Sigma )$ and $\eta^*\,\mathcal{S}_g^{s+1}\cap \mathcal{S}_g^{s+1}\not=\emptyset$, then $\eta \in I(\Sigma ,g)$. This affine version \cite{bleecker,IseMarsd} of Ebin--Palais slice theorem allows to locally parametrize  $\mathcal{M}et^{\,s+1}(\Sigma )$, in a neighborhood of a given $(\Sigma, g)$, by means of the diffeomorphisms
$\varphi \in \mathcal{D}iff^{s+1}(\Sigma )$ defined by the cross section  $\chi (g')=\varphi ^*\,g$ and of the divergence free tensor fields $h^{\,T}$ in the slice $\mathcal{S}_g^{s+1}$.
\begin{remark}
 Regularity arguments show that the existence of the  slice map (\ref{slicethe}) can be extended to 
$\mathcal{M}et(\Sigma )$, (obtained  as the (inverse) limit space $\{\mathcal{M}et^{\,\,s+1}(\Sigma )\}_{s\rightarrow \infty }$), and henceforth we shall confine our analysis to the smooth case. 
\end{remark}
\begin{figure}[h]
\includegraphics[scale=0.4]{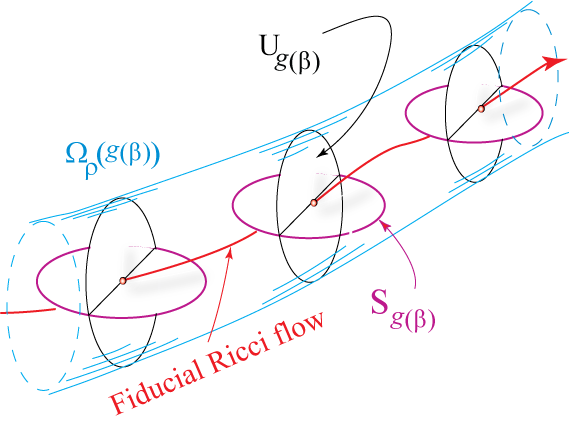}
\caption{The affine Berger--Ebin slice theorem along the fiducial Ricci flow.}
\label{fig:10}   
\end{figure}
\noindent
Let $\beta \rightarrow g_{ab}(\beta)$ be a fiducial Ricci flow of bounded geometry on $\Sigma\times [0,\beta^*]$ in the sense of 
Definition \ref{Fidudefn}. 
The hypothesis of bounded geometry implies that we can apply the Berger--Ebin splitting (\ref{L0split}) along the fiducial Ricci flow. In particular we have the following induced affine slice parametrization in a neighborhood of the given Ricci flow 
\begin{lemma}
Let $\beta\mapsto g(\beta)$ be a fiducial  Ricci flow $\beta\mapsto g(\beta)$ of bounded geometry on $\Sigma\times[0,\beta^*]$, then
there exists an affine slice parametrization of a tubular neighborhood, $\Omega _{\rho}(g(\beta))$, 
of $\beta \rightarrow g_{ab}(\beta)$ such that 
\begin{equation}
\Omega _{\rho}(g(\beta))\doteq \left(\mathcal{U}_{g(\beta)}\times \mathcal{S}_{g(\beta)}\right)\times [0,\beta^*]\;,
\label{fidpar}
\end{equation}
where, for each given $\beta\in[0,\beta^*]$, $\mathcal{U}_{g(\beta)}\subset \mathcal{O}_{g(\beta)}$ is an open neighborhood of the Ricci flow metric $g(\beta)$ in the  $\mathcal{D}iff(\Sigma )$--orbit $\mathcal{O}_{g(\beta)}$, and 
\begin{equation}
\mathcal{S}_{g(\beta)}\doteq \left\{\left. g(\beta)+h^{\,T}\,\right|\; h^T\in  Ker\;\delta_{g(\beta)}\;,\;\| h^T\| _{L^{2}(\Sigma,d\mu_{g(\beta)})}<\rho   \right\}\;,
\label{}
\end{equation}
is, for $\rho>0$ independent from $\beta$ and small enough, the associated affine slice through $g(\beta)$.
\end{lemma}
\begin{proof}
Along $\beta\mapsto g(\beta)$, $0\leq \beta\leq \beta^*$, the Berger--Ebin decomposition
\begin{equation}
T_{g(\beta)}\mathcal{M}et(\Sigma )\cong Ker\;\delta _{g(\beta)} \,\oplus Im\;\delta^{*}_{g(\beta)}
\end{equation}
is well--defined since the fiducial Ricci flow is of bounded geometry. At each given $\beta$, the corresponding affine slice is provided by,
(see (\ref{affslice})),
\begin{equation}
\widetilde{\mathcal{S}}_{g(\beta)}\simeq \left\{ g(\beta)\,+\, \widetilde{\mathcal{B}}_{\rho(\beta) }(g(\beta))\right \}\;,
\label{}
\end{equation}
where the open ball $\widetilde{\mathcal{B}}_{\rho(\beta)}(g(\beta))$ of divergence--free tensor in $\mathcal{T}_{(\Sigma ,g(\beta))}\mathcal{M}et(\Sigma )$ is defined according to
\begin{eqnarray}
\widetilde{\mathcal{B}}_{\rho(\beta)}(g(\beta))
\doteq \left\{\left.h^T(\beta)\in  Ker\;\delta_{g(\beta)} \,\right|\;\| h^T(\beta)\| _{L^{2}}<\rho(\beta)   \right\}\;.
\label{tildeB}
\end{eqnarray}
The hypothesis of bounded geometry implies that, for $0\leq\beta\leq\beta^*$, the set of $\{\rho(\beta)\}$ is uniformly bounded away from zero by some positive constant $\rho:=\inf_{0\leq\beta\leq\beta^*}\,\{\rho(\beta)\}>0$. We correspondingly define 
\begin{eqnarray}
{\mathcal{B}}_{\rho}(g(\beta))
\doteq \left\{\left.h^T(\beta)\in  Ker\;\delta_{g(\beta)} \,\right|\;\| h^T(\beta)\| _{L^{2}}<\rho   \right\}\;.
\label{B}
\end{eqnarray}
and set 
\begin{equation}
{\mathcal{S}}_{g(\beta)}\simeq \left\{ g(\beta)\,+\,{\mathcal{B}}_{\rho }(g(\beta))\right \}\;.
\label{}
\end{equation}
Since ${\mathcal{S}}_{g(\beta)}\subseteq \widetilde{\mathcal{S}}_{g(\beta)}$,  the slice map (\ref{slicethe})) associated to $\widetilde{\mathcal{S}}_{g(\beta)}$ restricts naturally to ${\mathcal{S}}_{g(\beta)}$, to the effect that for each $\beta\in [0,\beta^*]$ there is a neighborhood $\mathcal{U}_{g(\beta)}$ of $g(\beta)\in\mathcal{O}_{g(\beta)}(\Sigma )$ and a section $\chi_{\beta} :\mathcal{U}_{g(\beta)}\rightarrow \mathcal{D}iff(\Sigma )$, $g'\mapsto \chi_{\beta} (g')\in \mathcal{D}iff(\Sigma )$, such that the map $\Upsilon_{\beta}:  \mathcal{U}_{g(\beta)}\times \mathcal{S}_{g(\beta)}\longrightarrow \mathcal{M}et(\Sigma )$
\begin{equation}
\Upsilon_{\beta} (g',g(\beta)+h^{\,T}(\beta))\doteq \chi_{\beta} (g')^*\,(g(\beta)+h^{\,T}(\beta))\;,
\label{slicebeta}
\end{equation}
is a local homeomorphism  onto a  neighborhood of $g(\beta)$ in $\mathcal{M}et(\Sigma )$. 
\end{proof}

\subsection{The Hodge--DeRham--Lichnerowicz heat operator}
\label{HDRsubsect}

 The Ricci flow interacts with the slice map $\Upsilon_{\beta}$ in a rather sophisticated way: a perturbation of the Ricci flow which propagates an $h\in \mathcal{T}_{(\Sigma ,g(\beta=0))}\mathcal{M}et(\Sigma )\cap \mathcal{U}_{g(\beta=0)}$ will give rise to a perturbed Ricci flow evolution in $\mathcal{U}_{g(\beta)}$, whereas a perturbation propagating an $h\in \mathcal{T}_{(\Sigma ,g(\beta=0))}\mathcal{M}et(\Sigma )\cap \mathcal{S}_{g(\beta)}$ in general fails to evolve in $\mathcal{S}_{g(\beta)}$. Naively, this can be attributed to the dissipative (weakly--parabolic) nature of the Ricci flow, however  the underlying rationale is quite subtler and holds a few surprises.
\noindent
To discuss this point,
let $\epsilon \mapsto g^{(\epsilon )}_{ab}(\beta )$, $0\leq\epsilon \leq 1$, be a smooth one--parameter family of Ricci flows in the 
tubular neighborhood $\Omega _{\rho }(g(\beta))$ defined above.  For $\epsilon \searrow 0$, this set $\{g^{(\epsilon )}_{ab}(\beta )\}$ is locally characterized by the tangent vector  $h_{ab}(\beta)$ in $T_{g(\beta)}\mathcal{M}et(\Sigma )$, covering the fiducial curve $\beta \rightarrow g_{ab}(\beta)$, $0\leq \beta  \leq \beta^*$, and defined by the first jet $h_{ab}(\beta)\doteq \frac{d}{d\epsilon }g^{(\epsilon )}_{ab}(\beta )|_{\epsilon =0}$ of
 $g^{(\epsilon )}_{ab}(\beta )$. Any such $h_{ab}(\beta)$ satisfies the linearized Ricci flow equation
\begin{equation} 
\begin{tabular}{l}
$\frac{\partial }{\partial \beta }\,h_{ab}(\beta)=-2\,\frac{d}{d\epsilon }\mathcal{R}^{(\epsilon )}_{ab}(\beta )|_{\epsilon =0}\,\doteq -2\,D\,\mathcal{R}ic(g(\beta))\circ \,h_{ab}(\beta)$\, \\ 
\\ 
$h_{ab}(\beta =0)=h_{ab}$\, ,\;\; $0\leq \beta \leq \beta^*$\;.%
\end{tabular}
   \label{mflowlin}
\end{equation}
\begin{figure}[h]
\includegraphics[scale=0.3]{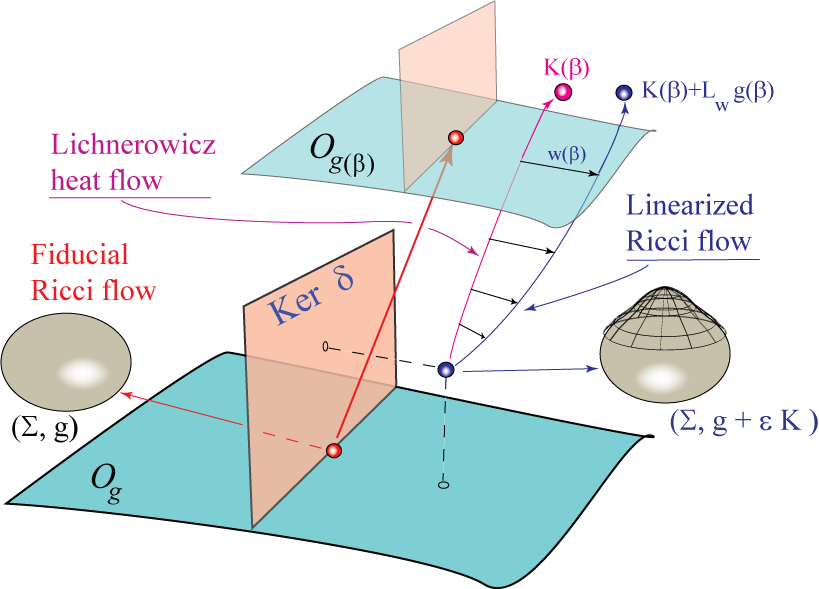}
\caption{The geometry of the linearized Ricci flow is related to the observation that if $K(\beta)$ is a solution of the Lichnerowicz heat equation, then
$\widetilde{K}(\beta):=K(\beta)+\mathcal{L}_{w(\beta)}g(\beta)$ is a solution of the linearized Ricci flow (with the same initial datum $K$) as long as the covector field is such that $\frac{\partial }{\partial \beta}\,w(\beta)=\delta_{g(\beta)} (K(\beta)-\frac{1}{2} tr\,K(\beta)\,K(\beta))$, with $w(\beta=0)=0$.}
\label{fig: 11}   
\end{figure}
As is well-known, this linearization is not parabolic due to the equivariance of the Ricci flow under diffeomorphisms. However, there is a natural choice \cite{01}, (see also Chap.2 of \cite{chowluni}), for fixing the action of the diffeomorphism group $\mathcal{D}iff(\Sigma )$, and (\ref{mflowlin}) takes the form of  the dynamical system $\beta \mapsto {h}_{ab}(\beta )\in T_{g(\beta)}\mathcal{M}et(\Sigma )$ defined, along the fiducial Ricci flow $\beta \mapsto g(\beta)$, by  the Lichnerowicz heat equation 
\begin{equation} 
\begin{tabular}{l}
$\bigcirc _{L}\,{h}_{ab}(\beta  )\doteq \left(\frac{\partial  }{\partial  \beta   }-\Delta _{L}\right)\,{h}_{ab}(\beta  )=0\;,$ \\ 
\\ 
${h}_{ab}(\beta  =0)=\,{h}_{ab}$\, ,\;\; $0\leq \beta  \leq \beta^*$\;,%
\end{tabular}
  \label{divfree}
\end{equation}
where $\Delta _{L}\,:$ $C^{\infty }(\Sigma,\otimes ^{2}T^{*}\,\Sigma)\rightarrow C^{\infty }(\Sigma,\otimes ^{2}T^{*}\,\Sigma)$ is the Lichnerowicz-DeRham Laplacian \cite{lichnerowicz}  on symmetric bilinear
forms defined, (with respect to $g_{ab}(\beta )$), by 
\begin{equation}
\Delta _{L}{h}_{ab}\doteq \triangle {h}%
_{ab}-R_{as}{h}_{b}^{s}-R_{bs}{h}_{a}^{s}+2R_{asbt}{h%
}^{st},
\label{LDR}
\end{equation}
$\triangle \doteq g^{ab}(\beta )\,\nabla _{a}\,\nabla _{b}$ denoting the rough  Laplacian. 
Henceforth, when discussing the linearized Ricci flow  (\ref{mflowlin}) we will explicitly refer to the gauge reduced version (\ref{divfree}).
\begin{figure}[h]
\includegraphics[scale=0.3]{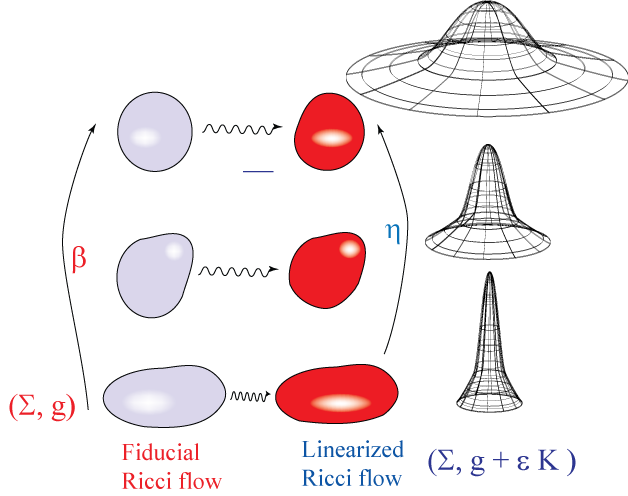}
\caption{The Hodge--DeRham--Lichnerowicz representation of the linearized Ricci flow is a geometrical heat equation which unifies the scalar, vector, and tensor heat equations which we often use to analyze specific properties of the Ricci flow.}
\label{fig: 12}   
\end{figure}
\begin{remark} \emph{(The Hodge--DeRham--Lichnerowicz heat operator)}.
 It is worthwhile recalling that the elliptic operators defined, along the fiducial Ricci flow $\beta \rightarrow g_{ab}(\beta)$, by: \emph{(i)} The standard Laplacian acting on scalar functions $\triangle$; \emph{(ii)}  The vector Laplacian acting on (co)vector fields, $\Delta _{ vec}\doteq \Delta-\mathcal{R}ic$; and \emph{(iii)} The Lichnerowicz--DeRham laplacian acting on symmetric bilinear forms $\triangle_{L}$, can all be formally identified\footnote{The fact that formally the Hodge laplacian $\triangle_{d}$ acts on 2--forms in the same way that the Lichnerowicz--DeRham laplacian acts on symmetric 2--tensors is a well--known property of $\Delta _L$--see \emph{e.g.} \cite{chowluni}.}  with the $g(\beta)$--Hodge--DeRham laplacian acting on $p$--differential forms 
\begin{equation}
\triangle _{d}\doteq -\,(d\,\delta_{g(\beta)}+\delta_{g(\beta)}\,d)\;.
\end{equation}
Thus,  along a Ricci flow of bounded geometry $\beta \mapsto g_{ab}(\beta)$, $0\leq \beta \leq \beta ^{*}$, the scalar heat flow  $(\frac{\partial }{\partial \beta}-\Delta )\,\omega (\beta)=0$, the covector heat flow $(\frac{\partial }{\partial \beta}-\Delta_{vec} )\,v_a(\beta)=0$, and the linearized Ricci flow  $(\frac{\partial }{\partial \beta}-\Delta_L )\,h_{ab}(\beta)=0$, can be compactly represented by the kernel of the Hodge--DeRham--Lichnerowicz (HDRL) heat operator
\begin{equation} 
\bigcirc _{d}\doteq \frac{\partial  }{\partial  \beta   }-\Delta _{d} \;,
\label{allops}
\end{equation}
thought of as acting on the appropriate parabolic space of $\beta$--dependent sections: $C^{\infty }(\Sigma\times \mathbb{R},\mathbb{R})$ for the scalar heat operator, $C^{\infty }(\Sigma\times \mathbb{R},T^*\Sigma )$ for the covector heat operator, and finally $C^{\infty }(\Sigma\times \mathbb{R},\otimes ^2_ST^*\Sigma )$ for the Lichnerowicz heat equation. Alternatively, we may consider $\bigcirc _{d}$ as acting on the cartesian product $\times _{p=0}^2\,C^{\infty }(\Sigma\times \mathbb{R},\otimes ^p_ST^*\Sigma )$, ($p=0,1,2$), and  use the compact notation

\begin{eqnarray}
        \bigcirc _{d}\,\left( \begin{array}{ll}
        \omega(\beta)   \\
        v_i(\beta)  \\
        h_{ab}(\beta)   \end{array} \right):= \left( \begin{array}{ll}
        (\frac{\partial  }{\partial  \beta   }-\Delta)\,\omega (\beta)    \\
        (\frac{\partial  }{\partial  \beta   }-\Delta _{vec} )\,v_i (\beta) \\
        (\frac{\partial  }{\partial  \beta   }-\Delta _{L} )\,h_{ab} (\beta)   \end{array} \right)\;.
        \label{compact1}
\end{eqnarray}
\end{remark}

\noindent If we consider the $L^2(\Sigma\times \mathbb{R} ,d\beta\,d\mu_{g(\beta)})$ parabolic pairing between the spaces $C^{\infty }(\Sigma\times \mathbb{R},\otimes ^p_ST^*\Sigma )$  and $C^{\infty }(\Sigma\times \mathbb{R},\otimes ^p_ST\Sigma )$, ($p=0,1,2$), we can also introduce \cite{carfback} the  backward $L^2$--conjugated flow associated with (\ref{divfree}), generated by the operator 
\begin{equation}
\bigcirc^{*} _{d}\doteq -\frac{\partial }{\partial \beta }-\triangle _{d} +\mathcal{R}\;,
\label{conjflow0}
\end{equation}
acting on the appropriate space of $\beta$-dependent sections $C^{\infty }(\Sigma\times \mathbb{R},\otimes ^p_ST\Sigma )$, ($p=0,1,2$), or in a  more compact form 

\begin{eqnarray}
        \bigcirc^{*} _{d}\,\left( \begin{array}{ll}
        \varpi(\beta)   \\
        W^i (\beta)\\
        H^{ab} (\beta)  \end{array} \right):= \left( \begin{array}{ll}
        (-\frac{\partial  }{\partial  \beta   }-\Delta+\mathcal{R})\,\varpi(\beta)   \\
        (-\frac{\partial  }{\partial  \beta   }-\Delta _{vec}+\mathcal{R} )\,W^i(\beta) \\
        (-\frac{\partial  }{\partial  \beta   }-\Delta _{L}+\mathcal{R} )\,H^{ab}(\beta)  \end{array} \right)\;,
        \label{compact2}
\end{eqnarray}

\noindent for $(\varpi ,W,H)\in \times _{p=0}^2\,C^{\infty }(\Sigma\times \mathbb{R},\otimes ^p_ST^\Sigma )$.\\
\begin{figure}[h]
\includegraphics[scale=0.3]{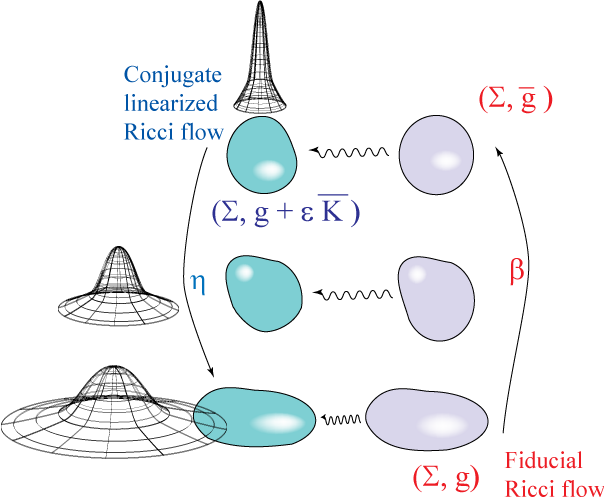}
\caption{The conjugate Hodge--DeRham--Lichnerowicz heat equation has a number of unexpected properties allowing a better control of the linearized Ricci flow. These properties are related to the fact that the Ricci curvature evolves according to the forward Lichnerowicz heat equation.}
\label{fig: 13}   
\end{figure}
\\
\noindent The role of $\bigcirc _{d}$ and $\bigcirc^{*} _{d}$ in Ricci flow conjugation  is connected with their interaction with the Berger--Ebin splitting of $\mathcal{T}_{(\Sigma ,g(\beta))}\mathcal{M}et(\Sigma )$.  
To make this interaction explicit, we organize in a unique pattern a number of  commutation rules among $\bigcirc _{d}$, $\bigcirc^* _{d}$, and the action of $\delta_{g(\beta)}$, and $\delta^*_{g(\beta)}$. Some of these relations are rather familiar, a few others extend, in a non trivial way, known properties of the conjugate scalar heat flow.  
\begin{lemma}
\label{commlemm}
{(Commutation rules)}\\
\noindent
If $\bigcirc _{d}$ and $\bigcirc^{*} _{d}$ respectively denote the  Hodge--DeRham--Lichnerowicz heat operator and its conjugate along a fiducial Ricci flow $\beta \rightarrow g_{ab}(\beta)$ on $\Sigma\times [0,\beta^*]$, then  for any $v(\beta)\in C^{\infty }(\Sigma\times \mathbb{R},T^*\Sigma)$,  
$h(\beta)\in C^{\infty }(\Sigma\times \mathbb{R},\otimes ^2_ST^*\Sigma ) $, and $H(\beta)\in C^{\infty }(\Sigma\times \mathbb{R},\otimes ^2_ST\Sigma )$ we have
\begin{eqnarray}
&&\frac{d}{d\beta}\,\left(h(\beta),H(\beta) \right)_{L^2(\Sigma ,d\mu_{g(\beta)})}:=
\frac{d}{d\beta}\,\int_{\Sigma }h_{ab}(\beta)\,H^{ab}(\beta)\,d\mu_{g(\beta)}
\label{ldue}\\
\nonumber\\
&&= \int_{\Sigma }\left[H^{ab}(\beta)\bigcirc _{d}\,h_{ab}(\beta)\,-h_{ab}(\beta)\bigcirc^{*} _{d}\,H^{ab}(\beta) \right]d\mu_{g(\beta)}\;, \nonumber
\end{eqnarray}
moreover the following set of commutation rules hold:
\begin{equation}
tr_{g(\beta)}\,\left(\bigcirc _{d}\,h(\beta)\right)= \bigcirc _{d}\,\left(tr_{g(\beta)}\,h(\beta) \right)-2\,\mathcal{R}^{ik}(\beta)\,h_{ik}(\beta)\;,
\label{comm0}
\end{equation}
\begin{equation}
tr_{g(\beta)}\,\left(\bigcirc^{*} _{d}\,H(\beta)\right)= \bigcirc^{*} _{d}\,\left(tr_{g(\beta)}\,H(\beta) \right)-2\,\mathcal{R}_{ik}(\beta)\,H^{ik}(\beta)\;,
\label{comm0*}
\end{equation}
\begin{equation}
\bigcirc _{d}\;\left(\delta ^*_{g(\beta)}\,v^{\sharp }(\beta) \right)=\delta ^*_{g(\beta)}\,\left(\bigcirc _{d}\;v(\beta) \right)^{\sharp }\;,
\label{comm1}
\end{equation}
\begin{equation}
\bigcirc^{*} _{d}\;\left(\delta _{g(\beta)}\,H(\beta) \right)=\delta _{g(\beta)}\,\left(\bigcirc^{*} _{d}\;H(\beta) \right)\;,
\label{comm1*}
\end{equation}
\begin{eqnarray}
&& \bigcirc _{d}\,\left( \delta _{g(\beta)}\,h(\beta)\right)
=  \delta_{g(\beta)}\,\left(\bigcirc _{d}\,h(\beta)\right)-2\mathcal{R}^{ik}(\beta)\,\nabla_i\,h_{kl}(\beta)\,dx^l  \label{comm2} \\
\nonumber\\
&&-2\,h_{ik}(\beta)\,\left(\nabla_l\,\mathcal{R}^{ik}(\beta)-\,\nabla^k\,\mathcal{R}^{i}_l(\beta)\right)\,dx^l \;,\nonumber
\end{eqnarray}
\begin{eqnarray}
&& \bigcirc^* _{d}\,\left( \delta^* _{g(\beta)}\,v(\beta)\right)
= \delta^*_{g(\beta)}\,\left(\bigcirc^* _{d}\,v(\beta)\right)-\left[v_a\nabla^k\mathcal{R}_{kb}(\beta)\,+
v_b\nabla^k\mathcal{R}_{ka}(\beta)\right.  \label{comm2*} \\
\nonumber\\
&&\left. +2\,v^{k}(\beta)\,\left(\nabla_a\,\mathcal{R}_{bk}(\beta)+\nabla_b\,\mathcal{R}_{ak}(\beta)-\,\nabla_k\,\mathcal{R}_{ab}(\beta)\right)\right]
\,dx^a\otimes dx^b \;,\nonumber
\end{eqnarray}
where $tr_{g(\beta)}$ and  ${}^{\sharp }$ respectively denote the $g(\beta)$--dependent trace and the $g(\beta)$--rising operator along the fiducial flow.
\end{lemma}
\begin{proof}
The relation (\ref{ldue}), describing the evolution of the $L^2$ pairing between $\mathcal{T}_{(\Sigma ,g(\beta))}\mathcal{M}et(\Sigma )$ and $C^{\infty }(\Sigma\times \mathbb{R},\otimes ^2_ST\Sigma )$,  immediately follows from adding and subtracting $H^{ab}\Delta _L\,h_{ab}$ to
\begin{eqnarray}
&&\;\;\;\;\;\;\;\;\;\frac{d}{d\beta}\,\left(h(\beta),H(\beta) \right)_{L^2(\Sigma ,d\mu_{g(\beta)})}
\label{}\\
\nonumber\\
&&= \int_{\Sigma }\left[h_{ab}\frac{\partial}{\partial \beta}H^{ab}+H^{ab}\frac{\partial}{\partial \beta}h_{ab} -
\mathcal{R}h_{ab}H^{ab}\right]d\mu_{g(\beta)}\;, \nonumber
\end{eqnarray}
and exploiting the fact that $\Delta _L$ is a self--adjoint operator with respect to $d\mu_{g(\beta)}$.
The commutation relations (\ref{comm0}) and (\ref{comm0*}) between the $g(\beta)$--depenedent trace and the operators $\bigcirc _{d}$ and $\bigcirc^* _{d}$ are elementary consequences of the Ricci flow evolution. Similarly well--known, (see \emph{e.g.} \cite{chowluni}), is the commutation rule (\ref{comm1}). A direct computation shows that (\ref{comm2}) is a consequence of the Weitzenb$\ddot o$ck formula, (see \emph{e.g.} \cite{carfback}), 
\begin{equation}
\nabla ^{k}\,\triangle _{L}\,S_{kl}=\triangle \,\nabla ^{k}\,S_{kl}+S_{ka}\,\nabla _{l}\,R^{ka}-R_{l}^{a}\,\nabla ^{k}\,S_{ka}-2S_{ka}\,\nabla ^{k}\,R_{l}^{a}
\label{lichWeitz}\;,
\end{equation}
and of the Ricci flow rule
\begin{equation}
\frac{\partial }{\partial \beta }\,\nabla ^{k}\,S_{kl}= g^{ik}\,\nabla _{i}\left(\frac{\partial }{\partial \beta }\,S_{kl} \right)+    2R^{ik}\nabla _{i}\,S_{kl}+\,S_{mi}\,\nabla_{l}R^{mi}\;,
\label{timepart}
\end{equation}
both valid for any symmetric bilinear form $S\in C^{\infty }(\Sigma\times \mathbb{R},\otimes ^2_ST^*\Sigma )$. Explicitly we compute
\begin{eqnarray}
&&\left(\frac{\partial }{\partial \beta}-\Delta_d\right)\,\nabla^k\,h_{kl}=\nabla ^{k}\left(\frac{\partial }{\partial \beta }\,h_{kl} \right)+    2\mathcal{R}^{ik}\nabla _{i}\,h_{kl}\\
&&+\,h_{mi}\,\nabla_{l}\mathcal{R}^{mi}-\Delta(\nabla^k\,h_{kl})+\mathcal{R}^j_l\,\nabla^k\,h_{kl} \nonumber \\
&&=\nabla ^{k}\left(\frac{\partial }{\partial \beta }\,h_{kl} \right)+    2\mathcal{R}^{ik}\nabla _{i}\,h_{kl}
+\,h_{mi}\,\nabla_{l}\mathcal{R}^{mi} \nonumber\\
&&-\nabla^k\left(\Delta_L\,h_{kl} \right)+h_{ka}\nabla_l\mathcal{R}^{ka}-\mathcal{R}^a_l\nabla^k\,h_{ka} \nonumber\\
&&-2\,h_{ka}\,\nabla^k\mathcal{R}^a_l+\mathcal{R}^j_l\nabla^k\,h_{kj}\nonumber\\
&&=\nabla^k\left(\frac{\partial }{\partial \beta}\,h_{kl}-\Delta_L\,h_{kl} \right)+2\mathcal{R}^{ik}\nabla_i\,h_{kl} \nonumber\\
&&+2h_{ik}(\nabla_l\mathcal{R}^{ik}-\nabla^k\mathcal{R}^i_l)\nonumber\;,
\end{eqnarray}
where in the first and in the forth line we exploited (\ref{timepart}) and (\ref{lichWeitz}), respectively. This provides (\ref{comm1}). The basic relation (\ref{comm1*}) follows from a rather lengthy but otherwise straightforward computation. According to the definition of $\bigcirc^* _{d}$ we have 
\begin{eqnarray}
&&\bigcirc^* _{d}\,\left(\nabla_a H^{ab} \right)= -\frac{\partial }{\partial\beta}\left(g^{bl}\,\nabla_a H^{a}_{l}\right)-\left(\Delta_d-\mathcal{R}\right)\,\nabla_a H^{ab}\\
&&=-2\mathcal{R}^{bl}\nabla^a H_{al}-g^{bl}\frac{\partial}{\partial\beta}\nabla^a H_{al}-\Delta\nabla_a H^{ab}+\mathcal{R}^b_l\nabla_a\,H^{al}
+\mathcal{R}\nabla_a H^{ab} \nonumber\;.
\end{eqnarray}
By exploiting again (\ref{timepart}) and (\ref{lichWeitz}), this latter expression reduces to
\begin{eqnarray}
&&=-2\mathcal{R}^{bl}\nabla^a H_{al}-g^{bl}\nabla^a \left(\frac{\partial}{\partial\beta}H_{al}\right)-2g^{bl}\mathcal{R}^{ia}\nabla_i\,H_{al}\\
&&-g^{bl}H_{ai}\nabla_l\mathcal{R}^{ai}-\Delta\nabla_a H^{ab}+\mathcal{R}^b_l\nabla_a\,H^{al}
+\mathcal{R}\nabla_a H^{ab}\nonumber\\
&&=-\mathcal{R}^{bl}\nabla^a H_{al}-g^{bl}\nabla^a \left(\frac{\partial}{\partial\beta} \left(g_{ac}\,g_{ld}H^{cd}\right) \right)-2g^{bl}\mathcal{R}^{ia}\nabla_i\,H_{al}\nonumber\\
&&-g^{bl}H_{ai}\nabla_l\mathcal{R}^{ai}-\Delta\nabla_a H^{ab}+\mathcal{R}\nabla_a H^{ab}\nonumber\\
&&=-\mathcal{R}^{bl}\nabla^a H_{al}-g^{bl}\nabla^a \left[g_{ac}\,g_{ld}\frac{\partial}{\partial\beta}H^{cd}-2\mathcal{R}_{ac}\,g_{ld}\,H^{cd}  \right.\nonumber\\
&&\left. -2\mathcal{R}_{ld}\,g_{ac}\,H^{cd} \right]-2g^{bl}\mathcal{R}^{ia}\nabla_i\,H_{al}\nonumber\\
&&-g^{bl}H_{ai}\nabla_l\mathcal{R}^{ai}-\Delta\nabla_a H^{ab}+\mathcal{R}\nabla_a H^{ab}\nonumber\\
&&=-\mathcal{R}^{bl}\nabla^a H_{al}+2H^{cb}\,\nabla^a\mathcal{R}_{ac}+2\mathcal{R}_{ac}\,\nabla^a\,H^{cb}+2H^{d}_a\,\nabla^a\mathcal{R}^b_d\nonumber\\
&&+2\mathcal{R}^{b}_d\nabla^a H_{a}^d-\nabla_c\left(\frac{\partial}{\partial\beta}H^{cb} \right)-2\mathcal{R}^{ia}\nabla_i H_a^b-H^{mi}\nabla^b\mathcal{R}_{mi}\nonumber\\
&&-\nabla_a\,\Delta_LH^{ab}+H_{ka}\nabla^b\mathcal{R}^{ka}-\mathcal{R}^a_b\nabla^kH_{ka}-2H_{ka}\nabla^k\mathcal{R}^{ab}+\mathcal{R}\nabla_aH^{ab}\nonumber\\
&&=\nabla_c\left[-\frac{\partial}{\partial\beta}-\Delta_L+\mathcal{R} \right]\,H^{cb}\nonumber\;,
\end{eqnarray}
where the last line follows from cancelling terms and by using the contracted Bianchi identity. This proves (\ref{comm1*}). Finally, (\ref{comm2*}) is a  consequence of the known Ricci flow identities, (see \emph{e.g.} \cite{chowluni}),
\begin{equation}
\nabla_i\left(\Delta v_j-\mathcal{R}_j^kv_k  \right)= \Delta_L \nabla_iv_j-v^k \left(\nabla_i\mathcal{R}_{jk}+\nabla_j\mathcal{R}_{ik}-\nabla_k\mathcal{R}_{ij}\right)\;,
\end{equation}
\begin{equation}
\frac{\partial}{\partial\beta}\nabla_i v_j= \nabla_i\left(\frac{\partial}{\partial\beta} v_j\right)+v^k\left(\nabla_i\mathcal{R}_{jk}+\nabla_j\mathcal{R}_{ik}-\nabla_k\mathcal{R}_{ij}\right)\;,
\end{equation}
which hold for any smooth $\beta$--dependent covector field $v(\beta)$. According to these we compute
\begin{eqnarray}
&&\left(\frac{\partial}{\partial\beta}+\Delta_d-\mathcal{R} \right)(\nabla_a v_b+\nabla_b v_a)\\
&&=\nabla_a\left(\frac{\partial}{\partial\beta}+\Delta_d-\mathcal{R} \right)v_b+\nabla_b\left(\frac{\partial}{\partial\beta}+\Delta_d-\mathcal{R} \right)v_a
\nonumber\\
&&+4v^k \left(\nabla_a\mathcal{R}_{bk}+\nabla_b\mathcal{R}_{ak}-\nabla_k\mathcal{R}_{ab}\right)+v_a\nabla^k\mathcal{R}_{kb}+v_b\nabla^k\mathcal{R}_{ka}\;,\nonumber
\end{eqnarray}
from which (\ref{comm2*}) follows.
\end{proof}
From the commutation rule (\ref{comm1}) we get a familiar property of the linearized Ricci flow which we express as the
\begin{lemma}
\label{lemmDiff}
If $\beta\mapsto v(\beta)$, $0\leq\beta\leq\beta^*$, is a solution of  $\bigcirc _{d}\;v(\beta)=0$, then the induced flow
\begin{equation}
\beta\mapsto \delta ^*_{g(\beta)}\,v^{\sharp }(\beta)\in \mathcal{U}_{g(\beta)}\cap\,\mathcal{T}_{(\Sigma ,g(\beta))}\mathcal{M}et(\Sigma )\times [0,\beta^*]\;,
\end{equation}
is a solution of the linearized Ricci flow $\bigcirc _{d}\;\delta ^*_{g(\beta)}\,v^{\sharp }(\beta)=0$. 
\end{lemma}
\noindent This implies that the forward evolution along the linearized Ricci flow naturally preserves $Im\,\delta ^*_{g(\beta)}$, and that data 
$h(\beta=0)=\left.\delta ^*_{g}\,v^{\sharp }\right|_{\beta=0}$ evolve in  $\mathcal{U}_{g(\beta)}\times [0,\beta^*]$.
Conversely, if $h(\beta)\in Ker\,\bigcirc _{d}\,\cap\,\mathcal{T}_{(\Sigma ,g(\beta))}\mathcal{M}et(\Sigma )$, $0\leq\beta\leq\beta^*$, is a solution of the linearized Ricci flow with $h(\beta=0)\in Ker\,\delta_{g(\beta=0)}$, then in general $h(\beta)\not\in Ker\,\delta_{g(\beta)}$ for $\beta>0$, and  $\beta \mapsto h(\beta)$ does not evolve in the affine slices $\mathcal{S}_{g(\beta)}\times [0,\beta^*]$.
\begin{figure}[h]
\includegraphics[scale=0.4]{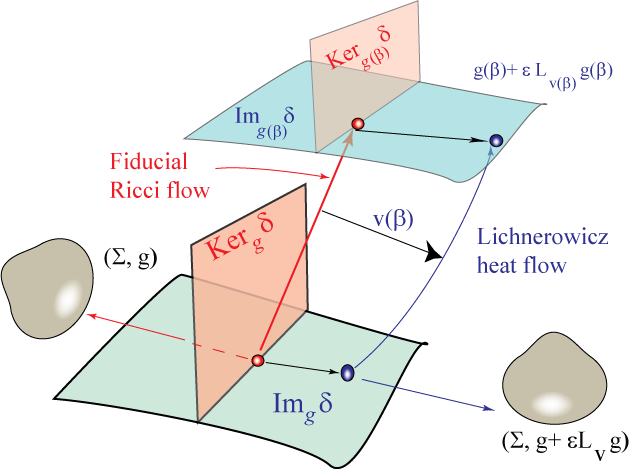}
\caption{The forward evolution along the linearized Ricci flow (in the HDRL representation) naturally preserves $Im\,\delta ^*_{g(\beta)}$.}
\label{fig: 14}   
\end{figure}
\\
\begin{remark}
If we decompose $h(\beta)\in Ker\, \bigcirc _{d}$ according to the $g(\beta)$--dependent Berger--Ebin splitting 
\begin{equation}
h_{ab}(\beta)=h_{ab}^{T}(\beta)+L_{w^{\sharp }(\beta)}\,g_{ab}(\beta)\,,\;\;\; \nabla^a\,h_{ab}^{T}(\beta)=0\;,
\end{equation}
we immediately get from $\bigcirc _{d}\,h_{ab}(\beta)=0$, and the commutation rule (\ref{comm1}), the relation
\begin{equation}
 \bigcirc _{d}\,h_{ab}^{T}(\beta)=-2\,\delta _{g(\beta)}^*\,\left( \bigcirc _{d}\,w(\beta)\right)^{\sharp }\;.
 \end{equation}
This directly shows that the non--trivial part $h^{T}_{ab}(\beta)$ of a solution of the linearized Ricci flow $h(\beta)$  dynamically generates $\mathcal{D}iff(\Sigma )$ reparametrizations, $\delta_{g(\beta)} ^{*}\left(\bigcirc _{d}\,\,w(\beta) \right)$, of the given fiducial flow. 
Thus, whereas elements in $Ker\,\bigcirc _{d}\,\cap C^{\infty }(\Sigma\times \mathbb{R},T^*\Sigma )$ generate a natural evolution in 
$\mathcal{U}_{g(\beta)}\times [0,\beta^*]$,   there is no natural way of preserving the subspace $\ker \,\delta_{g(\beta)}$ along the forward flow (\ref{divfree}) if we do not impose strong restrictions on the underlying fiducial Ricci flow $\beta \mapsto g_{ab}(\beta )$ \cite{avez, buzzanca, guenther, guenther3, Ye}. 
\end{remark}
\noindent
The situation  is fully reversed if we consider the conjugated flow generated on $C^{\infty }(\Sigma\times \mathbb{R} ,\otimes ^{2}T\,\Sigma)$  by $\bigcirc^{*} _{d}$, since in such a case the commutation relation (\ref{comm1*}) immediately implies \cite{carfback} the 
\begin{lemma}
\label{divlemm}
Let $\beta \mapsto g_{ab}(\beta )$ be a Ricci flow with bounded geometry on $\Sigma_{\beta }\times [0,\beta ^{*}]$, $\beta^{*}<T_0$, and let 
$\eta \mapsto g_{ab}(\eta )$, \ $\eta \doteq \beta ^{*}-\beta $, denote the corresponding backward Ricci flow on $\Sigma_{\eta } \times [0,\beta ^{*}]$obtained by the time reversal $\beta\mapsto \eta \doteq \beta^*-\beta$. Then $Ker\,\delta_{g(\eta)} $ is an invariant subspace for  $\bigcirc^{*} _{d}$, \emph{i.e.},
\begin{equation} 
\bigcirc^{*} _{d}\left(Ker\,\delta_{g(\eta)}  \right)\subset Ker\,\delta_{g(\eta)}\;,
\end{equation}
along $\eta \mapsto g_{ab}(\eta )$. In particular, if $\eta\mapsto H(\eta)$ with $H(\eta=0)\in Ker\,\delta_{g(\eta=0)}$ is a flow solution of the parabolic initial value problem $\bigcirc^{*} _{d}\,H(\eta)=0$ on $\Sigma\times[0,\beta ^{*}]$, i.e.
\begin{equation} 
\begin{tabular}{l}
$\frac{\partial }{\partial \eta }{H}^{ab}=\Delta_{d}{H}^{ab}-\,\mathcal{R}{H}^{ab}\;,$\\
\\
${H}^{ab}(\eta=0)={H}^{ab}\in C^{\infty }(\Sigma ,\otimes ^{2}T\,\Sigma)\cap Ker\,\delta_{g} \;,$%
\end{tabular}
\;   \label{transVol}
\end{equation}  
 then $\eta\mapsto H(\eta)\in \mathcal{S}_{g(\eta)}\cap\,\mathcal{T}_{(\Sigma ,g(\eta))}\mathcal{M}et(\Sigma )\times [0,\beta^*]$. 
\end{lemma}
\begin{figure}[h]
\includegraphics[scale=0.4]{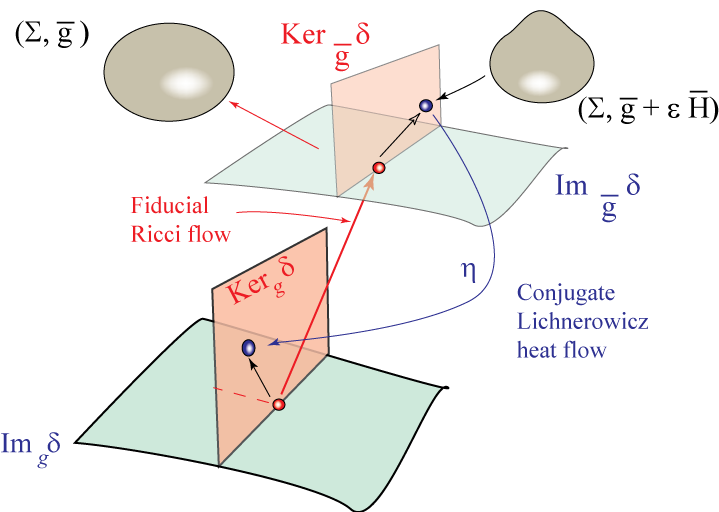}
\caption{The backward evolution along the conjugate linearized Ricci flow (in the HDRL representation) naturally preserves metric perturbations which are in $Ker\,\delta _{\overline{g}}$.}
\label{fig: 15}   
\end{figure}
\noindent Notice that, according to (\ref{ldue}),  the conjugate linearized Ricci flow $\eta\mapsto H^{ab}(\eta)$ is characterized by
\begin{equation}
\frac{d}{d\eta }\,\int_{\Sigma }h_{ab}(\eta )H^{ab}(\eta )d\mu _{g(\eta )}=0\;,
\label{hevol} 
\end{equation}
along any solution $\beta\mapsto h_{ab}(\beta)$ of the linearized Ricci flow $\bigcirc _{d}\,h(\beta)=0$ on $(\Sigma,g(\beta))\times [0,\beta^*]$, (with $\beta=\beta^*-\eta $). In particular, from the commutation relation (\ref{comm1}), it follows  that
\begin{equation}
\frac{d}{d\eta }\,\int_{\Sigma }\left(\delta ^*_{g(\beta)}\,v^{\sharp }(\eta ) \right)_{ab}\,H^{ab}(\eta )d\mu _{g(\eta )}=0\;,
\end{equation}
$\forall \beta\mapsto v(\beta)$,\; $\bigcirc_{d}\,v(\beta)=0$. Surprisingly, these elementary properties directly imply the following strong geometrical characterization of $\eta\mapsto H^{ab}(\eta)$:

\begin{theorem} (see \cite{carfback})
\label{nondiss}
Let $\eta \mapsto {{H}}^{ab}(\eta )$ be a solution of the conjugate linearized Ricci flow (\ref{transVol}) on $\Sigma\times[0,\beta ^{*}]$, then $\int_{\Sigma }R_{ab}(\eta )H^{ab}(\eta )d\mu _{g(\eta )}$ and $\int_{\Sigma }\left(g_{ab}(\eta)-2\eta\,R_{ab}(\eta )\right)H^{ab}(\eta )d\mu _{g(\eta )}$ are conserved along $\eta\mapsto (g(\eta), H(\eta))$,
\begin{equation}
\label{bello1}
\frac{d}{d\eta }\,\int_{\Sigma }R_{ab}(\eta )H^{ab}(\eta )d\mu _{g(\eta )}=0\;,
\end{equation}

\begin{equation}
\label{bello2}
\frac{d}{d\eta }\,\int_{\Sigma }\left(g_{ab}(\eta)-2\eta\,R_{ab}(\eta )\right)H^{ab}(\eta )d\mu _{g(\eta )}=0\;.
\end{equation} 
\end{theorem}

\noindent This result characterizes the solutions $\eta\mapsto H(\eta)$, of the conjugate linearized Ricci flow $\bigcirc_d^*\,H(\eta)=0$, as providing the localizing directions (in $L^2$ sense) for the (non--linear) evolution of $\mathcal{R}ic(g(\beta))$ 
\begin{equation}
\frac{\partial }{\partial \beta }\,{\mathcal{R}}_{ij}=\Delta_L{\mathcal{R}}_{ij}=\Delta
{\mathcal{R}}_{ij}+2{\mathcal{R}}_{kijl}{\mathcal{R}}^{kl}-2\mathcal{R}{\mathcal{R}}
_{ik}{\mathcal{R}}_j^{k}\;.
 \label{scaleRic}
\end{equation}

\section{Conjugated mode expansion}
\label{conjmodexp}
 Let us consider the pair of conjugated  heat flows (\ref{heatDef}),  $\mathcal{C}(\beta)\in \times _{p=0}^2\,C^{\infty }(\Sigma\times \mathbb{R},\otimes ^p_ST^*\Sigma )$ and 
$\bar{\mathcal{C}}^{\;\sharp }(\eta)$ $\in \times _{p=0}^2\,C^{\infty }(\Sigma\times \mathbb{R},\otimes ^p_ST\Sigma )$, solutions of (\ref{linDT0comp0}) and (\ref{conlinDT0comp0}). Notice that
\begin{equation}
\left( \begin{array}{ll} 
      \omega (\beta)  \\
       v_i(\beta)  \\
       h_{ab}(\beta)   \end{array} \right)\mapsto \left( \begin{array}{ll} 
       \varrho(\beta)  \\
         J_{i}(\beta)  \\
       K_{ab}(\beta)   \end{array} \right)\;,\;\;\;\;\left( \begin{array}{ll} 
      \varpi  (\eta)  \\
       W^i(\eta)  \\
       H^{ab}(\eta)   \end{array} \right)\mapsto \left( \begin{array}{ll} 
       \overline{\varrho}(\eta)  \\
         \overline{J}^{i}(\eta)  \\
       \overline{K}^{ab}(\eta)   \end{array} \right)
\end{equation}

\noindent provides the obvious dictionary among the fields discussed in the previous sections  and their physical counterparts defining the data sets $\mathcal{C}_{g}(\Sigma)$ and $\bar{\mathcal{C}}_{g}(\Sigma)$. Theorem \ref{nondiss}, Lemma \ref{commlemm} and Lemma  \ref{divlemm} allow a rather complete analysis of 
the conjugation between two given $n$--dimensional Einstein initial data sets $\mathcal{C}_{g}(\Sigma)$ and $\bar{\mathcal{C}}_{g}(\Sigma)$
as characterized by definition \ref{theDEF}. We start by exploiting the conjugacy relation between $\bigcirc_d$, $\bigcirc_d^*$, and the $L^2(\Sigma ,d\mu_{g(\eta)})$ spectral resolution of  $C^{\infty }(\Sigma,\otimes _S^2T\Sigma )$ generated by the operator $\Delta _d-\mathcal{R}(\eta)$ at a given fixed $\eta\in [0,\beta^*]$. \\
\\
\noindent 
Since we are considering $\mathcal{C}_{\bar{g}}(\Sigma )$ as the reference data,  let us set $\eta=0$ so that $g(\eta=0)=\bar{g}\in \mathcal{C}_{\bar{g}}(\Sigma )$. From the spectral theory of Laplace type operators on closed Riemannian manifolds (see  \cite{gilkey}, and \cite{gilkey2} (Th. 2.3.1)), it follows that, on $(\Sigma ,\bar{g})$, the elliptic operator $P_{d}\doteq -\Delta _{d}+\mathcal{R}(\bar{g})$  has a discrete spectral resolution $\left\{\overline{\Phi} ^{\;\sharp}_{(n)},\,\lambda^{(d)} _{(n)}\right\}$,
\begin{equation}
\overline{{\Phi }}^{\;\sharp }_{(n)}:=\,\left( \begin{array}{ll}
        \overline{\Phi}_{(n)}  \\
        \\
        \overline{\Phi}^{\;i}_{(n)}  \\
        \\
        \overline{\Phi}^{\;ab}_{(n)}   \end{array} \right)\;,
\end{equation}
where $\overline{\Phi}_{(n)}$,\;$\overline{\Phi}^{\;i}_{(n)}$, \; and $\overline{\Phi}^{\;ab}_{(n)}$ respectively  are the eigenfunction of the conjugate scalar Laplacian $-\Delta _{d}+\mathcal{R}(\bar{g})$, of the conjugate vector Laplacian  $-\Delta _{vec}+\mathcal{R}(\bar{g})$, and of  the cojugate Lichnerowicz--DeRham Laplacian $-\Delta _{L}+\mathcal{R}(\bar{g})$. The eigenvalues
$\lambda _{(1)}^{(d)}\leq\lambda _{(2)}^{(d)}\leq \ldots\infty $ have finite multiplicities, and are contained in $[-\bar{C}^{(d)},\,\infty )$ for some constant $\bar{C}^{(d)}$ depending from the (bounded) geometry of $(\Sigma ,\bar{g})$. 
Moreover, for any $\varepsilon >0$, there exists an integer $n_{(d)}(\varepsilon)$ so that $n^{\frac{2}{3}-\varepsilon }\leq \lambda _{(n)}^{(d)}\leq n^{\frac{2}{3}+\varepsilon }$, for $n\geq n_{(d)}(\varepsilon)$. 
The set of eigentensor $\left\{\overline{\Phi} ^{\;\sharp}_{(n)}\right\}$, \ $\overline{\Phi} ^{\;\sharp}_{(n)}\in C^{\infty }(\Sigma ,\otimes ^{p}T\,\Sigma)$,\; $p=0,1,2$, with
\begin{equation}
P_{d}\,\overline{\Phi} ^{\;\sharp}_{(n)}=\left(-\Delta _{d}+\mathcal{R}(\overline{g})\right)\,\overline{\Phi} ^{\;\sharp}_{(n)}=\lambda _{(n)}^{(d)}\,
\overline{\Phi} ^{\;\sharp}_{(n)}
\end{equation}
provide a complete orthonormal basis for $L^{2 }(\Sigma ,\otimes ^{p}T\,\Sigma;\,d\mu_{\bar{g}})$,\;$p=0,1,2$. If for a tensor field 
$H\in L^{2 }(\Sigma ,\otimes ^{p}T\,\Sigma)$, $p=0,1,2$, we denote by $c_{n}\doteq \left(H,\overline{\Phi}^{\;\sharp} _{(n)} \right)_{L^{2}(\Sigma)}$ the corresponding Fourier coefficients, then we have that  
$H\,\in C^{\infty }(\Sigma ,\otimes ^{p}T\,\Sigma)$,\;$p=0,1,2$, iff $\lim_{n\rightarrow \infty }n^{k}\,c_{n}=0$, $\forall k\in \mathbb{N}$, (\emph{i.e}, the $\{c_{n}\}$ are rapidly decreasing). Also, if $|H|_{k}$ denotes the $\sup$--norm of $k^{th}$ covariant derivative of $H$, then there exists $j(k)$ so that $|H|_{k}\leq \,n^{j(k )}$ if $n$ is large enough. This result implies in particular that the series $H=\sum_{n}c_{n}\,\overline{\Phi}^{\;\sharp} _{(n)}$ converges absolutely to $H$, and that the linear span of the $\left\{\overline{\Phi} ^{\;\sharp}_{(n)}\right\}$ is dense in the $C^{\infty }$ topology.
\begin{figure}[h]
\includegraphics[scale=0.4]{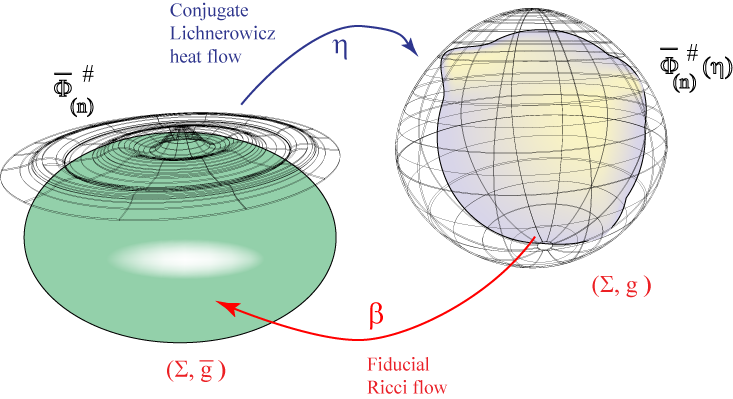}
\caption{The backward evolution, along the conjugate linearized Ricci flow, of the eigenmodes  $\left\{\overline{\Phi} ^{\;\sharp}_{(n)}\right\}$ of the 
 the elliptic operator $P_{d}\doteq -\Delta _{d}+\mathcal{R}(\bar{g})$ does not preserve the eigenfunction property. However the associated Fourier coefficients are preserved.}
\label{fig: 16}   
\end{figure}
\noindent
With these preliminary remarks along the way, we have
\begin{theorem} (Conjugate mode expansion)
\label{LtwoK}\\
\noindent
Along the fiducial backward Ricci flow $\eta\longmapsto g(\eta)$ on $\Sigma\times [0,\beta^*]$,
 let $\left\{\Phi ^{\sharp}_{(n)}(\eta)\right\}\in C^{\infty }(\Sigma\times [0,\beta^*] ,\otimes ^{p}T\,\Sigma)$,\; $p=0,1,2$, denote the flows defined by
\begin{equation}
\bigcirc^*_d\,\Phi ^{\sharp}_{(n)}(\eta)=0\;,\;\;\;\;\Phi ^{\sharp}_{(n)}(\eta=0):=\overline{\Phi} ^{\;\sharp}_{(n)}\;,\;\;\;n\in \mathbb{N}\;.
 \end{equation}
If $\mathcal{C}(\beta)\in C^{\infty }(\Sigma \times [0,\beta^*],\otimes ^pT^*\Sigma )$,\; $\bigcirc_d\,\mathcal{C}(\beta)=0$,\; $\mathcal{C}(\beta=0)=\mathcal{C}$, is the forward evolution 
of $(\varrho,\,J_i,\,K_{ab})\in \mathcal{C}_{g}(\Sigma )$ then, in terms of the initial data $\mathcal{C}(\beta=0):=(\varrho,\,J_i,\,K_{ab})\in \mathcal{C}_{g}(\Sigma )$,  we can write
\begin{equation}
\varrho(\beta^*)=\,\sum_n\,\overline{\Phi}^{\;(n)}\,\int_{\Sigma }\,\varrho\;\Phi_{(n)}\,d\mu_{g}\;,
\label{modexprho}
\end{equation}
\begin{equation}
J_{a}(\beta^*)=\,\sum_n\,\overline{\Phi} _{a}^{\;(n)}\,\int_{\Sigma }\,J_{i}\;\Phi ^{i}_{(n)}\,d\mu_{g}\;,
\label{modexpJ}
\end{equation}
\begin{equation}
K_{ab}(\beta^*)=\,\sum_n\,\bar{\Phi} _{ab}^{(n)}\,\int_{\Sigma }\,K_{ij}\,\Phi ^{ij}_{(n)}\,d\mu_{g}\;,
\label{modexp}
\end{equation}
where the integrals appearing on the right hand side are all evaluated at $\beta=0$, \emph{e.g.}  $\int_{\Sigma }\,K_{ij}\,\Phi ^{ij}_{(n)}\,d\mu_{g}$ $:=\int_{\Sigma }\,K_{ij}(\beta=0)\,\Phi ^{ij}_{(n)}(\eta=\beta^*)\,d\mu_{g(\beta=0)}$.
Moreover, under the same hypotheses and notation, we have
\begin{equation}
\int_{\Sigma }\,\left|\varrho (\beta^*)\right|^2\,d\mu_{\overline{g}}=\sum_n\,\left|\int_{\Sigma }\,\varrho \,\Phi_{(n)}\,d\mu_{g}\right|^2\;,
\label{modexp2rho}
\end{equation}
\begin{equation}
\int_{\Sigma }\,\left|J(\beta^*)\right|^2\,d\mu_{\overline{g}}=\sum_n\,\left|\int_{\Sigma }\,J_{a}\,\Phi ^{a}_{(n)}\,d\mu_{g}\right|^2\;,
\label{modexp2J}
\end{equation}
\begin{equation}
\int_{\Sigma }\,\left|K(\beta^*)\right|^2\,d\mu_{\overline{g}}=\sum_n\,\left|\int_{\Sigma }\,K_{ij}\,\Phi ^{ij}_{(n)}\,d\mu_{g}\right|^2\;.
\label{modexp2}
\end{equation}
\end{theorem}
\begin{remark}
Note that the integral norms on the right side of the above relations, e.g.,  $\left|\int_{\Sigma }\,\varrho \,\Phi_{(n)}\,d\mu_{g}\right|^2$ only depend on the given initial, (for $\beta=0$), fields $\left(\varrho),\,J_a,\, K_{ab}\right)$, and on the geometry of the underlying Ricci flow $\beta\mapsto g(\beta)$, (via the backward flow  $\eta\mapsto \Phi ^{\;\sharp}_{(n)}$).
\end{remark}
\begin{proof}  We prove theorem \ref{LtwoK} explicitly for the $\beta$--evolution of the second fundamental form $\beta\mapsto K_{ab}(\beta)\in C^{\infty }(\Sigma ,\otimes ^{2}T^{*}\,\Sigma)$, the remaining cases for $\beta\mapsto \varrho(\beta)$ and $\beta\mapsto J_a(\beta)$ being similar.\\
\\
\noindent
From the evolution $\bigcirc_d^*\,\Phi ^{ab}_{(n)}(\eta)=0$, we compute
\begin{eqnarray}
&&\frac{\partial }{\partial \eta }\,|\Phi _{(n)}(\eta)|^2=\Delta\,|\Phi _{(n)}(\eta)|^2-2\,|\nabla \Phi _{(n)}(\eta)|^2-
\mathcal{R}(\eta)\,|\Phi _{(n)}(\eta)|^2\\
&& +\,4\,\Phi _{(n)}(\eta)\cdot\mathcal{R}iem(\beta)\cdot \Phi _{(n)}(\eta)\;,\nonumber
\end{eqnarray}
where $|\nabla \Phi _{(n)}(\eta)|^2:=\nabla^a\,\Phi ^{ij}_{(n)}\,\nabla\,\Phi_{ij} ^{(n)}(\eta)$, and $\Phi _{(n)}(\eta)\cdot  \mathcal{R}iem(\beta)\cdot \Phi _{(n)}(\eta):=\Phi ^{ij}_{(n)}\,\mathcal{R}_{ikjl}\;\Phi ^{kl}_{(n)}$. This implies that
\begin{eqnarray}
\label{normphi}
&&\frac{d}{d\eta}\,\int_{\Sigma }|\Phi _{(n)}(\eta)|^2\,d\mu_{g(\eta)}=-2\int_{\Sigma }|\nabla \Phi _{(n)}(\eta)|^2\,d\mu_{g(\eta)}\\
\nonumber\\
&&+ 4\,\int_{\Sigma }\Phi _{(n)}(\eta)\cdot\mathcal{R}iem(\eta)\cdot\Phi _{(n)}(\eta)\,d\mu_{g(\eta)}\;.\nonumber
\end{eqnarray}
This explicitly shows that, for $\eta>0$, the flows  $\left\{\Phi ^{\sharp}_{(n)}(\eta)\right\}\in C^{\infty }(\Sigma\times [0,\beta^*] ,\otimes ^{p}T\,\Sigma)$ do not preserve, in general, the orthonormality condition of the original $\left\{\overline{\Phi} ^{\;\sharp}_{(n)}\right\}$. 
To partial compensation of the lack of normalization, we can easily bound the $L^2$ norm of $\left\{\Phi ^{\sharp}_{(n)}(\eta)\right\}$. Since the Ricci flow $\beta\mapsto g_{ab}(\beta)$ on $\Sigma\times [0,\beta^*]$ is of bounded geometry, a direct application of the maximum principle to the reaction--diffusion equation governing the evolution of $|\mathcal{R}iem(g(\beta))|^2$ along $\beta\mapsto g_{ab}(\beta)$, implies the \emph{doubling time estimate}, (see \emph{e.g.}, \cite{chowluni}, lemma 6.1), according to which if $\left|\mathcal{R}iem(g(\beta=0))\right|\leq\, C_0$ then 
$\left|\mathcal{R}iem(g(\beta))\right|\leq\, 2\,C_0$, for all $0\leq\beta\leq 1/16\,C_0$.  Introducing this in (\ref{normphi}) we get
\begin{equation}
\int_{\Sigma }|\Phi _{(n)}(\eta)|^2\,d\mu_{g(\eta)}\,\leq\,e^{8\,C_0\,\eta}\,\,\left.\int_{\Sigma }|\Phi _{(n)}(\eta)|^2\,d\mu_{g}\right|_{\eta=0}
\,\leq\,e^{8\,C_0\,\eta}\;,
\label{boundphi}
\end{equation}
where we have exploited the orthonormality condition $\int_{\Sigma }|\Phi _{(n)}(\eta)|^2\,d\mu_{g}|_{\eta=0}=
\int_{\Sigma }|\overline{\Phi} _{(n)}(\eta)|^2\,d\mu_{\overline{g}}=1$.
Note that even without the doubling time estimate, (\emph{e.g.} if we run the interpolating length--scale $\beta$ over an interval $[0,\beta^*]$ such that $\beta^*>1/16\,C_0$), the hypothesis of bounded geometry implies that 
$\left|\mathcal{R}iem(g(\beta))\right|\leq\, C(\beta)$ on $\Sigma\times[0,\beta^*]$, for some $\beta$--depending constant 
$C(\beta)<\infty $. In such a case  we get the weaker estimates 
\begin{equation}
\int_{\Sigma }|\Phi _{(n)}(\eta)|^2\,d\mu_{g(\eta)}\,\leq\,e^{4\,\int_{0}^{\beta}C(s)\,ds}\;,
\label{boundphiweak}
\end{equation}
which suffices to control, in terms of the geometry of the underlying backward Ricci flow $\eta\mapsto g(\eta)$, the $L^2$--norm of the flows 
$\{\Phi ^{ab}_{(n)}\}$. It is also not difficult to check that  $\left\{\Phi ^{\sharp}_{(n)}(\eta)\right\}$ are not, for $\eta>0$, the eigentensors of the family of $\eta$--dependent elliptic operators $P_{d}(\eta)\doteq -\Delta _{d}+\mathcal{R}({g}(\eta))$. However, as we shall prove momentarily,  
the conjugacy between $\bigcirc_d\,\mathcal{C}(\beta)=0$ and $\bigcirc_d^*\,\left\{\Phi ^{\sharp}_{(n)}(\eta)\right\}=0$ preserves the Fourier coefficients $\forall \beta\in\,[0,\beta^*]$,  \emph{e.g.} 
\begin{equation}
\int_{\Sigma }\,K_{ij}(\beta)\,\Phi ^{ij}_{(n)}(\beta)\,d\mu_{g(\beta)}=\int_{\Sigma }\,K_{ij}\,\Phi ^{ij}_{(n)}(\beta^*)\,d\mu_{g}\;.
\label{modecoeff}
\end{equation}
\noindent
With these preliminary remarks along the way, let $K_{ab}(\beta^*)\in C^{\infty }(\Sigma ,\otimes ^{2}T^{*}\,\Sigma)$ be the evaluation, for $\beta=\beta^*$, of the flow $K_{ab}(\beta)\in C^{\infty }(\Sigma \times [0,\beta^*],\otimes _S^2T^*\Sigma )$,\; $\bigcirc_d\,K(\beta)=0$,\; $K(\beta=0)=K$.
Since the  set of eigentensors $\left\{\bar{\Phi} _{ik}^{(n)}\right\}$ provide  a complete orthonormal basis for $L^{2 }(\Sigma ,\otimes ^{2}T^{*}\,\Sigma;\,d\mu_{\bar{g}})$ and their linear span is dense in  $C^{\infty }(\Sigma,\otimes _S^2T^*\Sigma )$, the smoothness of  $K_{ab}(\beta^*)$ implies that  we can write
\begin{equation}
K_{ab}(\beta^*)=\,\sum_n\,\bar{\Phi} _{ab}^{(n)}\,\int_{\Sigma }\,K_{ij}(\beta^*)\,\bar{\Phi} ^{ij}_{(n)}\,d\mu_{\bar{g}}\;,
\end{equation}
where the series converges absolutely in the $C^{\infty }$ topology. 
Along the conjugate heat evolutions $\left\{\Phi _{ik}^{(n)}(\eta)\right\}\in C^{\infty }(\Sigma\times [0,\beta^*] ,\otimes ^{2}T^{*}\,\Sigma)$,\;
$\bigcirc^*_d\,\Phi _{ik}^{(n)}(\eta)=0$,\;$\Phi _{ik}^{(n)}(\eta=0):=\bar{\Phi} _{ik}^{(n)}$,\;$n\in \mathbb{N}$, and 
$K_{ab}(\beta)\in C^{\infty }(\Sigma \times [0,\beta^*],\otimes _S^2T^*\Sigma )$,\; $\bigcirc_d\,K(\beta)=0$,\; $K(\beta=0)=K$, 
we have (see (\ref{ldue}))
\begin{equation}
\frac{d}{d\beta }\,\int_{\Sigma }\,K_{ij}(\beta)\,{\Phi} ^{ij}_{(n)}(\beta)\,d\mu_{g(\beta)}=0\;,
\end{equation}
where ${\Phi} ^{ij}_{(n)}(\beta):={\Phi} ^{ij}_{(n)}(\eta=\beta^*-\beta)$. This imples (\ref{modecoeff}), and in particular (by evaluating the left member for $\beta=\beta^*$ and the right member for $\beta=0$), 
\begin{equation}
\int_{\Sigma }\,K_{ij}(\beta^*)\,\bar{\Phi} ^{ij}_{(n)}\,d\mu_{\bar{g}}
=\int_{\Sigma }\,K_{ij}\,{\Phi} ^{ij}_{(n)}(\beta^*)\,d\mu_{g}\;,
\end{equation}
which yields (\ref{modexp}). Under the stated smoothness hypotheses, (\ref{modexp2}) immediately follows from Parseval identity.
\end{proof}
\begin{figure}[h]
\includegraphics[scale=0.5]{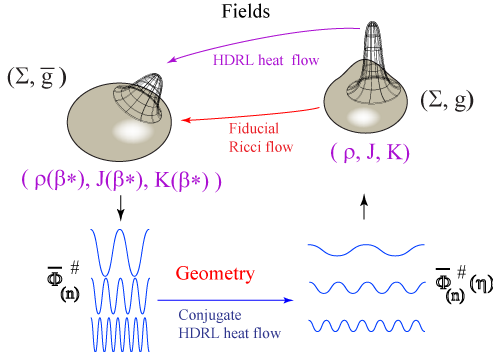}
\caption{The forward evolved fields $(\varrho(\beta^*),\,J(\beta^*),\,K(\beta^*))$ can be expanded in terms of the initial data $(\varrho,\,J,\,K)\in \mathcal{C}_{g}(\Sigma )$ with respect to the conjugated modes $\{\Phi^{\sharp }_{(n)}\}$ at $\beta=0$. Note that these backward propagated modes only depend on the geometry of the fiducial Ricci flow. Thus the $\{\Phi^{\sharp }_{(n)}(\eta)\}$ provide the \emph{geometrical directions} along which the fields $(\varrho,\,J,\,K)\in \mathcal{C}_{g}(\Sigma )$ do not dissipate in the $L^2$--sense.}
\label{fig: 17}   
\end{figure}

\noindent The above theorem can be applied to the Ricci tensor $\mathcal{R}ic(\beta)$ and the Ricci flow 
metric $g(\beta)$ itself. Indeed, by exploiting theorem  \ref{nondiss},  we directly get
\begin{lemma}
Let $\beta\mapsto g(\beta)$, $0\leq\beta\leq\beta^*$, be a Ricci flow of bounded geometry interpolating between $g\in {\mathcal{C}}_{{g}}(\Sigma )$ and $\bar{g}\in \overline{\mathcal{C}}_{\bar{g}}(\Sigma )$. If $\mathcal{R}ic(\bar{g})$ and  $\mathcal{R}ic(g)$ respectively denote the Ricci tensor of the metric 
$\bar{g}\in \overline{\mathcal{C}}_{\bar{g}}(\Sigma )$ and ${g}\in \mathcal{C}_{g}(\Sigma )$, then  
\begin{equation}
\mathcal{R}_{ab}(\bar{g})=\,\sum_n\,\bar{\Phi} _{ab}^{(n)}\,\int_{\Sigma }\,\mathcal{R}_{ij}(g)\,\Phi ^{ij}_{(n)}(\beta^*)\,d\mu_{g}\;,
\label{modexpRic}
\end{equation}
and 
\begin{equation}
\overline{g}_{ab}=\,\sum_n\,\bar{\Phi} _{ab}^{(n)}\,\int_{\Sigma }\,\left(g_{ij}-2\beta^*\,\mathcal{R}_{ij}\right)\,\Phi ^{ij}_{(n)}(\beta^*)\,d\mu_{g}\;.
\label{modexpRic2}
\end{equation}
\end{lemma}
\begin{proof}
The proof follows simply by noticing that the (non--linear) evolution of  $\mathcal{R}_{ab}(\beta)$ and $g_{ab}(\beta)-2(\beta^*-\beta)\,\mathcal{R}_{ab}(\beta)$,  along the underlying Ricci flow, is governed by $\bigcirc_d\,\mathcal{R}_{ab}(\beta)=0$ and 
$\bigcirc_d\,\left(g_{ab}(\beta)-2(\beta^*-\beta)\,\mathcal{R}_{ab}(\beta^*) \right)=0$, respectively. Thus $\mathcal{R}_{ab}(\beta)$ and $g_{ab}(\beta)-2(\beta^*-\beta)\,\mathcal{R}_{ab}(\beta)$ are smooth and  conjugated to the flows $\left\{\Phi _{ik}^{(n)}(\eta)\right\}$, $n\in\mathbb{N}$, and in analogy with (\ref{modexp}), we get the stated result.
\end{proof}

\noindent
There is a useful (somewhat tautological) rewriting of  theorem \ref{LtwoK} which better emphasizes the relation among the Ricci flow conjugated data  $\mathcal{C}_g(\Sigma)$  and  $\overline{\mathcal{C}}_{\overline{g}}(\Sigma)$. This relation will be further stressed later on when we will  introduce the heat kernel associated with the conjugate operator $\bigcirc_d^*\,$. 
\begin{lemma}
\label{tautlemma}
Let 
\begin{eqnarray}
\overline{\varrho}\,&=&\,\sum_n\,\overline{c}_n\left(\overline{\varrho}\right)\;\overline{\Phi}^{(n)}\;, \label{refexpan}\\
\overline{J}_{\;a}\,&=&\,\sum_n\,\overline{c}_n\left(\overline{J}\,\right)\;\overline{\Phi}_{\;a}^{(n)}\;, \nonumber\\
\overline{K}_{\;ab}\,&=&\,\sum_n\,\overline{c}_n\left(\overline{K}\,\right)\;\overline{\Phi}_{\;ab}^{(n)}\;,\nonumber 
\end{eqnarray}
the mode expansion on $(\Sigma,\,\overline{g})$ of the data $\in \overline{\mathcal{C}}_{\overline{g}}\;(\Sigma)$, where
\begin{eqnarray}
\overline{c}_n\left(\overline{\varrho}\,\right):&=&\,\int_{\Sigma }\,\overline{\varrho}\,\overline{\Phi}_{(n)}\,d\mu_{\overline{g}}\;,\\
\overline{c}_n\left(\overline{J}\,\right):&=&\,\int_{\Sigma }\,\overline{J}_{i}\,\overline{\Phi} ^{\;i}_{(n)}\,d\mu_{\overline{g}}\;,\\
\overline{c}_n\left(\overline{K}\,\right):&=&\,\int_{\Sigma }\,\overline{K}_{ij}\,\overline{\Phi} ^{\;ij}_{(n)}\,d\mu_{\overline{g}}\;.
\end{eqnarray}
Then, if we define
\begin{equation}
\delta \,\varrho_{(n)}:=\left[\int_{\Sigma }\,\varrho\;\Phi_{(n)}\,d\mu_{g}-\overline{c}_n\left(\overline{\varrho}\right)\right]\;,
\label{fluxprhoas}
\end{equation}
\begin{equation}
\delta \,J_{(n)}:=\left[\int_{\Sigma }\,J_{i}\;\Phi ^{i}_{(n)}\,d\mu_{g}-\overline{c}_n\left(\overline{J}\,\right)\right]\;,
\label{fluxpJas}
\end{equation}
\begin{equation}
\delta \,K_{(n)}:=\left[\int_{\Sigma }\,K_{ij}\,\Phi ^{ij}_{(n)}\,d\mu_{g}-\overline{c}_n\left(\overline{K}\,\right)\right]\;,
\label{fluxpas}
\end{equation}
we can write
\begin{equation}
\varrho(\beta^*)=\overline{\varrho}\,+\,\sum_n\,\overline{\Phi}^{\;(n)}\,\delta \,\varrho_{(n)}\;,
\label{modexprhoas}
\end{equation}
\begin{equation}
J_{a}(\beta^*)=\overline{J}_{\;a}\,+\,\sum_n\,\overline{\Phi} _{a}^{\;(n)}\,\delta \,J_{(n)}\;,
\label{modexpJas}
\end{equation}
\begin{equation}
K_{ab}(\beta^*)=\overline{K}_{\;ab}\,+\,\sum_n\,\overline{\Phi} _{ab}^{\;(n)}\,\delta \,K_{(n)}\;.
\label{modexpas}
\end{equation}
\end{lemma}
\begin{proof}
The lemma trivially follows by first adding and subtracting to the expressions for $\left(\varrho(\beta^*),\,J_{a}(\beta^*),\,K_{ab}(\beta^*)\right)$ the terms $\left(\overline{\varrho},\,\overline{J}^{\;a},\,\overline{K}^{\;ab}\right)$ and then expanding according to theorem \ref{LtwoK} and (\ref{refexpan}).
\end{proof}
\begin{figure}[h]
\includegraphics[scale=0.5]{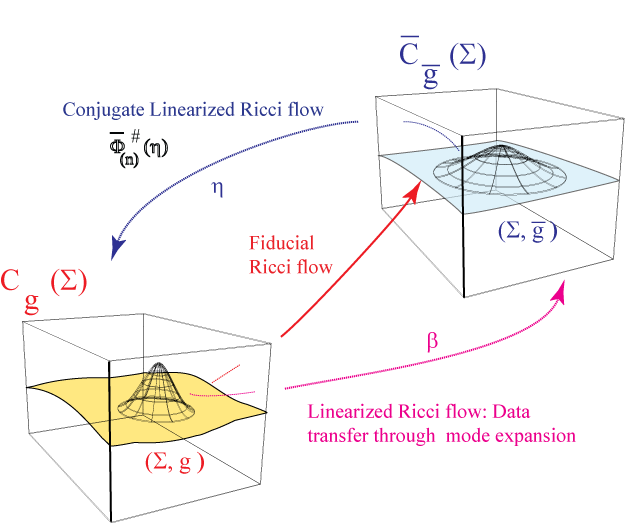}
\caption{By exploiting the mode expansion of the fluctuating fields (see (\ref{modexprhoas}), (\ref{modexpJas}), and (\ref{modexpas}))  we can compare  the initial data sets ${\mathcal{C}}_{{g}}(\Sigma )$ and $\bar{\mathcal{C}}_{\bar{g}}(\Sigma )$.}
\label{fig: 18}   
\end{figure}
\noindent
Roughly speaking, this lemma implies that the physical Einstein data $\in\,\mathcal{C}_g(\Sigma)$ evolved, along the fiducial Ricci flow, according to $\bigcirc_d\,\mathcal{C}(\beta)\,=0$, generate fields $\left(\varrho(\beta^*),\,J_{a}(\beta^*),\,K_{ab}(\beta^*)\right)$ which are expressible in terms of  reference Einstein data  $\overline{\mathcal{C}}_{\overline{g}}(\Sigma)$ plus fluctuation terms. These latter can be parametrized in terms of the eigen--modes $\overline{\Phi }^{\;\sharp}_{(n)}$ on the reference $(\Sigma ,\overline{g})$ and of their conjugate evolution along the given backward  Ricci flow.  This shows that Ricci flow conjugation is a rather natural procedure for comparing  the initial data sets ${\mathcal{C}}_{{g}}(\Sigma )$ and $\bar{\mathcal{C}}_{\bar{g}}(\Sigma )$. \\
\\
\noindent
Explicitly, we can rewrite the Hamiltonian and the divergence constraints
\begin{eqnarray}
\mathcal{R}(\overline{g})-\left(2\overline{\Lambda} + |\overline{K}\,|^{2}_{\overline{g}} -(tr_{\overline{g}}\,\overline{K}\,)^{2}\right) &=&16\pi \overline{\varrho} 
\;,  \label{refconstraint0} \\
\nonumber\\
2\,\nabla ^{a}\left(\overline{K}_{\;ab}-\overline{g}_{ab}\,(tr_{\overline{g}}\,\overline{K}\,)\right) &=&16\pi \overline{J}_{b}\;,  \label{refconstraints0}
\end{eqnarray}
which are assumed to hold for the reference data set $\overline{\mathcal{C}}_{\overline{g}}(\Sigma)$, in terms of the Ricci evolved physical data
$(\varrho(\beta^*), J_i(\beta^*), K_{ab}(\beta^*))$ and their fluctuations according to 
\begin{lemma}
\label{constrandfluct}
On the reference manifold $(\Sigma,\overline{g}\,)$, the Hamiltonian and divergence constraints (\ref{refconstraint0}) and (\ref{refconstraints0}) take the following form when expressed in terms of the Ricci evolved physical data
$\beta\mapsto (\varrho(\beta^*), J_i(\beta^*), K_{ab}(\beta^*))$ and of their fluctuations $(\delta \,\varrho_{(n)}, \delta\,J_{(n)}, \delta\,K_{(n)})$,
\begin{eqnarray}
&&\mathcal{R}(\overline{g})-\left[2\overline{\Lambda} + \left|{K}_{ab}(\beta^*)-\sum_n\,\overline{\Phi} _{ab}^{\;(n)}\,\delta \,K_{(n)}\right|^{2}_{\overline{g}}\right.\\
\nonumber\\
&&\left. -\left( \overline{g}^{\;ab}\,{K}_{ab}(\beta^*)-\sum_n\,\overline{g}^{\;ab}\,\overline{\Phi} _{ab}^{\;(n)}\,\delta \,K_{(n)} \,\right)^{2}\right]\nonumber
\nonumber\\
&&=\, 16\pi\,\varrho(\beta^*)- 16\pi\,\,\sum_n\,\overline{\Phi}^{\;(n)}\,\delta \,\varrho_{(n)}\;,\nonumber
\end{eqnarray}
and
\begin{eqnarray}
&&2\,\nabla^{a} \left[{K}_{ab}(\beta^*)-\sum_n\,\overline{\Phi} _{ab}^{\;(n)}\,\delta \,K_{(n)}\right]\\
\nonumber\\
&& -2\,\nabla _b\,\left[ \overline{g}^{\;cd}\,{K}_{cd}(\beta^*)-\sum_n\,\overline{g}^{\;cd}\,\overline{\Phi} _{cd}^{\;(n)}\,\delta \,K_{(n)} \right]\nonumber
\nonumber\\
&&=\, 16\pi\,J_b(\beta^*)- 16\pi\,\,\sum_n\,\overline{\Phi}^{\;(n)}_b\,\delta \,J_{(n)}\;.\nonumber
\end{eqnarray}
\end{lemma}
\begin{proof}
An obvious rewriting of (\ref{refconstraint0}) and (\ref{refconstraints0}) in terms of  (\ref{modexprhoas}), (\ref{modexpJas}), and (\ref{modexpas}).
\end{proof}
This lemma becomes a geometrically and physically significant statement if one can prove that  the Ricci evolution 
$\beta\mapsto (\varrho(\beta^*), J_i(\beta^*), K_{ab}(\beta^*))$ of the physical data $\mathcal{C}_g(\Sigma)$ entails a form of geometrical averaging controlling  the fluctuations $(\delta \,\varrho_{(n)}, \delta\,J_{(n)}, \delta\,K_{(n)})$, and complying with the dominant energy condition characterizing the given matter field $(\varrho,J)\in \mathcal{C}_g(\Sigma)$. In the next section we do actually prove that the Ricci evolution 
$\beta\mapsto (\varrho(\beta^*), J_i(\beta^*), K_{ab}(\beta^*))$ is in a technical sense a geometrical averaging as seen from the reference data 
$\overline{\mathcal{C}}_{\overline{g}}(\Sigma)$, at least for sufficiently small $\eta$. This is directly suggested by 
Theorem \ref{LtwoK} which indicates that the mode expansion formally behaves as a heat kernel for the operator  of $\bigcirc ^{*}_{d}$. This is indeed the case, and the averaging properties of Ricci flow conjugation become quite manifest when we consider the  the heat kernel of $\bigcirc ^{*}_{d}:=\frac{\partial }{\partial \eta }-\Delta _L+\mathcal{R}$, along the backward Ricci flow $\eta \mapsto g_{ab}(\eta )$.  

\section{Asymptotics for Ricci flow conjugated data}
\label{asympsec}
Let $\beta\mapsto (\Sigma , g_{ab}(\beta))$, $\beta\in [0,\beta^*]$ be the fiducial Ricci flow of bounded geometry interpolating between the two data sets $\mathcal{C}_g(\Sigma)$ and  $\overline{\mathcal{C}}_{\overline{g}}(\Sigma)$, and let $U_{\beta }\subset (\Sigma,g(\beta ))$ be a geodesically convex neighborhood 
containing the generic point $x\in\Sigma$. For a chosen base point $y\in U_{\beta }$, denote  by $l_{\beta }(y,x)$ the unique 
$g(\beta )$--geodesic segment $x=\exp_{y}\,u$,\;with $u\in T_{y}\Sigma $, connecting $y$ to $x$. Parallel transport along $l_{\beta }(y,x)$ 
allows to define a canonical isomorphism between the tangent space 
$T_{y}\Sigma $ and $T_{x}\Sigma $ which maps any given vector 
$\vec{v}(y)\in T_{y}\Sigma$ into a corresponding vector 
$\vec{v}_{P_{l_{\beta }(y,x)}}\in T_{x}\Sigma$. If $\{e_{(h)}(x) \}_{h=1,2,3}$ 
and $\{e_{(k')}(y) \}_{k'=1,2,3}$ respectively denote basis vectors in 
$T_{x}\Sigma $ and $T_{y}\Sigma$, (henceforth, primed indexes will always refer to components of elements of the tensorial algebra over $T_{y}\Sigma_{\beta } $), then the components of 
$\vec{v}_{P_{l_{\beta }(y,x)}}$ can be expressed as 
\begin{equation}
\left(v_{P_{l_{\beta }(y,x)}}\right)^{k}(x)= \tau ^{k}_{h'}(y,x;\beta )\,v^{h'}(y)\;,
\end{equation}
where $\tau ^{k}_{h'}$ $\in T\Sigma\boxtimes T^{*}\Sigma$  denotes the bitensor 
 associated with the parallel transport 
along $l_{\beta }(y,x)$. The Dirac $p$--tensorial measure  in 
$U_{\beta }\subset (\Sigma,g(\beta ))$ is defined according to
\begin{equation}
\delta ^{k_{1}\ldots k_{p}}_{h_{1}'\ldots h_{p}'}(y,x;\beta ):= \otimes _{(\alpha =1)}^{p}\,\tau ^{k_{\alpha }}_{h'_{\alpha }}(y,x;\beta )\,\,\delta_{\beta } (y,x)\; ,
\end{equation}
where $\delta_{\beta } (y,x)$ is the standard Dirac measure over the Riemannian manifold 
$(\Sigma_{\beta } ,g({\beta }))$ (see \cite{lichnerowicz}). With these notational remarks along the way we have 
\begin{theorem}
\label{HeatKernTh}
The flow $\beta\mapsto \mathcal{C}(\beta)$ admits, along the backward Ricci flow $\eta\mapsto (\Sigma , g_{ab}(\eta))$, $\eta\in [0,\beta^*]$, the $L^{2}(\Sigma\times [0,\beta^*],\;d\mu_{g(\eta)})$--averaging kernel 
\begin{equation}
\eta\longmapsto \mathbb{H}(x,y;\eta)\,\doteq\, \left( \begin{array}{ll}
         H(y,x;\eta)  \\
        H^{a}_{i'}(y,x;\eta)  \\
         H^{ab}_{i'k'}(y,x;\eta) \end{array} \right)\;,
\end{equation}
defined by the fundamental solution to the Hodge--DeRham--Lichnerowicz  conjugate heat equation
\begin{equation} 
\begin{tabular}{l}
$\left(\,\frac{\partial }{\partial \eta }-\,\Delta_{d}^{(x)}+\,\mathcal{R}\,\right)\;\mathbb{H}(y,x;\eta )=0\;,$\\
\\
$\lim_{\;\eta \searrow 0^{+}}\;\mathbb{H}(y,x;\eta )={\bf\delta }(y,x)\;,$%
\end{tabular}
\;   \label{fundcollect0}
\end{equation} 
where
\begin{equation}
{\bf\delta }(y,x)\;\doteq\;\left( \begin{array}{ll}
         \delta(y,x)  \\
        \delta^{a}_{i'}(y,x)  \\
         \delta^{ab}_{i'k'}(y,x) \end{array} \right)\;,
\end{equation}
is the corresponding p--tensorial Dirac measure. 
\end{theorem}
\begin{figure}[h]
\includegraphics[scale=0.4]{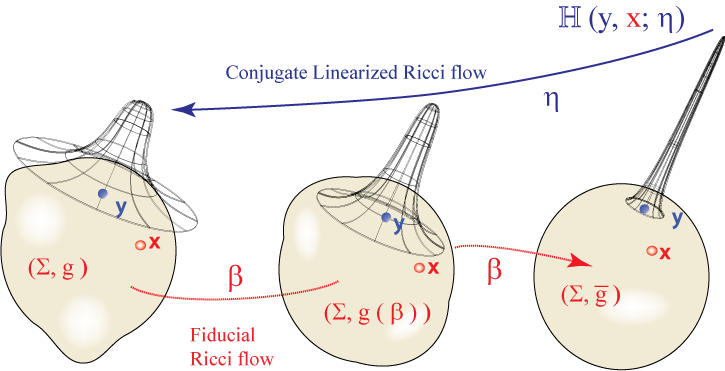}
\caption{The heat kernel $\mathbb{H}(y,x;\eta )$ of the conjugate linearized Ricci flow (in the HDRL representation) in the $\eta$--dependent geometry associated with the fiducial Ricci flow. }
\label{fig: 19}   
\end{figure}
\noindent

\begin{proof}
If $(\Sigma, g_{ab}(\eta ))$ is a smooth solution to the backward Ricci flow on $\Sigma_{\eta } \times [0,\beta ^{*}]$ with bounded curvature, then we can consider the $g(\eta )$--dependent  fundamental solution ${H}^{ab}_{i'k'}(y,x;\eta )$ to the conjugate heat equation (\ref{transVol}), \emph{i.e.},
\begin{equation} 
\begin{tabular}{l}
$\left(\,\frac{\partial }{\partial \eta }-\,\Delta_{L}^{(x)}+\,\mathcal{R}\,\right)\;{H}^{ab}_{i'k'}(y,x;\eta )=0\;,$\\
\\
$\lim_{\;\eta \searrow 0^{+}}\;{H}^{ab}_{i'k'}(y,x;\eta )={\delta }^{ab}_{i'k'}(y,x;)\;,$%
\end{tabular}
\;   \label{fund}
\end{equation} 
where $(y,x;\eta )\in (\Sigma \times \Sigma \backslash Diag(\Sigma \times \Sigma ))\times [0,\beta ^{*}]$, $\eta \doteq \beta ^{*}-\beta $,  $\Delta_{L}^{(x)}$ denotes the  Lichnerowicz--DeRham laplacian with respect to the variable $x$, and ${H}^{ab}_{i'k'}(y,x;\eta )$\,  is a smooth section of $(\otimes ^{2}T\Sigma)\boxtimes (\otimes ^{2}T^{*}\Sigma)$. The Dirac initial condition is understood in the distributional sense, \emph{i.e.}, for any smooth symmetric bilinear form with compact support $w^{i'k'}\in C^{\infty }_{0}(\Sigma ,\otimes ^{2}T\Sigma )$, 
\begin{equation}
\int_{\Sigma _{\eta }}{H}^{ab}_{i'k'}(y,x;\eta )\;w^{i'k'}(y)\;d\mu^{(y)} _{g(\eta )}\rightarrow w^{ab}(x)\;\;\;as\;\;\eta \searrow 0^{+}\;, 
\end{equation}
where the limit is meant in the uniform norm on  $C^{\infty }_{0}(\Sigma ,\otimes ^{2}T\Sigma )$. Since
 along a backward Ricci flow on $\Sigma_{\eta }\times [0,\beta ^{*}]$ with bounded geometry, the metrics $g_{ab}(\eta )$ are uniformly bounded above and below for $0\leq \eta \leq \beta ^{*}$, it does not really matter which metric we use in topologizing the spaces $C^{\infty }(\Sigma_{\eta } ,\otimes ^{2}T^{*}\,\Sigma_{\eta })$, and we can readily adapt to our setting the parametrix--deformation methods used in \cite{guenther2} and in  \cite{chowluni} to prove \cite{carfback} that
 along a backward Ricci flow on $\Sigma_{\eta }\times [0,\beta ^{*}]$, with bounded geometry, there exists a unique fundamental solution $\eta\longmapsto {H}^{ab}_{i'k'}(y,x;\eta )$ of the conjugate (Lichnerowicz) heat operator 
$\left(\,\frac{\partial }{\partial \eta }-\,\Delta_{L}^{(x)}+\,\mathcal{R}\,\right)$.  For the explicit (and rather lengthy) proof of this latter result and for the general properties of the integral kernel ${H}^{ab}_{i'k'}(y,x;\eta )$  we refer the reader to \cite{carfback}. Here we just need to recall its $\eta\searrow 0^+$ asymptotics, since this latter will be related to the explicit structure of the averaging we are considering. 
\noindent The  kernel ${H}^{ab}_{i'k'}(y,x;\eta )$ is singular as $\eta\searrow 0^+$, 
the general strategy for discussing its asymptotics is to model the corresponding parametrix around the Euclidean heat kernel $\left(4\pi \,\eta \right)^{-\frac{3}{2}}\,\exp\left(-\frac{d^{2}_{0}(y,x)}{4\eta } \right)$
defined in $T_{y}\Sigma$ by means of the exponential mapping associated with the initial manifold $(\Sigma ,g_{ab}(\eta =0)=\overline{g}_{ab})$. To this end,
denote by $d_{\eta }(y,x)$ the (locally Lipschitz)  distance function on $(\Sigma ,g_{ab}(\eta ))$ and by $inj\,(\Sigma ,g(\eta ))$ the associated injectivity radius. Adopt, with respect to the metric $g_{ab}(\eta )$, geodesic polar cordinates about $y\in \Sigma $, \emph{i.e.}, $x^{j'}=d_{\eta }(y,x)\,u^{j'}$, with $u^{j'}$ coordinates on the unit sphere $\mathbb{S}^{2}\subset T_{y}\Sigma $.
By adapting the analysis in \cite{{chowluni}}, \cite{lanconelli}, and \cite{gilkey3}, \cite{gilkey4} to (\ref{fund}) we have that, as $\eta \searrow 0^{+}$, and for all $(y,x)\in \Sigma $ such that $d_{0}(y,x)<\, inj\,(\Sigma ,g(0))$, there exists a sequence of smooth sections ${\Upsilon [h]\, }^{ab}_{i'k'}\,(y,x;\eta )$\, $\in C^{\infty }(\Sigma\times \Sigma ' ,\otimes ^{2}T\Sigma\boxtimes \otimes ^{2}T^{*}\Sigma)$,   with ${\Upsilon [0]\, }^{ab}_{i'k'}\,(y,x;\eta )={\tau }^{ab}_{i'k'}\,(y,x;\eta )$, such that
\begin{equation}
\frac{\exp\left(-\frac{d^{2}_{0}(y,x)}{4\eta } \right)}{\left(4\pi \,\eta \right)^{\frac{3}{2}}}\;\sum_{h=0}^{N}\eta ^{h}{\Upsilon [h]\, }^{ab}_{i'k'}\,(y,x;\eta )\;,
\label{uniasi}
\end{equation}
\noindent is uniformly asymptotic to ${H}^{ab}_{i'k'}(y,x;\eta )$, \;\, \emph{i.e.}, 
\begin{eqnarray}
 &&\left|{H}^{ab}_{i'k'}(y,x;\eta )- \frac{\exp\left(-\frac{d^{2}_{0}(y,x)}{4\eta } \right)}{\left(4\pi \,\eta \right)^{\frac{3}{2}}}\;\sum_{h=0}^{N}\eta ^{h}{\Upsilon [h]\, }^{ab}_{i'k'}\,(y,x;\eta )\right|_{\eta \searrow 0^{+}}\\
&&  =O\left(\eta ^{N-\frac{1}{2}} \right)\;,\nonumber
\end{eqnarray}
in the uniform norm on $C^{\infty }(\Sigma\times \Sigma ' ,\otimes ^{2}T\Sigma\times \otimes ^{2}T^{*}\Sigma)$.
A detailed presentation of the $\eta \searrow 0^{+}$ asymptotics of generalized Laplacians on vector bundles with time--varying geometries is discussed in \cite{gilkey3}, \cite{gilkey4}. \\
\\
\noindent
In analogy to  (\ref{fund}) let us  introduce the fundamental solutions $H(y,x;\eta)$ and $H^{a}_{i'}(y,x;\eta)$ of the scalar conjugate heat equation
$(\frac{\partial }{\partial \eta }-\Delta +\mathcal{R})$ and of the vector conjugate heat equation 
$(\frac{\partial }{\partial \eta }-\Delta_{vec} +\mathcal{R})$,\; (see (\ref{compact2})).  Then by defining
\begin{equation}
\mathbb{H}(x,y;\eta)\,\doteq\, \left( \begin{array}{ll}
         H(y,x;\eta)  \\
        H^{a}_{i'}(y,x;\eta)  \\
         H^{ab}_{i'k'}(y,x;\eta) \end{array} \right)\;,
\end{equation}
we can write in compact form
\begin{equation} 
\begin{tabular}{l}
$\left(\,\frac{\partial }{\partial \eta }-\,\Delta_{d}^{(x)}+\,\mathcal{R}\,\right)\;\mathbb{H}(y,x;\eta )=0\;,$\\
\\
$\lim_{\;\eta \searrow 0^{+}}\;\mathbb{H}(y,x;\eta )={\bf\delta }(y,x)\;,$%
\end{tabular}
\;   \label{fundcollect}
\end{equation} 
where
\begin{equation}
{\bf\delta }(y,x)\;\doteq\;\left( \begin{array}{ll}
         \delta(y,x)  \\
        \delta^{a}_{i'}(y,x)  \\
         \delta^{ab}_{i'k'}(y,x) \end{array} \right)\;,
\end{equation}
is the corresponding array of p--tensorial Dirac measures. In particular, it follows that the various asymptotic expansions of the fundamental solutions $H(y,x;\eta)$, $H^{a}_{i'}(y,x;\eta)$, and $H^{ab}_{i'k'}(y,x;\eta)$, which can be obtain in full analogy with (\ref{uniasi}), can be written
in a compact notation according to
\begin{eqnarray}
 &&\left|\mathbb{H}(y,x;\eta )- \frac{\exp\left(-\frac{d^{2}_{0}(y,x)}{4\eta } \right)}{\left(4\pi \,\eta \right)^{\frac{3}{2}}}\;\sum_{h=0}^{N}\eta ^{h}{{\bf{\Upsilon}} [h]\, }\,(y,x;\eta )\right|_{\eta \searrow 0^{+}}\label{compasy}\\
&&  =O\left(\eta ^{N-\frac{1}{2}} \right)\;,\nonumber
\end{eqnarray}
where ${\bf{\Upsilon}} [h]\,(y,x;\eta )$ is a collective notation for the appropriate set of sections characterizing the asymptotics of the various heat kernels involved.
\end{proof}
\begin{figure}[h]
\includegraphics[scale=0.5]{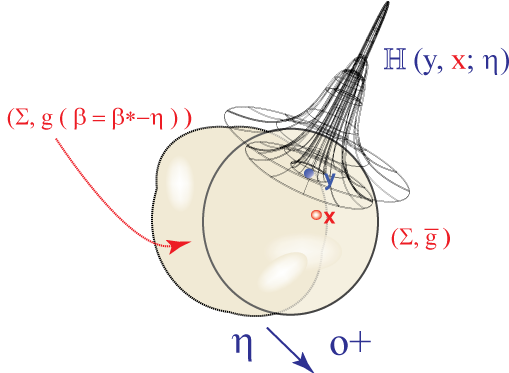}
\caption{The asymptotics of the heat kernel $\mathbb{H}(y,x;\eta )$. }
\label{fig: 20}   
\end{figure}
\noindent
Further details of heat kernels associated with a parameter--dependent metric are discussed  in \cite{guenther2}, \cite{{chowluni}}, (see Appendix A, \S 7 for a characterization of the parametrix of the heat kernel in such a case), and in a remarkable paper by N. Garofalo and E. Lanconelli \cite{lanconelli}. Strictly speaking, in all these works, the analysis is confined to the scalar laplacian, possibly with a potential term, but the theory readily extends to generalized laplacians, under the assumption that the metric $g_{ab}(\beta )$ is smooth as $\nearrow \beta ^{*}$. In particular, the case of generalized Laplacian on vector bundles with time--varying geometry has been studied in considerable detail by P. Gilkey and collaborators \cite{gilkey3}, \cite{gilkey4}. \\

\begin{remark}
The structure of the asymptotics (\ref{compasy}) of the heat kernel  $\mathbb{H}(y,x;\eta )$ directly shows that, at least for small $\eta$, the main contribution to $\mathbb{H}(y,x;\eta )$  comes from a neighborhhod of $y\in\Sigma$ consisting of all points $x\in \Sigma$ which, as measured in the \emph{reference} geometry $(\Sigma,\overline{g})$, are at a distance $d_0(y,x)\leq  2\,\sqrt{\eta}$. This remark implies that the integral kernel $\mathbb{H}(y,x;\eta )$ averages out over a length scale given by
\begin{equation}
\tau(\beta)\simeq 2\,\sqrt{\eta}=2\,\sqrt{\beta^*-\beta}\;.
\end{equation}
\end{remark}

\subsection{Asymptotics of the averaged data}

The averaging properties of $\mathbb{H}(y,x;\eta)$ are readily stated by exploiting the properties of the conjugated linearized Ricci flow. Explicitly, we get
\begin{theorem}
\label{avtheor}
Let $\eta \mapsto g_{ab}(\eta )$ be a backward Ricci flow with bounded geometry on $\Sigma_{\eta }\times [0,\beta ^{*}]$ and let ${H}^{ab}_{i'k'}(y,x;\eta )$ be the (backward) heat kernel of the corresponding conjugate linearized Ricci operator 
$\bigcirc ^{*}_{L}\, {H}^{ab}_{i'k'}(y,x;\eta )=0$, for  $\eta \in (0,\beta ^{*}]$, \, with ${H}^{ab}_{i'k'}(y,x;\eta\searrow 0^{+} )={\delta }^{ab}_{i'k'}(y,x)$. Then 
\begin{equation}
\mathcal{R}_{i'k'}(y,\eta=0 )=\int_{\Sigma }{H}^{ab}_{i'k'}(y,x;\eta )\,\mathcal{R}_{ab}(x,\eta )\,d\mu _{g(x,\eta )}
\;,
\label{avricci}
\end{equation}
for all $0\leq \eta \leq \beta ^{*}$.
Moreover, as $\eta\searrow 0^{+} $, we have the uniform asymptotic expansion
\begin{eqnarray}
&&\;\;\;\;\;\mathcal{R}_{i'k'}(y,\eta=0 )=\label{ricciasi}\\
\nonumber\\
&&\frac{1}{\left(4\pi \,\eta \right)^{\frac{3}{2}}}\,\int_{\Sigma }\exp\left(-\frac{d^{2}_{0}(y,x)}{4\eta } \right)\,{\tau }^{ab}_{i'k'}(y,x;\eta )\,\mathcal{R}_{ab}(x,\eta )\,d\mu _{g(x,\eta )}\nonumber\\
\nonumber\\
&&+\sum_{h=1}^{N}\frac{\eta ^{h}}{\left(4\pi \,\eta \right)^{\frac{3}{2}}}\,\int_{\Sigma }\exp\left(-\frac{d^{2}_{0}(y,x)}{4\eta } \right)\,{\Upsilon [h] }^{ab}_{i'k'}(y,x;\eta )\,\mathcal{R}_{ab}(x,\eta )\,d\mu _{g(x,\eta )}\nonumber\\
\nonumber\\
&&+O\left(\eta ^{N-\frac{1}{2}} \right)\nonumber\;,
\end{eqnarray}
\noindent where ${\tau }^{ab}_{i'k'}(y,x;\eta )$ $\in T\Sigma_{\eta }\boxtimes T^{*}\Sigma_{\eta } $  is the parallel transport operator
associated with  $(\Sigma ,g(\eta ))$,   $d_{0}(y,x)$ is the distance function in $(\Sigma ,g(\eta=0 ))$, and ${\Upsilon [h] }^{ab}_{i'k'}(y,x;\eta )$ are the smooth section  $\in C^{\infty }(\Sigma\times \Sigma ' ,\otimes ^{2}T\Sigma\boxtimes \otimes ^{2}T^{*}\Sigma)$, (depending on the geometry of $(\Sigma ,g(\eta ))$), characterizing the asymptotics of the heat kernel ${K}^{ab}_{i'k'}(y,x;\eta )$.
\label{princ1}
\end{theorem}
\begin{proof}
From proposition \ref{nondiss} we get that
along the backward Ricci flow on
$\Sigma \times [0,\beta ^{*}]$,  we can write, for all $0\leq \eta \leq \beta ^{*}$,
\begin{eqnarray}
\mathcal{R}_{i'k'}(y,\eta=0 )&=&\lim_{\,\eta \nearrow 0^{+}}\int_{\Sigma }{H}^{ab}_{i'k'}(y,x;\eta )\,\mathcal{R}_{ab}(x,\eta )\,d\mu _{g(\eta )}\label{avprop}\\
\nonumber\\
&=&\int_{\Sigma }{H}^{ab}_{i'k'}(y,x;\eta )\,\mathcal{R}_{ab}(x,\eta )\,d\mu _{g(\eta )}\;.\nonumber
\end{eqnarray}
Since the asymptotics (\ref{uniasi}) is uniform, we can integrate term by term, and by isolating the lower order term, we immediately get (\ref{ricciasi}). 
\end{proof}
As an illustrative example, let us consider the case in which $(\Sigma, \overline{g}_{ab})\in\overline{\mathcal{C}}_{\overline{g}}(\Sigma )$ is a manifold of constant curvature $C$,\; \emph{i.e.}, 
$\mathcal{R}_{i'k'}(y,\eta=0 )=2\,C\,\overline{g}_{i'k'}=\frac{\overline{\mathcal{R}}}{3}\,\overline{g}_{i'k'}$.  By tracing 
(\ref{avricci}) with respect to $\overline{g}^{i'k'}(y)$ we get
\begin{equation}
\overline{\mathcal{R}}(y)=\int_{\Sigma }\overline{g}^{i'k'}(y)\,{H}^{ab}_{i'k'}(y,x;\eta )\,\mathcal{R}_{ab}(x,\eta )\,d\mu _{g(x,\eta )}
\;,
\label{avscal}
\end{equation}
which nicely shows that the scalar curvature $\overline{\mathcal{R}}(y)$ of $(\Sigma, \overline{g}_{ab})$ is obtained upon averaging the Ricci curvature of the data  $\mathcal{C}_{g}(\Sigma )$ along the interpolating Ricci flow $\beta\longmapsto g_{ab}(\beta)$. This is even more explicitly seen from the asymptotics (\ref{ricciasi}). Indeed,  by tracing 
(\ref{ricciasi}) with respect to $\overline{g}^{i'k'}(y)$, and taking into account that, at order $\mathcal{O}(\eta^{\frac{1}{2}})$, we can write $\overline{g}^{i'k'}(y){\tau }^{ab}_{i'k'}(y,x;\eta )\,\mathcal{R}_{ab}(x,\eta )\simeq g^{ab}(x,\eta)\mathcal{R}_{ab}(x,\eta )=\mathcal{R}(x,\eta )$, we get
\begin{eqnarray}
&&\;\;\;\;\;\overline{\mathcal{R}}(y)=\label{scalasi}\\
\nonumber\\
&&\frac{1}{\left(4\pi \,\eta \right)^{\frac{3}{2}}}\,\int_{\Sigma }\exp\left(-\frac{d^{2}_{0}(y,x)}{4\eta } \right)\,\mathcal{R}(x,\eta )\,d\mu _{g(x,\eta )}\nonumber\\
\nonumber\\
&&+\sum_{h=1}^{N}\frac{\eta ^{h}}{\left(4\pi \,\eta \right)^{\frac{3}{2}}}\,\int_{\Sigma }e^{\left(-\frac{d^{2}_{0}(y,x)}{4\eta } \right)}\,
\overline{g}^{i'k'}{\Upsilon [h] }^{ab}_{i'k'}(y,x;\eta )\,\mathcal{R}_{ab}(x,\eta )\,d\mu _{g(x,\eta )}\nonumber\\
\nonumber\\
&&+O\left(\eta ^{N-\frac{1}{2}} \right)\nonumber\;.
\end{eqnarray}
\begin{figure}[h]
\includegraphics[scale=0.5]{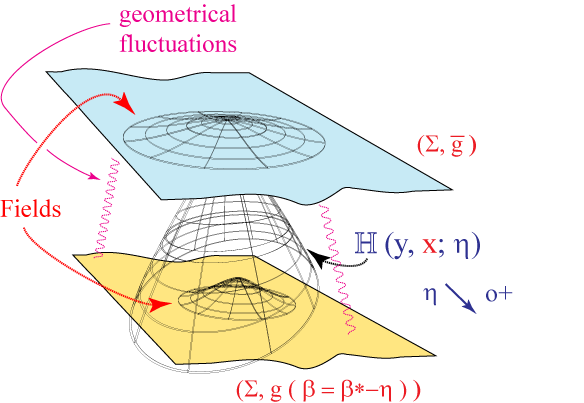}
\caption{The asymptotics of the averaged data (the fields) is, for small $\eta$, a form of Gaussian averaging dressed by geometrical fluctuations. }
\label{fig: 21}   
\end{figure}
\noindent
By the very definition of Ricci flow conjugation, (see \emph{e.g.} (\ref{hevol} )), it follows that a representation structurally similar to (\ref{avricci}) and (\ref{ricciasi}) holds also for the extrinsic curvature flow $\beta\mapsto {K}_{ab}(\beta)$, solution of the linearized Ricci flow
$\bigcirc_L\,K_{ab}(\beta)=0$, \emph{i.e.},
\begin{equation}
{K}_{i'k'}(y,\eta=0 )=\int_{\Sigma }{H}^{ab}_{i'k'}(y,x;\eta )\,{K}_{ab}(x,\eta )\,d\mu _{g(x,\eta )}
\;,
\label{avextrinsic}
\end{equation}
for all $0\leq \eta \leq \beta ^{*}$, and
\begin{eqnarray}
&&\;\;\;\;\;{K}_{i'k'}(y,\eta=0 )=\label{accaasi}\\
\nonumber\\
&&\frac{1}{\left(4\pi \,\eta \right)^{\frac{3}{2}}}\,\int_{\Sigma }\exp\left(-\frac{d^{2}_{0}(y,x)}{4\eta } \right)\,{\tau }^{ab}_{i'k'}(y,x;\eta )\,{K}_{ab}(x,\eta )\,d\mu _{g(x,\eta )}\nonumber\\
\nonumber\\
&&+\sum_{h=1}^{N}\frac{\eta ^{h}}{\left(4\pi \,\eta \right)^{\frac{3}{2}}}\,\int_{\Sigma }\exp\left(-\frac{d^{2}_{0}(y,x)}{4\eta } \right)\,{\Upsilon [h] }^{ab}_{i'k'}(y,x;\eta )\,{K}_{ab}(x,\eta )\,d\mu _{g(x,\eta )}\nonumber\\
\nonumber\\
&&+O\left(\eta ^{N-\frac{1}{2}} \right)\nonumber\;.
\end{eqnarray}
Again by tracing with respect to $\overline{g}^{i'k'}(y)$ we get 
\begin{equation}
{k}(y,\eta=0 )=\int_{\Sigma }\overline{g}^{i'h'}(y)\,{H}^{ab}_{i'h'}(y,x;\eta )\,{K}_{ab}(x,\eta )\,d\mu _{g(x,\eta )}
\;,
\label{avemean}
\end{equation}
 and
\begin{eqnarray}
&&\;\;\;\;\;{k}(y,\eta=0 )=\label{meanasi}\\
\nonumber\\
&&\frac{1}{\left(4\pi \,\eta \right)^{\frac{3}{2}}}\,\int_{\Sigma }\exp\left(-\frac{d^{2}_{0}(y,x)}{4\eta } \right)
\,{k}(x,\eta )\,d\mu _{g(x,\eta )}\nonumber\\
\nonumber\\
&&+\sum_{h=1}^{N}\frac{\eta ^{h}}{\left(4\pi \,\eta \right)^{\frac{3}{2}}}\,\int_{\Sigma }e^{\left(-\frac{d^{2}_{0}(y,x)}{4\eta } \right)}\,
\overline{g}^{i'k'}{\Upsilon [h] }^{ab}_{i'k'}(y,x;\eta )\,{K}_{ab}(x,\eta )\,d\mu _{g(x,\eta )}\nonumber\\
\nonumber\\
&&+O\left(\eta ^{N-\frac{1}{2}} \right)\nonumber\;.
\end{eqnarray}
\noindent 
 Since 
$\lim_{\,\eta \searrow 0^{+}}\int_{\Sigma }{H}^{ab}_{i'k'}(y,x;\eta )\,{g}_{ab}(x,\eta )\,d\mu _{g(\eta )}$\,$=$\, ${g}_{i'k'}(y,\eta=0)$, the conservation law (\ref{bello2}), applied to ${H}^{ab}_{i'k'}(y,x;\,\eta))$, directly provides the
\begin{lemma}
Let $\beta \mapsto g_{ab}(\beta )$ be a Ricci flow with bounded geometry on $\Sigma_{\beta }\times [0,\beta ^{*}]$, and let ${H}^{ab}_{i'k'}(y,x;\eta )$ be the (backward) heat kernel of the corresponding conjugate linearized Ricci operator 
$\bigcirc ^{*}_{L}$, for  $\eta=\beta ^{*}-\beta $. Then, along the backward flow $\eta \mapsto g_{ab}(\eta )$,
\begin{equation}
 \overline{g}_{i'k'}\,(y,\eta=0)=\int_{\Sigma }{H}^{ab}_{i'k'}(y,x;\eta)\,\left[{g}_{ab}(x,\eta)-2\eta \,\,\mathcal{R}_{ab}(x,\eta)\right]\,d\mu _{g(x,\eta )}\;,
\label{grepres0}
\end{equation}
\noindent for all $0\leq \eta \leq \beta ^{*}$, and 
\begin{eqnarray}
&&\;\;\;\;\;\;\;\;\overline{g}_{i'k'}(y,\eta=0)=\\
\nonumber\\
&&\frac{1}{\left(4\pi \,\eta \right)^{\frac{3}{2}}}\,\int_{\Sigma }e^{-\frac{d^{2}_{0}(y,x)}{4\eta }}\,{\tau }^{ab}_{i'k'}(y,x;\eta )\,\left[{g}_{ab}(x,\eta)-2\eta\,\mathcal{R}_{ab}(x,\eta)\right]\,d\mu _{g(x,\eta)}     \nonumber\\
\nonumber\\
&&+\sum_{h=1}^{N}\frac{\eta ^{h}}{\left(4\pi \,\eta \right)^{\frac{3}{2}}}\,\int_{\Sigma }e^{-\frac{d^{2}_{0}(y,x)}{4\eta }}\,{\Upsilon [h] }^{ab}_{i'k'}(y,x;\eta )\left[{g}_{ab}(x,\eta)-2\eta\,\mathcal{R}_{ab}(x,\eta)\right]\,d\mu _{g(x,\eta)} \nonumber\\
\nonumber\\
&&+O\left(\eta ^{N-\frac{1}{2}} \right)\nonumber\;.
\end{eqnarray}
holds uniformly, as $\eta \searrow 0^{+}$.
\end{lemma}
\noindent
From the scalar part of the conjugate heat kernel $\mathbb{H}(y,x;\eta )$ solution of (\ref{fundcollect}) we get, for the matter density 
flow $\beta\longmapsto \varrho (\beta)$,
\begin{equation}
{\varrho }(y,\eta=0 )=\int_{\Sigma }\,{H}(y,x;\eta )\,{\varrho }(x,\eta )\,d\mu _{g(x,\eta )}
\;,
\label{avematt}
\end{equation}
where, as usual, $\eta:=\beta^*-\beta$, and
\begin{eqnarray}
&&{\varrho }(y,\eta=0 )=\frac{1}{\left(4\pi \,\eta \right)^{\frac{3}{2}}}\,\int_{\Sigma }\exp\left(-\frac{d^{2}_{0}(y,x)}{4\eta } \right)
\,{\varrho }(x,\eta )\,d\mu _{g(x,\eta )}\label{mattasi}\\
\nonumber\\
&&+\sum_{h=1}^{N}\frac{\eta ^{h}}{\left(4\pi \,\eta \right)^{\frac{3}{2}}}\,\int_{\Sigma }e^{\left(-\frac{d^{2}_{0}(y,x)}{4\eta } \right)}\,
{\Upsilon [h] }(y,x;\eta )\,{\varrho }(x,\eta )\,d\mu _{g(x,\eta )}\nonumber\\
\nonumber\\
&&+O\left(\eta ^{N-\frac{1}{2}} \right)\nonumber\;.
\end{eqnarray}
Since ${H}(y,x;\eta )$ solves the conjugate heat equation (\ref{fundcollect0}), the relation (\ref{avematt}) is a statement of mass conservation in the averaging region supporting, as $\eta$ varies, the probability measure ${H}(y,x;\eta )\,d\mu _{g(x,\eta )}$, (again by the defining  conjugacy relation, it is immediate to verify that $\int_{\Sigma }\,{H}(y,x;\eta )\,d\mu _{g(x,\eta )}=1$, $\forall \eta\in[0,\beta^*]$). \\

\noindent
Finally, for the matter current density  $J_{i}\in\mathcal{C}^{P}_{g}(\Sigma )$,  
(evolving according to $\bigcirc_d,J_{i}=0$--see def. \ref{theDEF}), we get
\begin{equation}
{J}_{i'}(y,\eta=0 )=\int_{\Sigma }{H}^{a}_{i'}(y,x;\eta )\,{J}_{a}(x,\eta )\,d\mu _{g(x,\eta )}
\;,
\label{avcurreny}
\end{equation}
for all $0\leq \eta \leq \beta ^{*}$, and
\begin{eqnarray}
&&\;\;\;\;\;{J}_{i'}(y,\eta=0 )=\label{currentasi}\\
\nonumber\\
&&\frac{1}{\left(4\pi \,\eta \right)^{\frac{3}{2}}}\,\int_{\Sigma }\exp\left(-\frac{d^{2}_{0}(y,x)}{4\eta } \right)\,{\tau }^{a}_{i'}(y,x;\eta )\,{J}_{a}(x,\eta )\,d\mu _{g(x,\eta )}\nonumber\\
\nonumber\\
&&+\sum_{h=1}^{N}\frac{\eta ^{h}}{\left(4\pi \,\eta \right)^{\frac{3}{2}}}\,\int_{\Sigma }\exp\left(-\frac{d^{2}_{0}(y,x)}{4\eta } \right)\,{\Upsilon [h] }^{a}_{i'}(y,x;\eta )\,{J}_{a}(x,\eta )\,d\mu _{g(x,\eta )}\nonumber\\
\nonumber\\
&&+O\left(\eta ^{N-\frac{1}{2}} \right)\nonumber\;.
\end{eqnarray}
\noindent 
We can exploit the above asymptotics for giving a rather convenient representation, as $\eta\searrow 0^+$, of the fluctuations $(\delta \,\varrho_{(n)}, \delta\,J_{(n)}, \delta\,K_{(n)})$. For instance, in the case of the matter density, by comparing (\ref{mattasi}) with the corresponding expression (\ref{modexprhoas}) in lemma \ref{tautlemma},  we can write 
\begin{eqnarray}
&&\varrho(\beta^*,y)=\overline{\varrho}(y)+\sum_n\,\overline{\Phi}^{\;(n)}\,\delta \,\varrho_{(n)}\,\label{mattasifluct}\\
\nonumber\\
&&=\overline{\varrho}(y)\,+\, \frac{1}{\left(4\pi \,\eta \right)^{\frac{3}{2}}}\,\int_{\Sigma }\exp\left(-\frac{d^{2}_{0}(y,x)}{4\eta } \right)
\,\left[{\varrho }(x,\eta )-\overline{\varrho }(y)\right]\,d\mu _{g(x,\eta )}\nonumber\\
\nonumber\\
&&+\sum_{h=1}^{N}\frac{\eta ^{h}}{\left(4\pi \,\eta \right)^{\frac{3}{2}}}\,\int_{\Sigma }e^{\left(-\frac{d^{2}_{0}(y,x)}{4\eta } \right)}\,
{\Upsilon [h] }(y,x;\eta )\,\left[{\varrho }(x,\eta )-\overline{\varrho }(y)\right]\,d\mu _{g(x,\eta )}\nonumber\\
\nonumber\\
&&+O\left(\eta ^{N-\frac{1}{2}} \right)\nonumber\;.
\end{eqnarray}
Similar expressions can be easily written down for $\sum_n\,\overline{\Phi}^{\;(n)}_a\,\delta \,J_{(n)}$, 
and $\sum_n\,\overline{\Phi}^{\;(n)}_{ab}\,\delta \,K_{(n)}$, and clearly show that, at least for small $\eta$, Ricci flow conjugation is an averaging procedure as suggested by  the  spectral resolution described in Lemma \ref{tautlemma}.\\

\subsection{The Matter--Geometry content of $\mathcal{C}(\beta^*)$ and its asymptotics}

The actual computation of the sections ${\Upsilon [h]^{\;ab}_{\;i'k'} }(y,x;\eta )$ is, in general, quite demanding and the above asymptotic expansions are mostly of theoretical rather than practical value in most situations. A more useful result can be obtained if, rather than looking at the pointwise expressions for the deformed data $\mathcal{C}(\beta^*)$, we consider the following integral quantities:\\
\\
\noindent \emph{(i)} The matter content  of $\mathcal{C}(\beta^*)$ 
with respect to $\overline{\mathcal{C}}^{\;\sharp }\,(\eta)$
\begin{equation}
\mathfrak{M}(\eta):=
\int_{\Sigma }\,\varrho(y,\beta^*)\,\overline{\varrho}(y,\eta) \,d\mu_{{g}(y,\eta)}\;;
\end{equation}
\emph{(ii)} The current content of $\mathcal{C}(\beta^*)$ 
with respect to $\overline{\mathcal{C}}^{\;\sharp }\,(\eta)$
\begin{equation}
\mathfrak{J}(\eta):=\int_{\Sigma }J_i(y,\beta^*)\,\overline{J}^{\,i}(y,\eta)\,d\mu_{{g}(y,\eta)}\;;
\end{equation}
\emph{(ii)} The extrinsic curvature content  of $\mathcal{C}(\beta^*)$ 
with respect to $\overline{\mathcal{C}}^{\;\sharp }\,(\eta)$
\begin{equation}
\mathfrak{K}(\eta):= \int_{\Sigma }K_{ab}(y,\beta^*)\,\overline{K}^{\,ab}(y,\eta)\,d\mu_{{g}(y,\eta)}\;. 
\end{equation}
By introducing the compact notation 
\begin{equation}
\mathfrak{MG}(\eta):=\int_{\Sigma }\,\mathcal{C}(y,\beta^*)\odot\, \overline{\mathcal{C}}^{\;\sharp }\,(y,\eta)\,d\mu_{{g}(y,\eta)}:=\,\left( \begin{array}{ll}
        \mathfrak{M}(\eta)  \\
        \mathfrak{J}(\eta) \\
        \mathfrak{K}(\eta)   \end{array} \right)\;,
\end{equation} 
we collectively refer to the above expressions as defining the \emph{Matter--Geometry content}  of $\mathcal{C}(\beta^*)$ 
with respect to $\overline{\mathcal{C}}^{\;\sharp }\,(\eta)$. Note that whereas $\mathfrak{MG}(\eta=0)$ is a conserved quantity along the interpolating Ricci flow, in general $\mathfrak{MG}(\eta)$  is not. Making a parallel with heat propagation, the integrals defining  $\mathfrak{MG}(\eta)$ play the role of the \emph{heat content} of a system characterized by a  distribution given by $\mathcal{C}(\beta^*-\eta)$ and by an $\eta$--dependent \emph{specific heat} proportional to $\overline{\mathcal{C}}(\eta)$. It provides, as $\eta$ varies, a relevant physical quantity which through  $\mathfrak{MG}(\eta)-\mathfrak{MG}(0)$ can be conveniently used to describe, at least for small $\eta$, the fluctuations of the averaged physical data  $\mathcal{C}(\beta^*-\eta)$ with respect to the reference data. To show that this is indeed the case, let us note that in terms of the heat kernel $\mathbb{H}(y,x;\,\eta)$ we can write 
\begin{equation}
\mathcal{C}(y,\beta^*)\,=\, \int_{\Sigma }\,\mathbb{H}(y,x;\,\eta)\,\mathcal{C}\,(x,\eta)\,d\mu_{{g}(x,\eta)}\;.
\end{equation}
Thus 
\begin{equation}
\mathfrak{MG}(\eta)=
\int\int_{\Sigma }\,\mathbb{H}(y,x;\,\eta)\,
\mathcal{C}\,(x,\eta)\odot \overline{\mathcal{C}}^{\;\sharp }(y,\eta)\,d\mu_{{g}(x,\eta)}\,d\mu_{g(y,\eta)};.
\end{equation}
This expression has the structure of the heat content (in the sense of P. Gilkey \cite{gilkey3, gilkey3a}) in a time ($\eta$) variable geometry with \emph{specific heat} given  \footnote{Actually, in order to compare with the structure theorems in
 \cite{gilkey3}, the role of the specific heat should be played by the expression 
$\overline{\mathcal{C}}^{\;\sharp }(\eta)\,\sqrt{\det\,g(\eta)}/\sqrt{\det\,\overline{g}}$. However, since we are interested in the small
 $\eta$ asymptotics, we can equivalently use $\overline{\mathcal{C}}^{\;\sharp }(\eta)$.} by $\overline{\mathcal{C}}^{\;\sharp }(y,\eta)$.
In particular, if we specialize the results of \cite{gilkey3, gilkey3a} to the case of the  Lichnerowicz--Hodge--DeRham heat flow discussed here we get the
\begin{theorem}
\label{Giltheo}
Let $\Delta _{d}^{(\beta(\eta))}:= -(d\,\delta_{g(\beta^*-\eta)}+\delta_{g(\beta^*-\eta)}\,d)$ denote the Hodge Laplacian, with respect to the backward Ricci evolving metric $\eta\longmapsto g(\eta)$, thought of as acting on the generic section $W\in C^{\infty }(\Sigma \times [0,\beta^*],\,\otimes ^p_S\,T^*\Sigma)$, $p=0,1,2$. For $\eta\in[0,\beta^*]$ small, let
\begin{equation}  
\Delta _{d}^{(\beta(\eta))}\,W\sim \overline{\Delta} _{d}\, W+\,\eta\left\{ \overline{A}^{\,ab}\,\overline{\nabla }_a \overline{\nabla }_b\,W+
\overline{B}^{\,b}\overline{\nabla }_b\,W+\overline{E}\,W \right\} + o(\eta^2)
\label{Taylorexp}
\end{equation}
be the first order Taylor expansion of $\Delta _{d}^{(\beta(\eta))}$ around $g_{ik}(\beta^*)=\overline{g}_{ik}\in \overline{\mathcal{C}}_{\overline{g}}(\Sigma)$, with $\overline{\Delta} _{d}$ and  $\overline{\nabla }_a$ respectively denoting the Hodge Laplacian and the Levi--Civita connection on $(\Sigma,\overline{g})$, and where the coefficients $\overline{A}^{\,ab}$, $\overline{B}^{\,b}$, 
$\overline{E}$ are geometrical quantities constructed with the Riemann and the Ricci tensor of $(\Sigma,\overline{g})$.
With these preliminary remarks along the way,
the \emph{Matter--Geometry content}  of $\mathcal{C}(\beta^*)$ with respect to $\overline{\mathcal{C}}^{\;\sharp }\,(\eta)$ admits, for $\eta\searrow 0^+$, the asymptotic expansion
\begin{equation}
\mathfrak{MG}(\eta)\sim \sum_{n=0}^{\infty }\,\mathfrak{B}_n(\mathcal{C},\overline{\mathcal{C}})\,\eta^{n/2}\;,
\end{equation}
where the coefficients $\mathfrak{B}_n(\mathcal{C},\overline{\mathcal{C}})$ are all $0$ for $n$ odd, and where the coefficients for $n=0,2, 4$ are provided by 
\begin{equation}
\mathfrak{B}_0(\mathcal{C},\overline{\mathcal{C}})=\,
\int_{\Sigma }\,\mathcal{C}(\beta^*)\odot\, \overline{\mathcal{C}}^{\;\sharp }\,d\mu_{\overline{g}}\;,
\end{equation}
\begin{equation}
\mathfrak{B}_2(\mathcal{C},\overline{\mathcal{C}})=\,
-\,\int_{\Sigma }\,\overline{\Delta} _d\,\mathcal{C}(\beta^*)\odot\, \overline{\mathcal{C}}^{\;\sharp }\,d\mu_{\overline{g}}\;,
\end{equation}
\begin{eqnarray}
&&\;\;\;\;\;\;\;\;\;\;\;\;\;\;\;\mathfrak{B}_4(\mathcal{C},\overline{\mathcal{C}})=\,\frac{1}{2} \int_{\Sigma }\,\overline{\Delta} _d\,\mathcal{C}(\beta^*)\odot\, 
\overline{\Delta} _d\;\overline{\mathcal{C}}^{\;\sharp }\,d\mu_{\overline{g}}\\
\nonumber\\
&&-\,\frac{1}{2} \int_{\Sigma }\,\left[ \overline{A}^{\,ab}\,\overline{\nabla }_a \overline{\nabla }_b\,\mathcal{C}(\beta^*)+
\overline{B}^{\,b}\overline{\nabla }_b\,\mathcal{C}(\beta^*)+\overline{E}\,\mathcal{C}(\beta^*) \right]
\odot\, \overline{\mathcal{C}}^{\;\sharp }\,d\mu_{\overline{g}}\;. \nonumber
\end{eqnarray}
\end{theorem}
\noindent Note that explicit formulae for $\mathfrak{B}_n(\mathcal{C},\overline{\mathcal{C}})$, with $n\geq 6$, are in general, at the time of writing, not known.

\begin{proof}
The theorem is a direct application of Gilkey's analysis of the heat content asymptotics for the heat propagation generated by Laplace type operators in time dependent geometries. In particular one can apply theorem 9.2 in \cite{gilkey3a}, (this is stated for the more general case of the heat type operators evolving in domains with Dirichlet boundary conditions, where both the geometry as well as the boundaries are time dependent). 
\end{proof}
As an illustrative example we start working out the asymptotics for the matter content $\mathfrak{M}(\eta)$. In such a case, $\Delta _{d}^{(\beta(\eta))}$ reduces to the Laplace--Beltrami operator $\Delta^{(\beta(\eta))}$ on $(\Sigma, g(\beta^*-\eta))$, and one easily computes
\begin{equation}
\Delta^{(\beta(\eta))}\sim \overline{\Delta }\,-\,2\eta\, \overline{\mathcal{R}}^{\;ab}\,\overline{\nabla }_a \overline{\nabla }_b\,+ o(\eta^2)\;,
\end{equation}
where we have exploited the relation
\begin{equation}
\frac{\partial }{\partial \beta}\,\Delta^{(\beta)}\,=\,2\mathcal{R}^{ab}\,\nabla_a \nabla_b\;, 
\end{equation} 
which holds along the Ricci flow (see \emph{e.g.} \cite{chowluni}). From the above theorem we easily get
\begin{eqnarray}
&&\mathfrak{M}(\eta):=
\int_{\Sigma }\,\varrho(\beta^*)\,\overline{\varrho}(\eta) \,d\mu_{{g}(\eta)}\,=
\int_{\Sigma }\,\varrho(\beta^*)\,\overline{\varrho} \,d\mu_{\overline{g}}\,\\
&&-\,\eta\,\int_{\Sigma }\,\overline{\varrho} \,\overline{\Delta }\,\varrho(\beta^*)\,d\mu_{\overline{g}}
\,+\, \eta^4\, \int_{\Sigma }\,\overline{\mathcal{R}}^{\;ab}\,\overline{\nabla}_a\overline{\nabla}_b\,\varrho(\beta^*)\,\overline{\varrho} \,d\mu_{\overline{g}}\,+\ldots\;.\nonumber
\end{eqnarray}
It is clear that the main computational burden in writing down similar expressions for the full matter--geometry content lies in evaluating the $\mathfrak{B}_4(\mathcal{C},\overline{\mathcal{C}})$ coefficient which requires the Taylor expansion (\ref{Taylorexp}) of the vector and of the Lichnerowicz Laplacian  $\Delta_d^{(\beta(\eta))}$. While this does not present particular difficulties, the resulting expressions are long and not particularly illuminating, thus we simply write down the expansions up to the obvious $\mathfrak{B}_2(\mathcal{C},\overline{\mathcal{C}})$ term providing the relevant order--$\eta$ contribution to the  matter--geometry content. We get
\begin{eqnarray}
&&\mathfrak{J}(\eta):=
\int_{\Sigma }\,J_i(\beta^*)\,\overline{J}^{\,i}(\eta) \,d\mu_{{g}(\eta)}\,=
\int_{\Sigma }\,J_i(\beta^*)\,\overline{J}^{\,i} \,d\mu_{\overline{g}}\,\\
&&-\,\eta\,\int_{\Sigma }\,\overline{J}^{\,i}\,\overline{\Delta }\,J_i(\beta^*) \,d\mu_{\overline{g}}
\,+\, \eta\, \int_{\Sigma }\,J_a(\beta^*)\,\overline{\mathcal{R}}^{\;ab}\,\overline{J}_b \,d\mu_{\overline{g}}\,+\ldots\;.\nonumber
\end{eqnarray}
\begin{eqnarray}
&&\;\;\;\;\;\;\;\;\;\;\mathfrak{K}(\eta):=
\int_{\Sigma }\,K_{ab}(\beta^*)\,\overline{K}^{\;ab}(\eta) \,d\mu_{{g}(\eta)}\,=
\int_{\Sigma }\,K_{ab}(\beta^*)\,\overline{K}^{\;ab} \,d\mu_{\overline{g}}\,\\
&&-\,\eta\,\int_{\Sigma }\,\overline{K}^{\;ab}\,\overline{\Delta }\,K_{ab}(\beta^*) \,d\mu_{\overline{g}}\nonumber\\
&&+\, \eta\, \int_{\Sigma }\,\overline{K}^{\;ab}\left[\overline{\mathcal{R}}_{\;as}\,K_b^s(\beta^*)\,+
\overline{\mathcal{R}}_{\;bs}\,K_a^s(\beta^*)\,-       
2\overline{\mathcal{R}}_{\;asbt}\,K^{st}(\beta^*)\,\right]\, \,d\mu_{\overline{g}}\,+\ldots\;.\nonumber
\end{eqnarray}
where $\overline{\Delta }$ denotes the  rough Laplacian on $(\Sigma,\overline{g})$.\\
\\
\noindent 
It is clear from the above remarks that the spectrum of fluctuations, even for small $\eta$, is quite rich and one wonders if and in which sense 
 we are able to control, not just asymptotically, the fluctuations of the Ricci evolved fields $\left(\varrho(\beta^*),\,J_{a}(\beta^*),\,K_{ab}(\beta^*)\right)$ around the reference data  $\overline{\mathcal{C}}_{\overline{g}}(\Sigma)$. This will be done by studying  separately the behavior of the evolution of  matter fields $(\varrho(\beta), J_a(\beta))$ and that of the second fundamental form $K_{ab}(\beta)$. For the former we have quite a strong control in the entropy sense. For the latter, the situation is quite more complex, with the existence of possible non--dissipative directions for the fluctuations.

\section{Matter fields conjugation}
\label{matterfieldconj}
Let us start with an elementary but basic property of the scalar flow $\bar{\varrho}(\eta)\in C^{\infty }(\Sigma\times \mathbb{R},\mathbb{R})$, solution of the Ricci--conjugate heat equation 
$\bigcirc_d^*\,\bar{\varrho}(\eta)=(\frac{\partial }{\partial \eta }-\Delta +\mathcal{R})\,\bar{\varrho}(\eta) =0$. According to
\begin{equation}
 \frac{d}{d\eta}\,\int_{\Sigma }\bar{\varrho}(\eta)\,d\mu_{g(\eta)}=0\;,
\end{equation}
we can normalize  $\eta\mapsto \overline{\varrho}(\eta)$ so as to have $\int_{\Sigma }\bar{\varrho}(\eta)\,d\mu_{g(\eta)}=1$ on $\Sigma\times[0,\beta^*]$. Since we want to interpret $\bar{\varrho}(\eta)$ as a mass density,  we can further restrict our attention to positive solutions $\bar{\varrho}(\eta):=e^{-f(\eta)}$, for some $f(\eta)\in C^{\infty }(\Sigma\times \mathbb{R},\mathbb{R}^+)$, and consider $d\overline{\varrho}(\eta):=e^{-f(\eta)}\,d\mu_{g(\eta)}$ as a flow of probability measures on $\Sigma$. It is easily checked that these probability measures evolve according to the (backward) heat equation coupled with the fiducial Ricci flow 
\begin{equation} 
\begin{tabular}{l}
$\frac{\partial }{\partial \beta }g_{ab}(\beta )=-2\,\mathcal{R}_{ab}(\beta )\;,$\;\;\;\;$g_{ab}(\beta =0)=g_{ab}$\, , \\ 
\\ 
$\frac{\partial }{\partial \eta}\,d\overline{\varrho}(\eta)=\,\Delta_{g(\eta)}\,d\overline{\varrho}(\eta)\;,$\;\;\;$d\overline{\varrho} (\eta=0)=\bar{\varrho}(\eta=0)\, d\mu_{g(\beta^*)}$\, .
\end{tabular}
   \label{probeqf}
\end{equation}

\begin{remark}
When expressed in terms of $f(\eta)=-\ln\,\bar{\varrho}(\eta)$ this is  simply a (well--known) rewriting of  Perelman's Ricci flow coupling \cite{18} with the backward evolution $\eta\mapsto f(\eta)$
\begin{equation} 
\begin{tabular}{l}
$\frac{\partial }{\partial \beta }g_{ab}(\beta )=-2\,\mathcal{R}_{ab}(\beta )\;,$\;\;\;\;$g_{ab}(\beta =0)=g_{ab}$\, , \\ 
\\ 
$\frac{\partial f(\eta )}{\partial \eta}=\,\triangle_{g(\eta)} f-|\nabla f|_{g(\eta)}^2
+\mathcal{R}(\eta)\;,$\;\;\;$f(\eta =0)=f$\, .
\end{tabular}
   \label{eqf}
\end{equation}
In what follows, we shall indifferently use both representations. We identify $\bar{\varrho}(\eta)$ with the (reference) matter density flow induced by  $\bar{\varrho}(\eta=0)\in \overline{\mathcal{C}}_{\bar{g}}(\Sigma)$.
\end{remark}

\noindent
With these preliminary remarks along the way, let us consider the flow $\varrho (\beta)\in C^{\infty }(\Sigma\times\mathbb{R},\mathbb{R})$, solution of the scalar heat equation $\bigcirc_d\,\varrho (\beta)=0$. By the parabolic maximum principle, if $\varrho \geq 0$ we have $\varrho (\beta)\geq 0$, for all $\beta\in [0,\beta^*]$.
Moreover, since $\beta\mapsto \varrho (\beta)$ and $\eta\mapsto \overline{\varrho}(\eta)$ are conjugated flows on $\Sigma\times [0,\beta^*]$, we have 
\begin{equation}
\frac{d}{d\beta}\,\int_{\Sigma }\,\varrho(\beta)\,d\overline{\varrho}(\beta)=0\;,
\end{equation}
where $d\overline{\varrho}(\beta):=d\overline{\varrho}(\eta=\beta^*-\beta)$. Thus, we can  normalize the mass density flow $\beta\mapsto \varrho (\beta)$ associated with the data  ${\mathcal{C}}_{{g}}(\Sigma)$ so as to have $\int_{\Sigma }\,\varrho(\beta)\,d\overline{\varrho}(\beta)=1$, and assume that also  $d\Pi (\beta):=\varrho(\beta)\,d\overline{\varrho} (\beta)$ is a probability measure on $(\Sigma,g(\beta))$. This corresponds to  localize the matter content of ${\mathcal{C}}_{{g}}(\Sigma)$ with respect to the matter content of the reference $\bar{\mathcal{C}}_{\bar{g}}(\Sigma)$.\\
\\
\noindent In order to discuss the behavior of $\beta\rightarrow \varrho(\beta)$ with respect to the reference flow $\eta\rightarrow d\overline{\varrho}(\eta)$ let us introduce the relative entropy functional \cite{Deu},
\begin{equation}\label{scalentr}
\begin{split}
\mathcal{S} [d\Pi(\beta)| d\overline{\varrho}(\beta)] :=
\begin{cases}
\;\;\int_{\Sigma }\dfrac{d\Pi(\beta) }{d\overline{\varrho}(\beta) }\ln \dfrac{d\Pi(\beta)
}{d\overline{\varrho}(\beta) }d\overline{\varrho}(\beta)  & \hbox{if}\ d\Pi(\beta) \ll d\overline{\varrho}(\beta)\;,   \\
\;\;\infty  & \hbox{ otherwise}\;,
\end{cases}
\end{split}
\end{equation}
where $d\Pi(\beta) \ll d\overline{\varrho}(\beta) $ stands for absolute continuity. More explicitly, we can write
\begin{equation}
\mathcal{S}[d\Pi(\beta)| d\overline{\varrho}(\beta)]:=\int_{\Sigma }\,{\varrho }(\beta )\,\ln{\varrho  }(\beta )\,d\overline{\varrho}(\beta)\;,
\label{scalentr2}
\end{equation} 
\noindent also note that $\mathcal{S} [d\Pi(\beta)| d\overline{\varrho}(\beta)]$ is minus the physical relative entropy; the positive sign is more convenient for the analysis to follow.\\
\begin{figure}[h]
\includegraphics[scale=0.5]{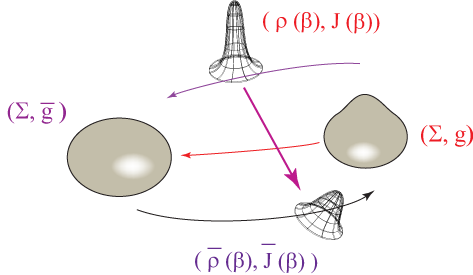}
\caption{The relative entropy $\mathcal{S} [d\Pi(\beta)| d\overline{\varrho}(\beta)]$, associated with the (normalized) distribution $\varrho(\beta)$ with respect to the reference $\overline{\varrho}(\eta)$, allows a rather strong control on the averaging effect that Ricci flow conjugation has on the matter fields. }
\label{fig: 22}   
\end{figure}
\noindent
\\
\noindent  Jensen's
inequality implies that  $\mathcal{S} [d\Pi(\beta)| d\overline{\varrho}(\beta)]\in $ $[0,+\infty ]$. Moreover, as a
function of the probability measures $d\overline{\varrho}(\beta)$ and $d\Pi(\beta) $,
$\mathcal{S} [d\Pi(\beta)| d\overline{\varrho}(\beta)]$ is convex and lower semicontinuous in the
weak topology on the space of probability measures $\hbox{Prob}(\Sigma,g(\beta))$ over $(\Sigma,g(\beta))$, and $\mathcal{S} [d\Pi(\beta)| d\overline{\varrho}(\beta)]=0$ iff $d\overline{\varrho}(\beta) =$ $d\Pi(\beta)$.  Along with $\mathcal{S} [d\Pi(\beta)| d\overline{\varrho}(\beta)]$ we also define the corresponding entropy production functional (the Fisher information) according to 
\begin{eqnarray}
\label{entrprod}
&&\mathcal{I} [d\Pi(\beta)| d\overline{\varrho}(\beta)]:= \int_{\Sigma }\dfrac{d\Pi(\beta) }{d\overline{\varrho}(\beta) }
\left| \nabla \,\ln\dfrac{d\Pi(\beta) }{d\overline{\varrho}(\beta) } \right|^{2}\,d\overline{\varrho}(\beta)\\
&&=\int_{\Sigma }\varrho(\beta )\left| \nabla \,\ln\varrho(\beta ) \right|^{2}\,d\overline{\varrho}(\beta)\;\geq 0\;.\nonumber
\end{eqnarray}

\begin{remark} (\emph{Logarithmic Sobolev inequality})\\
\noindent Since for each given $\beta\in [0,\beta^*]$,   $(\Sigma ,g(\beta))$ is  a compact Riemannian manifold of bounded geometry and  $(\Sigma,d\overline{\varrho}(\beta))$,\,$\beta\in[0,\beta^*]$, is a probability measure absolutely continuous with respect the Riemannian measure,   we can assume that a logarithmic Sobolev inequality with $\beta$--dependent constant $\tau(\beta)$, $LSI(\tau(\beta))$, holds on each $(\Sigma ,g(\beta))$. Explicitly, there exists a positive costant 
$\tau(\beta)$, uniformly bounded away from $0$ in $\beta\in[0,\beta^*]$, and depending from the geometry of  
$(\Sigma ,g(\beta),d\overline{\varrho}(\beta))$,  such that, for each given $\beta\in[0,\beta^*]$, we have 
\begin{equation}
\mathcal{S} [d\Xi (\beta)| d\overline{\varrho}(\beta)]\;\leq\;\frac{1}{2\,\tau(\beta)}\;\mathcal{I} [d\Xi(\beta)| d\overline{\varrho}(\beta)]\;,
\label{LSI}
\end{equation}
for all probability measures  $(\Sigma,d\Xi (\beta))$ absolutely continuous with respect to $(\Sigma, d\overline{\varrho}(\beta))$, \cite{17,22}. 
For each fixed $\beta\in[0,\beta^*]$, (\ref{LSI}) is equivalent to the standard form of of the logarithmic Sobolev inequality, (see \emph{e.g.} \cite{chowluni}). Notice that uniform logarithmic Sobolev estimates holding on the Ricci--flow spacetime $(\Sigma\times [0,\beta^*],g(\beta))$ have been established by R. Ye, (see \emph{e.g.} \cite{Ye2} and references therein)). For our purposes the simpler (\ref{LSI}) suffices.
\end{remark}
\noindent
By exploiting (\ref{LSI}) we can easily establish the following 
\begin{theorem} (Control in the entropy sense)\\
\label{lemmaSmatt}
\noindent The functional $\mathcal{S} [d\Pi(\beta)| d\overline{\varrho}(\beta)]$ is monotonically non--increasing along the 
flow $\beta \mapsto (g(\beta),d\overline{\varrho}(\beta))$
\begin{equation}
\frac{d }{d \beta }\,\mathcal{S} [d\Pi(\beta)| d\overline{\varrho}(\beta)]=\,-\, \mathcal{I} [d\Pi(\beta)| d\overline{\varrho}(\beta)]\;.
\label{monot1}
\end{equation}
Moreover, as the length scale $\beta$ increases, the matter distribution $d\Pi(\beta)$ is localized, around the reference distribution $d\overline{\varrho}(\beta)$, in the entropy sense according to
\begin{equation}\label{pinsker}
\frac{1}{2}\left\| d\Pi(\beta)
-d\overline{\varrho}(\beta)  \right\| _{\rm var}^{2}\leq\mathcal{S} [d\Pi(\beta)| d\overline{\varrho}(\beta)]\leq\;e^{-\,2\,\int_0^{\beta}\tau\,(t)\,dt}\;\mathcal{S}_0 [d\Pi| d\overline{\varrho}]\;,
\end{equation}
where ${S}_0 [d\Pi| d\overline{\varrho}]:=\mathcal{S} [d\Pi(\beta=0)| d\overline{\varrho}(\beta=0)]$, and where $\left\|\;\;   \right\| _{\rm var}^{2}$ denotes the total variation norm on
$\hbox{Prob}(\Sigma,g(\beta) )$ defined by
\begin{equation}\label{variation}
\left\| d\Pi(\beta) -d\overline{\varrho}(\beta)  \right\|_{\rm var}\doteq
\sup_{\left\| \phi \right\| _{b}\leq 1}\left\{ \left| \int_{\Sigma
}\phi d\Pi(\beta) -\int_{\Sigma }\phi d\overline{\varrho}(\beta)  \right| \right\}\;,
\end{equation}
$\parallel \phi\parallel _{b}$ being the uniform norm on the space of bounded measurable functions on $\Sigma $.
\end{theorem}

\begin{proof} For any $\varrho \in C^{2}(\Sigma\times[0,\beta^*] ,\mathbb{R}^{+})$, $\frac{\partial \,\varrho}{\partial \beta }=\triangle \,\varrho$,  we compute 
\begin{eqnarray}
&&\frac{d }{d \beta }\int_{\Sigma }\varrho\,\ln\,\varrho\,d\overline{\varrho} (\beta)=\\
\nonumber\\
&&\int_{\Sigma }(1+\ln\,\varrho)\,\triangle \varrho\,d\overline{\varrho} (\beta)
-\int_{\Sigma }\varrho\,\ln\,\varrho\,\triangle (d\overline{\varrho} (\beta))= \nonumber\\
\nonumber\\
&&=\int_{\Sigma }(1+\ln\,\varrho)\,\triangle \varrho\,d\overline{\varrho} (\beta)
-\int_{\Sigma }\triangle(\varrho\,\ln\,\varrho)\, d\overline{\varrho} (\beta)\;, \nonumber
\end{eqnarray}
where, in the last line, we have integrated by parts. From the identity, 
\begin{equation}
\triangle (\varrho\,\ln\,\varrho)=(1+\ln\,\varrho)\triangle \varrho+\varrho^{-1}|\nabla \varrho|^{2}\;,
\end{equation}
we get
\begin{eqnarray}
\frac{d }{d \beta }\int_{\Sigma }\varrho\,\ln\,\varrho\,d\overline{\varrho} (\beta)&=&-
\int_{\Sigma }\varrho^{-1}\,|\nabla \varrho|^{2}\, d\overline{\varrho} (\beta)=\\
\nonumber\\
&=&-\int_{\Sigma }\varrho\,|\nabla \,\ln\,\varrho|^{2}\, d\overline{\varrho} (\beta)\;,\nonumber
\end{eqnarray}
and (\ref{monot1}) follows. For  each fixed $\beta\in[0,\beta^*]$, the logarithmic Sobolev inequality (\ref{LSI}) and (\ref{monot1}) imply
\begin{equation}
\frac{d }{d \beta }\,\mathcal{S} [d\Pi(\beta)| d\overline{\varrho}(\beta)]\,\leq\,-\,2\,\tau(\beta)\,\mathcal{S} [d\Pi(\beta)| d\overline{\varrho}(\beta)]\;,
\end{equation}
which yields
\begin{equation}
\mathcal{S} [d\Pi(\beta)| d\overline{\varrho}(\beta)]\leq\;e^{-\,2\,\int_0^{\beta}\tau\,(t)\,dt}\;\mathcal{S}_0 [d\Pi| d\overline{\varrho}]\;.
\end{equation}
Finally, from the Csisz\'ar--Kullback--Pinsker inequality, (see \emph{e.g.} \cite{Deu}), 
\begin{equation}
\mathcal{S} [d\Pi(\beta)| d\overline{\varrho}(\beta)]\,\geq\,\frac{1}{2}\left\| d\Pi(\beta)
-d\overline{\varrho}(\beta)  \right\| _{\rm var}^{2}\;,
\end{equation}
we get  (\ref{pinsker}), as stated.
\end{proof}
Such a result imples that,  along  $\beta \mapsto (g(\beta),d\overline{\varrho} (\beta))$,  the distribution $d\Pi(\beta)$ localizes around the reference matter distribution $d\overline{\varrho}(\beta)$ in a rather strong sense. Stated differently, $\varrho (\beta)$ is averaged with respect to $d\overline{\varrho}(\beta)$. 
One can easily see this by observing that the variance of $\varrho (\beta)$, under the reference probability measure $d\overline{\varrho}(\beta)$, given by 
\begin{equation}
Var[d\Pi(\beta)|d\overline{\varrho}(\beta)]:=\,\int_{\Sigma }\,(\varrho(\beta)-1)^2\,d\overline{\varrho}(\beta)\;,   
\end{equation}
is strictly decreasing along the reference flow $\beta\rightarrow (g(\beta),d\overline{\varrho}(\beta))$. We have
\begin{lemma} (\emph{Evolution of variance})
The variance $Var[d\Pi(\beta)|d\overline{\varrho}(\beta)]$ is monotonically decreasing  along $\beta \mapsto (g(\beta),d\overline{\varrho}(\beta))$ 
\begin{equation}
\frac{d}{d\beta}\,Var[d\Pi(\beta)|d\overline{\varrho}(\beta)]=\,-2\,\int_{\Sigma }\,|\nabla \varrho(\beta)|^2\,d\overline{\varrho}(\beta)\;,
\label{varev}
\end{equation}
and
\begin{equation}\label{convar}
Var [d\Pi(\beta)| d\overline{\varrho}(\beta)]\leq\;e^{-\,2\,\int_0^{\beta}\tau\,(t)\,dt}\;Var_0 [d\Pi| d\overline{\varrho}]\;,
\end{equation}

\noindent 
where ${Var}_0 [d\Pi| d\overline{\varrho}]:=Var [d\Pi(\beta=0)| d\overline{\varrho}(\beta=0)]$.
\label{lemmavar}
\end{lemma}   

\begin{proof}
From 
\begin{eqnarray}
&&\frac{\partial}{\partial\beta}[(\varrho-1)^2\,\bar{\varrho}\,d\mu_{g}]=[\Delta(\varrho-1)^2-2|\nabla \varrho|^2]\bar{\varrho}\,d\mu_{g}\\
 &&+(\varrho-1)^2[-\Delta \overline{\varrho}+\mathcal{R}\,\bar{\varrho}]\,d\mu_{g}-(\varrho-1)^2\bar{\varrho}\mathcal{R}\,d\mu_{g}\;,\nonumber
\end{eqnarray}
which holds pointwise for  $\varrho \in C^{2}(\Sigma\times[0,\beta^*] ,\mathbb{R}^{+})$, we get (\ref{varev})     
by integrating over $\Sigma$ with respect to $d\overline{\varrho}(\beta)$. For each given $\beta\in[0,\beta^*]$, the logarithmic Sobolev inequality $LSI(\tau(\beta))$ implies, (see \emph{e.g.}\cite{LottVill}), the following Poincar\'e inequality  for the pair $(d\Pi (\beta), d\overline{\varrho}(\beta))$ 
\begin{equation}
\tau(\beta)\,Var [d\Pi (\beta)| d\overline{\varrho}(\beta)]\,\leq\,\int_{\Sigma }\,|\nabla \varrho(\beta)|^2\,d\overline{\varrho}(\beta)\;, 
\end{equation}
from which (\ref{convar}) immediately follows.
\end{proof}
\begin{remark}
\label{remmatt}
From $\varrho (\beta)\in C^{\infty }(\Sigma\times [0,\beta^*],\mathbb{R})$,   $\bigcirc_d\,\varrho (\beta)=0$ it  follows that $\nabla _i\varrho (\beta)\in C^{\infty }(\Sigma\times [0,\beta^*],T^*\Sigma)$ is a solution of  $\bigcirc_d\,\nabla \,\varrho (\beta)=0$, (a trivial consequence of  the well--known commutation $\nabla _i\Delta =\Delta _d\nabla _i$). A direct computation provides
\begin{equation}
\frac{\partial }{\partial \beta }\,|\nabla \varrho (\beta)|^2=\Delta \,|\nabla \varrho (\beta)|^2-2|\nabla \nabla \varrho (\beta)|^2\;,
\label{rhonorm}
\end{equation} 
which, by the maximum principle, implies that $\sup_{x\in\Sigma}\,|\nabla \varrho (\beta)|^2$ is non--increasing as $0\leq\beta\leq\beta^*$. Moreover, 
by integrating with respect to $d\overline{\varrho}  (\beta)$ we get 
\begin{eqnarray}
&&\frac{d}{d\beta }\int_\Sigma \,|\nabla \varrho (\beta)|^2\,d\overline{\varrho} (\beta)\\
&&=\int_\Sigma \,\left[\left(\Delta |\nabla \varrho (\beta)|^2-2|\nabla \nabla \varrho (\beta)|^2 \right)\bar{\varrho} (\beta)\right.\nonumber\\
&&\left.+ |\nabla \varrho (\beta)|^2\left(-\Delta\overline{\varrho} (\beta)+\mathcal{R}(\beta)\bar{\varrho} (\beta)\right)-|\nabla \varrho (\beta)|^2 \mathcal{R}(\beta)\bar{\varrho} (\beta)\right]\,d\mu_{g(\beta)}\;,\nonumber
\end{eqnarray}
which easily yields
\begin{equation}
\frac{d}{d\beta}\,\int_{\Sigma}\,|\nabla \varrho (\beta)|^2\,d\overline{\varrho} (\beta)=\,-2\,\int_{\Sigma}\,|\nabla \nabla \,\varrho (\beta)|^2\,d\overline{\varrho} (\beta)\;.
\label{L2om}
\end{equation}
\end{remark}
The above remark suggests that the flow $\beta\rightarrow \varrho (\beta)$ also dominates the matter current flow $\beta\rightarrow J(\beta)\in C^{\infty }(\Sigma,T^*\Sigma)$ defined, along the Ricci flow $\beta\rightarrow g(\beta)$ by $\bigcirc_d\,J(\beta)=0$, (see (\ref{compact1})). 
In particular, let us consider the evolution of the $\beta$--dependent norm, $|J(\beta)|^2:=J_{i}(\beta)J_{k}(\beta)g^{ik}(\beta)$, of $J_{i}(\beta)$. As in (\ref{rhonorm}) we get  
\begin{equation}
\frac{\partial }{\partial \beta }\,|J(\beta)|^2=\Delta \,|J(\beta)|^2-2|\nabla J(\beta)|^2\;,
\label{Jnorm}
\end{equation}
where $|\nabla J(\beta)|^2:=\nabla ^iJ_k(\beta)\nabla _iJ^k(\beta)$. The maximum principle implies that $\sup_{\;x\in\Sigma}\,\,|J(\beta)|^2$ is non--increasing as $0\leq\beta\leq\beta^*$, and integration with respect to the probability measure $\beta\rightarrow d\overline{\varrho} (\beta)$ provides
\begin{equation}
\frac{d}{d\beta}\,\int_{\Sigma}\,|J(\beta)|^2\,d\overline{\varrho} (\beta)=\,-\,2\int_{\Sigma}\,|\nabla \,J(\beta)|^2\,d\overline{\varrho} (\beta)\;.
\label{L2v}
\end{equation}
Moreover, if we consider  the evolution of $|J(\beta)|:=(J_{i}(\beta)J_{k}(\beta)g^{ik}(\beta))^{1/2}$, then  from (\ref{Jnorm}) we compute
\begin{equation}
\frac{\partial }{\partial \beta }\,|J(\beta)|=\Delta \,|J(\beta)|+|J(\beta)|^{-1}\left(|\nabla |J(\beta)||^2-|\nabla J(\beta)|^2\right)\;.
\label{Jn}
\end{equation}
By setting $J_k(\beta)=n_k(\beta)\,|J(\beta)|$, where $n(\beta)\in C^{\infty }(\Sigma,T^*\Sigma)$,  $n_i(\beta)n_k(\beta)g^{ik}(\beta)=1$, $\forall \beta\in [0,\beta^*]$, we get 
$|\nabla |J(\beta)||^2-|\nabla J(\beta)|^2=\,-|J(\beta)|^2|\nabla n(\beta)|^2$. Thus
\begin{equation}
\frac{\partial }{\partial \beta }\,|J(\beta)|=\Delta \,|J(\beta)|-|J(\beta)||\nabla n(\beta)|^2\;,
\label{Jmc}
\end{equation}
and by subtracting this latter expression to the evolution $\bigcirc_d\,\varrho (\beta)=0$, we eventually get
\begin{equation}
\frac{\partial }{\partial \beta }\,\left(\varrho (\beta)-|J(\beta)|\right)=\Delta \,\left(\varrho (\beta)-|J(\beta)|\right)+|J(\beta)||\nabla n(\beta)|^2\;.
\label{Jmc2}
\end{equation}
The maximum principle implies that $\left(\varrho (\beta)-|J(\beta)|\right)\geq 0$ on $\Sigma\times [0,\beta^*]$, as soon as  $\left(\varrho (\beta)-|J(\beta)|\right)_{\beta=0}\geq 0$. Thus, we have established the following
\begin{theorem}
\label{domlemma}
The dominant energy condition 
\begin{equation}
\varrho (\beta)\;\geq \;|J(\beta)|\;,
\end{equation}
holds along the flows $\beta\rightarrow (\varrho (\beta),J(\beta))$, $0\leq\beta\leq\beta^*$.
\end{theorem} 
\noindent
A similar result holds also for the conjugate flows $\eta\rightarrow (\bar{\varrho} (\eta),\bar{J}(\eta))$, $0\leq\eta\leq\beta^*$ solutions of 
$\bigcirc_d^{*}(\bar{\varrho} (\eta),\bar{J}(\eta))=0$. We have
\begin{lemma}
\label{lembarJ}
For any $(\bar{\varrho} (\eta),\bar{J}(\eta))\in C^{\infty }(\Sigma \times [0,\beta^*],\otimes ^p\,T\Sigma )$, $p=0,1$, solution  of the conjugate flow $\bigcirc_d^{*}(\bar{\varrho} (\eta),\bar{J}(\eta))=0$, with $\bar{\varrho} (\eta=0)\;\geq \;|\bar{J}(\eta=0)|$, we have 
\begin{equation}
\overline{\varrho} (\eta)\;\geq \;|\bar{J}(\eta)|\;,
\end{equation}
for all $0\leq\eta\leq\beta^*$.
\end{lemma}
\begin{proof}
As in the derivation of (\ref{Jmc}) above, if we set  $\bar{J}^k(\eta)=\bar{m}^k(\eta)\,|\bar{J}(\eta)|$, where $\bar{m}^i(\eta)\bar{m}^k(\eta)g_{ik}(\eta)=1$, $\forall \eta\in [0,\beta^*]$, we easily get from $\bigcirc_d^{*}\bar{J}(\eta)=0$ that $|\bar{J}(\eta)|$ evolves according to 
\begin{equation}
\frac{\partial }{\partial \eta }\,|\bar{J}(\eta)|=\Delta \,|\bar{J}(\eta)|-|\bar{J}(\eta)||\nabla \overline{m}(\eta)|^2-\mathcal{R}(\eta)|\bar{J}(\eta)|\;.
\label{barJmc}
\end{equation}
Subtracting this expression to the evolution $\bigcirc_d^*\,\bar{\varrho}(\eta)=(\frac{\partial }{\partial \eta }-\Delta +\mathcal{R})\,\bar{\varrho}(\eta) =0$ of the matter density $\bar{\varrho} (\eta)$, we get
\begin{equation}
\frac{\partial }{\partial \eta }\,\left(\bar{\varrho}-|\bar{J}|\right)=\Delta \,\left(\bar{\varrho}-|\bar{J}|\right)+|\bar{J}||\nabla \overline{m}|^2-\mathcal{R} \left(\bar{\varrho}-|\bar{J}|\right)\;,
\label{barrhoJmc}
\end{equation}
where we dropped the explicit $\eta$--dependence for notational ease. The presence of the scalar curvature term $\mathcal{R} \left(\bar{\varrho}-|\bar{J}|\right)$ is rather annoying, and to take care of it we exploit the fact that the Riemannian measure density $\sqrt{\det g(\beta)}$ is covariantly constant with respect to the Levi--Civita connection $\nabla$ associated with $(\Sigma ,g(\beta))$, \emph{i.e.} $\nabla _{k}\det g(\beta )=\det g(\beta)\, g^{ab}(\beta)\,\nabla _{i}g_{ab}(\beta)\equiv 0$, (this is equivalent to the familiar formula, $\partial _{i}\ln \sqrt{\det g(\beta )}=\delta _{a}^{c}\Gamma _{ic}^{a}(\beta )$, for the trace of
the Christoffel symbols $\Gamma _{ic}^{a}(\beta )$ associated with $g_{ab}(\beta )$). In particular, we have $\Delta\,\sqrt{\det g(\beta)}\equiv 0 $.
Thus, by computing the $\eta$--evolution of the tensor density $\sqrt{\det g(\eta)}\,\left(\bar{\varrho}-|\bar{J}|\right)$, from (\ref{barrhoJmc}) and 
the Ricci flow evolution $\frac{\partial }{\partial \eta }\,\sqrt{\det g(\eta)}\,=\sqrt{\det g(\eta)}\,\mathcal{R}(\eta)$, we get
\begin{equation}
\frac{\partial }{\partial \eta }\,\sqrt{\det g}\,\left(\bar{\varrho}-|\bar{J}|\right)=\Delta \,\sqrt{\det g}\,\left(\bar{\varrho}-
|\bar{J}|\right)+\sqrt{\det g}\,|\bar{J}||\nabla \overline{m}|^2\;.
\label{barrhoJmcdet}
\end{equation}
Again, a direct application of the parabolic maximum principle implies that $\sqrt{\det g(\eta)}\,\left(\bar{\varrho}(\eta)-|\bar{J}(\eta)|\right)\geq 0$ on $\Sigma\times [0,\beta^*]$, as soon as this condition holds for $\eta=0$. Since the Riemannian density $\sqrt{\det g(\eta)}$ is uniformly bounded away from zero on $\Sigma\times [0,\beta^*]$, we have $\left(\bar{\varrho}(\eta)-|\bar{J}(\eta)|\right)\geq 0$.
\end{proof}
The matter dominance also implies localization and averaging, in the entropy sense, of $|J(\beta)|$ with respect to the reference $|\bar{J}(\beta)|$.
If we define the matter current relative entropy  according to
\begin{equation}
\mathcal{S}\,\left[d\,J(\beta)\,|\,d\,\bar{J} (\beta) \right]\,:=\,\int_{\Sigma}\,|J (\beta)|\;\ln\,|J (\beta)|\,|\bar{J} (\beta)|\,d\mu_{g(\beta)}\;,
\end{equation} 
then theorem \ref{lemmaSmatt} implies

\begin{lemma} (Matter entropy dominance)\\
\noindent
Along the conjugated matter flows $({\varrho} (\beta),{J}(\beta))\in C^{\infty }(\Sigma \times [0,\beta^*],\otimes ^p\,T^*\Sigma )$, and $(\bar{\varrho} (\beta),\bar{J}(\beta))\in C^{\infty }(\Sigma \times [0,\beta^*],\otimes ^p\,T\Sigma )$,  $p=0,1$, we have  
\begin{eqnarray}
\label{pinskercurr}
\mathcal{S}\,\left[d\, J (\beta)\,|\,d\,\bar{J} (\beta) \right]
\leq\;e^{-\,2\,\int_0^{\beta}\tau\,(t)\,dt}\;\mathcal{S}_0 [d\Pi| d\overline{\varrho}]\;,
\end{eqnarray}
where ${S}_0 [d\Pi| d\overline{\varrho}]:=\mathcal{S} [d\Pi(\beta=0)| d\overline{\varrho}(\beta=0)]$, and where $\tau(\beta)>0$, $\beta\in[0,\beta^*]$, is 
the $\beta$--dependent log--Sobolev constant characterized by (\ref{LSI}). 
\end{lemma}

\begin{proof}
The dominant energy conditions established above  imply
\begin{equation}
{\varrho}(\beta)\,\ln\,{\varrho}(\beta)\,\bar{\varrho}(\beta)\geq |J(\beta)|\;\ln\,|J(\beta)|\,|\bar{J}(\beta)|\;,
\end{equation}
whenever both expression make sense.
The entropy dominance (\ref{pinskercurr}) directly follows from theorem \ref{lemmaSmatt}. 
\end{proof}

\noindent 
Not surprisingly the conjugate evolution (\ref{eqf}) of the matter density $\eta\mapsto \overline{\varrho }(\eta)$  is strictly related to   
the Perelman functional \cite{18} $\mathcal{F}:\mathcal{M}et(\Sigma)\times C^{\infty }(\Sigma,R)\rightarrow R$ defined by
\begin{equation}
\mathcal{F}[{g};f]\doteq \int_{\Sigma }(\mathcal{R}+|\nabla f|^{2}){\rm
e}^{-f}\,d\mu _{{g}}\;. \label{F-funct}
\end{equation}
In particular, if we introduce the relative entropy functional associated with the distribution of $\overline{\varrho}(\beta)$ with respect to the $\beta$--evolving Riemannian measure $d\mu_{g(\beta)}$,
\begin{equation}
\mathcal{S}[d\overline{\varrho}(\beta)| d\mu_{g(\beta)}]:=\int_{\Sigma }\,\overline{\varrho }(\beta )\,\ln\overline{\varrho  }(\beta )\,
d\mu_{g(\beta)}\;,
\label{scalentr3}
\end{equation} 
then one easily checks \cite{Muller, Glickenstein}  that on $\Sigma_{\beta }\times [0,\beta ^{*}]$   
\begin{equation}
\frac{d}{d\beta}\,\mathcal{S}[d\overline{\varrho}(\beta)| d\mu_{g(\beta)}]\,=\,\int_{\Sigma }(\mathcal{R}+|\nabla \ln\overline{\varrho }|^{2})\;
d\overline{\varrho}(\beta)=\;\mathcal{F}[{g};f]\;, 
\end{equation}
for $g$ evolving along the fiducial Ricci flow $\beta\mapsto g(\beta)$, and  $\eta\mapsto f=-\ln\overline{\varrho}(\eta)$ evolving backward according to (\ref{eqf}). Since along this evolution $\mathcal{F}[{g};f]$ is non--decreasing \cite{18}
\begin{equation}
\frac{d}{d\beta} \mathcal{F}[{g(\beta)};f(\beta)]= 2\, \int_{\Sigma }\left|\mathcal{R}_{ik}(\beta )+\nabla_i \nabla_k\, f(\beta)\right|^2{\rm
e}^{-f(\beta)}\,d\mu _{{g(\beta)}}\geq 0\;,
\end{equation}
we immediately get $\frac{d^2}{d\beta^2}\,\mathcal{S}[d\overline{\varrho}(\beta)| d\mu_{g(\beta)}]\geq 0$, \emph{i.e.}   $\mathcal{S}[d\overline{\varrho}(\beta)| d\mu_{g(\beta)}]$ is convex along $\beta\mapsto g(\beta)$ on $\Sigma_{\beta }\times [0,\beta ^{*}]$. The functional $\mathcal{F}[{g};f]$  will play a slightly more sophisticated  role in the analysis of geometric fields conjugation. 

\section{Dissipative and non--dissipative directions for $K(\beta)$}
\label{sezioneSei}
According to definition \ref{theDEF}, the conjugation between the second fundamental forms $K_{ab}\in \mathcal{C}_{g}(\Sigma )$ and $\bar{K}^{ab}\in \overline{\mathcal{C}}_{\bar{g}}(\Sigma )$ is defined by the heat flows  $K_{ab}(\beta)\in C^{\infty }(\Sigma \times [0,\beta^*],\otimes _S^2T^*\Sigma )$ and $\bar{K}^{ab}(\eta)\in C^{\infty }(\Sigma \times [0,\beta^*],\otimes _S^2T\Sigma )$, solutions of $\bigcirc_d\,K_{ab}(\beta)=0$, 
$K_{ab}(\beta=0)=K_{ab}$, and of $\bigcirc_d^*\,\bar{K}^{ab}(\eta)=0$, $\bar{K}^{ab}(\eta=0)=\bar{K}^{ab}$, respectively.
We have a rather obvious control on the forward evolution $K_{ab}(\beta)\in C^{\infty }(\Sigma \times [0,\beta^*],\otimes _S^2T^*\Sigma )$  whenever the initial datum  is a Lie derivative $K(\beta=0)= 2\,\delta ^*_{g}\,v^{\sharp }$, for some  $v\in C^{\infty }(\Sigma,T^*\Sigma )$. In such a case,  according to lemma \ref{commlemm}, (eq. (\ref{comm1})), if $v(\beta)\in C^{\infty }(\Sigma \times [0,\beta^*],T^*\Sigma ) $ evolves according to  $\bigcirc_d\,v_a(\beta)=0$, with $v_a(\beta=0)=v_a$, then  $K_{ab}(\beta)=\nabla _a\,v_b(\beta)+\nabla _b\,v_a(\beta)$ is the solution of  $\bigcirc_d\,K_{ab}(\beta)=0$, $K_{ab}(\beta=0)=\nabla _a\,v_b+\nabla _b\,v_a$. Thus, in such a case  the heat flow deformation of $K\in \mathcal{C}_{g}(\Sigma )$ simply gives rise to a $\mathcal{D}iff(\Sigma )$--reparametrization of the underlying Ricci flow. Explicitly,

\begin{lemma}
If the second fundamental form $K\in\mathcal{C}_{g}(\Sigma )$  belongs to $Im\, \delta^*_g$, then its evolution according to $\bigcirc_d\,K(\beta)=0$ generates a flow
\begin{equation}
Im\, \delta^*_g\ni K\mapsto \delta ^*_{g(\beta)}\,v^{\sharp }(\beta)\in \mathcal{U}_{g(\beta)}\cap\,\mathcal{T}_{(\Sigma ,g(\beta))}\mathcal{M}et(\Sigma )\times [0,\beta^*]\;,
\label{urew}
\end{equation}
where $\mathcal{U}_{g(\beta)}\cap\,\mathcal{T}_{(\Sigma ,g(\beta))}\mathcal{M}et(\Sigma )$ is an open neighborhood of $g(\beta)$, in the $\mathcal{D}iff(\Sigma )$ orbit $\mathcal{O}_{g(\beta)}$. Moreover, $\sup_{\,x\in\Sigma}\,|v(\beta)|$ is non--increasing for $0\leq\beta\leq\beta^*$, and if  a dominant energy condition $\varrho\geq |v|$ holds at $\beta=0$, we have that $\varrho(\beta)\geq |v(\beta)|$ along 
$v(\beta)\in C^{\infty }(\Sigma \times [0,\beta^*],T^*\Sigma )$. Finally, the $L^2(\Sigma ,d\overline{\varrho} (\beta))$ norm of $v(\beta)$ is monotonically decreasing  according to
\begin{equation}
\frac{d}{d\beta}\,\int_{\Sigma}\,|v(\beta)|^2\,d\overline{\varrho} (\beta)=\,-\,2\int_{\Sigma}\,|\nabla \,v(\beta)|^2\,d\overline{\varrho} (\beta)\;.
\label{L2J}
\end{equation}
\end{lemma}

\begin{proof}
The statement (\ref{urew}) is basically a rewriting of lemma \ref{lemmDiff}. We can also apply to this situation the results we obtained for the matter current flow $J(\beta)$, and ultimately for the matter density flow $\varrho(\beta)$, (see remark \ref{remmatt}, also note that $\bigcirc_d\,\varrho(\beta)=0$ implies 
$\bigcirc_d\,Hess\,\varrho(\beta)=0$, thus $K(\beta)$ and $Hess\,\varrho(\beta)$ satisfy the same heat evolution). It follows that $\sup_{\,x\in\Sigma}\,|v(\beta)|$ is non--increasing for $0\leq\beta\leq\beta^*$. Moreover, since by a rescaling we can always assume that a dominant energy condition $\varrho\geq |v|$ holds at $\beta=0$, we have, according to lemma 
\ref{domlemma}, that $\varrho(\beta)\geq |v(\beta)|$ along $v(\beta)\in C^{\infty }(\Sigma \times [0,\beta^*],T^*\Sigma )$. The $L^2(\Sigma ,d\overline{\varrho} (\beta))$ evolution (\ref{L2J}) is a direct rewriting of (\ref{L2v}).
\end{proof}

\noindent For the reference conjugate evolution $\bar{K}^{ab}(\eta)\in C^{\infty }(\Sigma \times [0,\beta^*],\otimes _S^2T\Sigma )$  we have a natural counterpart of these results

\begin{lemma}
Along the solution $\bar{K}^{ab}(\eta)\in C^{\infty }(\Sigma \times [0,\beta^*],\otimes _S^2T\Sigma )$ of the conjugate heat flow $\bigcirc_d^*\,\bar{K}^{ab}(\eta)=0$, $\bar{K}^{ab}(\eta=0)=\bar{K}^{ab}$, the scalar density $\sup_{\,x\in\Sigma}\,\sqrt{\det\,g(\eta)}\,\left|\delta_{g(\eta)}\,\bar{K}(\eta)\right|$ is non--increasing. Moreover
\begin{equation}
\frac{d}{d\eta}\,\int_{\Sigma } |\delta_{g(\eta)}\,\bar{K}(\eta)|\,d\mu_{g(\eta)}=\,-\,\int_{\Sigma } |\delta_{g(\eta)}\,\bar{K}(\eta)|\,
\left|\nabla \overline{Q}(\eta) \right|^2\,d\mu_{g(\eta)}\;,
\label{intdivK}
\end{equation}
where $\left|\delta_{g(\eta)}\,\bar{K}(\eta)\right|:=\left[g_{hi}(\eta)\nabla_a\,K^{ah}\,\nabla_b\,K^{bi}\right]^{1/2}$, and where 
$Q(\eta)\in C^{\infty }(\Sigma \times [0,\beta^*],T\Sigma )$ is defined by $Q^i(\beta):=|\delta_{g(\eta)}\,\bar{K}(\eta)|^{-1}\,\nabla_b\,K^{bi}$.
\end{lemma}

\begin{proof}
According to lemma \ref{commlemm} 
\begin{equation}
\bigcirc^{*} _{d}\;\left(\delta _{g(\eta)}\,\bar{K}(\eta) \right)=\delta _{g(\eta)}\,\left(\bigcirc^{*} _{d}\;\bar{K}(\eta) \right)\;,
\label{Kcomm1*}
\end{equation}
(see (\ref{comm1*})). This implies that, along $\bigcirc_d^*\,\bar{K}^{ab}(\eta)=0$, the divergence $\nabla_b\,\bar{K}^{bi}(\eta)$ evolves according to 
\begin{equation}
\frac{\partial }{\partial \eta}\,\nabla_b\,\bar{K}^{bi}(\eta)=\,\Delta _d\,\nabla_b\,\bar{K}^{bi}(\eta)-\,\mathcal{R}(\eta)\,\nabla_b\,\bar{K}^{bi}(\eta)\;,
\label{divH}
\end{equation}
where, as usual, $\Delta _d=\Delta -\mathcal{R}ic(\beta)$ is the Hodge-DeRham Laplacian on vectors.
By proceeding as in lemma \ref{lembarJ}, (see eqns. 
(\ref{barJmc}) and (\ref{barrhoJmcdet}), and also (\ref{Jmc})), we compute  
\begin{equation}
\frac{\partial }{\partial \eta}\,\sqrt{\det\,g}\,|\delta_{g}\,\bar{K}|=\,\Delta \sqrt{\det\,g}\,|\delta_{g}\,\bar{K}|-\,
\sqrt{\det\,g}\,|\delta_{g}\,\bar{K}|\left|\nabla \overline{Q} \right|^2\;,
\label{divK}
\end{equation}
along the flow (\ref{divH}). The non--increasing character of the scalar density $\sup_{\,x\in\Sigma}\,\sqrt{\det\,g(\eta)}\,\left|\delta_{g(\eta)}\,\bar{K}(\eta)\right|$ immediately follows from the
maximum principle, and a direct integration of (\ref{divK}) over $(\Sigma ,g(\eta))$ provides (\ref{intdivK}).
\end{proof} 

\noindent The above result shows that $\int_{\Sigma } |\delta_{g(\eta)}\,\bar{K}(\eta)|\,d\mu_{g(\eta)}$ decreases exponentially fast along the solution of the conjugate HDRL heat flow $\bigcirc_d^*\,\bar{K}^{ab}(\eta)=0$, $\bar{K}^{ab}(\eta=0)=\bar{K}^{ab}$. We also have some form of control on $tr_{\,g(\eta)}\,\bar{K}^{ab}(\eta)$.
\begin{lemma}
Along the solution $\bar{K}^{ab}(\eta)\in C^{\infty }(\Sigma \times [0,\beta^*],\otimes _S^2T\Sigma )$ of the conjugate heat flow $\bigcirc_d^*\,\bar{K}^{ab}(\eta)=0$, $\bar{K}^{ab}(\eta=0)=\bar{K}^{ab}$,  the integral of $tr_{\,g(\eta)}\,\bar{K}^{ab}(\eta)$ evolves linearly with $0\leq\eta\leq\beta^*$ according to 
\begin{equation}
\int_{\Sigma } tr_{\,g(\eta)}\,\bar{K}(\eta)\,d\mu_{g(\eta)}=\,\int_{\Sigma } tr_{\,\bar{g}}\,\bar{K}\,d\mu_{\bar{g}}+2\,\eta\,\int_{\Sigma }\,\mathcal{R}_{ab}(\bar{g})\,\bar{K}^{ab}\,d\mu_{\bar{g}}\;.
\label{intKtrace}
\end{equation}
\end{lemma}
\begin{proof}
The most direct way of proving this result is by exploiting theorem  \ref{nondiss} according to which
\begin{equation}
\int_{\Sigma }\,\left(g_{ab}(\eta)-2\eta\,\mathcal{R}_{ab}(\eta)\right)\,\bar{K}^{ab}(\eta)\,d\mu_{g(\eta)}\;,
\label{consuno}
\end{equation}
and 
\begin{equation}
\int_{\Sigma }\,\mathcal{R}_{ab}(\eta)\,\bar{K}^{ab}(\eta)\,d\mu_{g(\eta)}\;,
\label{consdue}
\end{equation}
are costant along the solution $\bar{K}^{ab}(\eta)\in C^{\infty }(\Sigma \times [0,\beta^*],\otimes _S^2T\Sigma )$ of the conjugate heat flow $\bigcirc_d^*\,\bar{K}^{ab}(\eta)=0$, $\bar{K}^{ab}(\eta=0)=\bar{K}^{ab}$. From (\ref{consuno}) we get (by evaluation at $\eta=0$)
\begin{equation}
\int_{\Sigma } tr_{\,\bar{g}}\,\bar{K}\,d\mu_{\bar{g}}\,=\,\int_{\Sigma }\,\left(g_{ab}(\eta)-2\eta\,\mathcal{R}_{ab}(\eta)\right)\,\bar{K}^{ab}(\eta)\,d\mu_{g(\eta)}\;.
\end{equation}
The constancy along the flow of (\ref{consdue}) immediately yields the stated result.
\end{proof}

As compared to a $K_{ab}(\beta)$ which is a pure Lie--derivative, the  evolution $K_{ab}(\beta)\in C^{\infty }(\Sigma \times [0,\beta^*],\otimes _S^2T^*\Sigma )$ of a second fundamental form possessing a div--free component,   $K(\beta=0)= 2\,\delta ^*_{g}\,w^{\sharp }+K_T$, with $K_T\in Ker \delta_g$, is rather subtler. $K_T$ can generate elements  $\in Im\, \delta_{g(\beta)}$, (recall that, according to lemma \ref{commlemm}  the div--free character of $K_T$ is not preserved by the linearized Ricci flow--see (\ref{comm2})). In particular, a basic issue one faces when dealing with the linearized Ricci flow is to characterize those $K_T\in Ker \delta_g$ which under the evolution $\bigcirc_d\,K_{ab}(\beta)=0$, $K_{ab}(\beta=0)=K^T_{ab}$ do not dissipate away and give rise to  a $K(\beta)$ with a non trivial component in  $Ker \delta_{g(\beta)}$, for all $\beta\in [0,\beta^*]$. 

\noindent A rather complete answer to such a question is provided by the
\begin{theorem} (Non dissipative directions for $K(\beta)$) \label{transvmode}\\
Let $H^{ab}(\eta)\in C^{\infty }(\Sigma \times [0,\beta^*],\otimes _S^2T\Sigma )$ be a  conjugate heat flow 
$\bigcirc_d^*\,H^{ab}(\eta)=0$, with $H(\eta=0)\in Ker\,\delta_{g(\eta=0)}$. If  $K_{ab}(\beta)\in C^{\infty }(\Sigma \times [0,\beta^*],\otimes _S^2T^*\Sigma )$ is a solution of  $\bigcirc_d\,K_{ab}(\beta)=0$, with a generic initial condition $K(\beta=0):=K\in \mathcal{C}_{g}(\Sigma )$, $K= 2\,\delta ^*_{g}\,w^{\sharp }+K_T$, such that
\begin{equation}
\int_{\Sigma }\,K_{ab}\,H^{ab}(\beta^*)\,d\mu_{g}\not=0\;,
\label{initialdiv}
\end{equation}
where $H^{ab}(\beta^*):=H^{ab}(\eta=\beta^*)$,  then $\{K(\beta)\}\cap \,Ker\,\delta_{g(\beta)}\not=\emptyset$, for all $\beta\in[0,\beta^*]$, and 
\begin{equation}
K\mapsto K(\beta)\in \mathcal{U}_{g(\beta)}\cup\,\mathcal{S}_{g(\beta )}\times [0,\beta^*]\;,
\label{Srew}
\end{equation}
provides a non--trivial deformation of the underlying Ricci flow. In particular, for $\beta=\beta^*$ we can write
\begin{equation}
K^T_{ab}(\beta^*)=\,\sum_n\,\bar{\Phi} _{ab}^{(n,\,T)}\,\int_{\Sigma }\,K_{ij}\,\Phi ^{ij}_{(n,\,T)}(\beta^*)\,d\mu_{g}\;,
\label{modexpTT}
\end{equation}
where $\{\bar{\Phi} ^{ab}_{(n,\,T)}\}$ and ${\Phi} ^{ab}_{(n,\,T)}(\eta)$ respectively denote the div--free eigentensors of the operator $-\Delta _L+\mathcal{R}(\bar{g})$ on $(\Sigma, \overline{g})$ and the associated backward flows on $\Sigma \times [0,\beta^*]$ generated by $\bigcirc_d^*\,{\Phi} ^{ab}_{(n,\,T)}(\eta)=0$, \, 
$\bar{\Phi} ^{ab}_{(n,\,T)}(\eta=0)=\bar{\Phi} ^{ab}_{(n,\,T)}$. 
\end{theorem}

\begin{proof}
We exploit the fact that according to (\ref{Kcomm1*}) and (\ref{divH}) (written for a generic $H^{ab}(\eta)\in C^{\infty }(\Sigma \times [0,\beta^*],\otimes _S^2T\Sigma )$),
the conjugate heat equation $\bigcirc_d^*\,H^{ab}(\eta)=0$ preserves, along the interpolating Ricci flow, the divergence--free character of $H^{ab}(\eta)$, if this holds initially (for $\eta=0$). Thus, if $\nabla_b\,H^{bi}(\eta=0)=0$, then 
$\nabla_b\,H^{bi}(\eta)=0$, for all $\eta\in[0,\beta^*]$. Let $H(\eta)$ be any such a solution of $\bigcirc_d^*\,H^{ab}(\eta)=0$ with 
$\nabla_b\,H^{bi}(\eta=0)=0$. Let us consider the heat flow $K(\beta)$, $\bigcirc_d\,K_{ab}(\beta)=0$, with a generic initial condition $K(\beta=0):=K\in \mathcal{C}_{g}(\Sigma )$,\; $K=2\delta^*_{g}\,v^{\sharp }\,+K^T$, for some $v^{\sharp }\in C^{\infty }(\Sigma,T^*\Sigma )$ and $K^T\in\, Ker\,\delta_{g}$. Since $H(\eta=\beta^*-\beta)\in Ker\,\delta_{g(\beta)}$ for all $\beta\in[0,\beta^*]$, if (\ref{initialdiv}) holds we necessarily have 
\begin{eqnarray}
&& 0\not=\int_{\Sigma }\,K_{ab}\,H^{ab}(\beta^*)\,d\mu_{g}= \int_{\Sigma }\,K^T_{ab}\,H^{ab}(\beta^*)\,d\mu_{g}\\
&& =\int_{\Sigma }\,K^T_{ab}(\beta)\,H^{ab}(\beta)\,d\mu_{g(\beta)}\;,\;\;\;\;\forall \beta\in [0,\beta^*]\;, \nonumber
\end{eqnarray}
 where $H^{ab}(\beta):=H^{ab}(\eta=\beta^*-\beta)$, and where we have exploited the fact that, by $L^2(\Sigma ,d\mu_{g(\beta)})$ conjugacy, the inner product $\int_{\Sigma }K_{ab}(\beta)\,H^{ab}(\beta)\,d\mu_{g(\beta)}$ is constant along the solutions $K(\beta)$ and  $H(\eta)$ of $\bigcirc_d\,K(\beta)=0$  and $\bigcirc^*_d\,H(\eta)=0$, respectively. Thus $\{K(\beta)\}\cap \,Ker\,\delta_{g(\beta)}\not=\emptyset$, for all $\beta\in[0,\beta^*]$, and the flow $K\mapsto K(\beta)$ necessarily has a non--vanishing $\mathcal{S}_{g(\beta )}$--component in the affine slice parametrization $\mathcal{U}_{g(\beta)}\cup\,\mathcal{S}_{g(\beta )}\times [0,\beta^*]$ associated with the underlying Ricci flow. 
If $\{\bar{\Phi }^{ab}_{n,\,T}\}$ denote the orthonormal set of div--free eigentensor of $-\Delta _L+\mathcal{R}(\bar{g})$ on $(\Sigma ,\bar{g})$, and since $K^T_{ab}(\beta^*)$ is $C^{\infty }$, we can consider the $L^2(\Sigma, d\mu_{\bar{g}} )$ mode expansion
\begin{equation}
K^T_{ab}(\beta^*)=\,\sum_n\,\bar{\Phi} _{ab}^{(n,\,T)}\,\int_{\Sigma } \,K^T_{ij}(\beta^*)\,\bar{\Phi} ^{ij}_{(n,\,T)}\,d\mu_{\bar{g}}\;.
\end{equation}
For $0\leq\eta\leq\beta^*$, let us denote by  $\{{\Phi }^{ab}_{(n,\,T)}(\eta)\}_{n\in \mathbb{N}}$, $0\leq\eta\leq\beta^*$ the flows defined by 
\begin{equation}
\bigcirc^*_d\,\{{\Phi }^{ab}_{(n,\,T)}(\eta)\}=0\;,\;\;\;\;\;\{{\Phi }^{ab}_{\,n,\,T)}(\eta=0)\}:=\{\bar{\Phi }^{ab}_{(n,\,T)}\}\;.
\end{equation}
According to (\ref{divH}), these flows preserve the div--free character of the initial $\{\bar{\Phi }_{ab}^{(n,\,T)}\}$, and are conjugated to the forward evolution  defining $K^T_{ab}(\beta)$. Thus,  as in the proof of theorem \ref{LtwoK}, the relation (\ref{modexpTT}) immediately follows from
\begin{equation} 
\frac{d}{d\beta}\int_{\Sigma}\,K_{ab}(\beta)\,{\Phi }^{ab}_{\,(n,\,T)}(\eta)\,d\mu_{g(\beta)}=0\;,
\end{equation}
which holds for each conjugated pair $\left({\Phi }^{ab}_{(n,\,T)}(\eta),\,K^T_{ab}(\beta) \right)$, $n\in\mathbb{N}$.
\end{proof}

\section{$\mathcal{F}$--energy  stability of Ricci flow conjugation}
\label{nondissdir}

If in theorem \ref{transvmode} we identify $H^{ab}$ with the second fundamental form 
$\bar{K}\in \overline{\mathcal{C}}_{\bar{g}}$, it follows  that the divergence--free part $\bar{K}_T\in Ker\,\delta_{\bar{g}}$ of 
$\bar{K}$ provides, through its conjugate evolution $\bigcirc^*_d\,\bar{K}(\eta)=0$,\; $\bar{K}(\eta=0)=\bar{K}_T$, a reference non--dissipative direction for the forward  evolution $K(\beta)$ of the second fundamental form $K\in {\mathcal{C}}_{{g}}$.  These reference directions are also related to the behavior of the Perelman functional $\mathcal{F}$ on the pencil of conjugated trajectories around the underlying fiducial Ricci flow $\beta\mapsto g(\beta)$. As a consequence, they can be used to characterize a form of entropic stability of Ricci flow conjugation along a generic interpolating Ricci flow. Roughly speaking, 
we expect that Ricci flow conjugation is a sensible mapping between Einstein data sets if the fiducial Ricci flow interpolating between $(\Sigma ,g)$ and  $(\Sigma ,\overline{g})$ is, in a suitable sense, stable. Generalized fixed point stability  (see \emph{e.g.} \cite{Cao},  \cite{guenther}, \cite{sesum2}) is not particularly interesting in our setting since if the interpolating flow is Ricci flat or, say, a shrinking Ricci soliton, then the associated Ricci flow conjugation is basically a diffusive rescaling of data. On the other hand, for the case of interest to us, i.e.  around a Ricci flow trajectory which is not a (generalized) fixed point, the only sensible notion of stability is entropic stability in moving from the fiducial flow to a nearby perturbed Ricci flow.  Thus we introduce the  

\begin{definition} 
A Ricci flow conjugation between the Einstein initial data sets ${\mathcal{C}}_{{g}}(\Sigma)$ and $\overline{\mathcal{C}}_{\overline{g}}(\Sigma)$ is said to be 
$\mathcal{F}$--stable if the $\mathcal{F}$--energy of the interpolating Ricci flow is non--increasing   under the perturbation induced by the reference data $\overline{\mathcal{C}}_{\overline{g}}(\Sigma)$.
\end{definition}
\noindent In order to discuss this characterization of $\mathcal{F}$--stability  we need a minor technical result extending Ricci flow conjugation to $L^{2}(\Sigma_{\eta}\times [0,\beta^*],e^{-f(\eta)}\,d\mu_{g(\eta)})$. 

\begin{lemma} 
\label{prop2}
Along the fiducial Ricci--Perelman flow $\eta \mapsto (g(\eta),f(\eta))$, defined by (\ref{eqf}), consider the  backward evolution $\bigcirc ^{*}_{L,f}\,{\psi }^{ab}=0$ of a symmetric bilinear form $\psi ^{ab}(\eta =0)\in C^{\infty }(\Sigma ,\otimes ^{2}T\,\Sigma)$ defined by the parabolic initial value problem  
\begin{equation} 
\begin{tabular}{l}
$\bigcirc ^{*}_{L,f}\,{\psi }^{ab}\doteq \left(\frac{\partial }{\partial \eta }-\Delta_{L}+\,2\nabla ^{i}f\,\nabla_{i} \right){\psi}^{ab}=0\;,$\\
\\
${\psi}^{ab}(\eta=0)={\psi}^{ab}_{*}\;.$%
\end{tabular}
\;   \label{transVol2}
\end{equation}  
Then, the resulting  flow $\eta \mapsto {{\psi}}^{ab}(\eta)$ is $L^{2}(\Sigma_{\eta}\times [0,\beta^*],e^{-f(\eta)}\,d\mu_{g(\eta)})$--conjugated to
the solution $\beta \mapsto {{h}}_{ab}(\beta )$, $\beta\in [0,\beta ^{*}]$,\,${{h}}_{ab}(\beta =0)=h_{ab}(\beta=0)$ of the linearized Ricci flow (\ref{divfree}), \emph{i.e.},
\begin{equation}
\frac{d}{d\eta }\int_{\Sigma }{\psi}^{ab}(\eta )\,{h}_{ab}(\eta )\,e^{-f(\eta)}\,d\mu _{g(\eta )}=0\;,
\end{equation}
and we get the conservation laws
\begin{equation}
\label{bello1f}
\frac{d}{d\eta }\,\int_{\Sigma }R_{ab}(\eta )\psi^{ab}(\eta )\,e^{-f(\eta)}\,d\mu _{g(\eta )}=0\;,
\end{equation}

\begin{equation}
\label{bello2f}
\frac{d}{d\eta }\,\int_{\Sigma }\left(g_{ab}(\eta)-2\eta\,R_{ab}(\eta )\right)\psi^{ab}(\eta )\,e^{-f(\eta)}\,d\mu _{g(\eta )}=0\;.
\end{equation}
Moreover, if, for $\eta=0$, $\left(\psi^{ab}\,e^{-f}\right)\in Ker\;\delta_{g}$, then $\left(\psi^{ab}(\eta)\,e^{-f(\eta)}\right)\in Ker\;\delta_{g(\eta)}$,  $\forall \eta \in [0,\beta ^{*}]$. 
\end{lemma}
\begin{proof}
It is easily checked that under the evolutions (\ref{eqf})  and (\ref{transVol2})  the flow $\eta\mapsto \psi^{ab}(\eta)\,e^{-f(\eta)}$ solves (\ref{transVol}). Thus, the above results immediately follows from theorem \ref{nondiss}.
\end{proof}
\noindent
Let $g(\eta)\mapsto\mathcal{F}[g(\eta),\,f(\eta)]$, $\eta\in[0,\beta^*]$, the valuation of Perelman $\mathcal{F}$--energy on $\eta\mapsto (g(\eta),f(\eta))$, considered as a fiducial flow  on $\Sigma\times [0,\beta^*]$. 
We are interested in  the behaviour of $\mathcal{F}[g(\eta),\,f(\eta)]$ in a tubular neighborhood of $\eta\mapsto (g(\eta),f(\eta))$. To this end  we need the explicit formula for an $\eta$--dependent linearization $D\,\mathcal{F}[{g(\eta)};f(\eta)]\circ \left(\psi ^{ab}(\eta),\phi (\eta)\right)$ of $\mathcal{F}$ in the direction of an arbitrary variation
\begin{equation}
g^{ab}_{(\epsilon )}(\eta):=g^{ab}(\eta)+\,\epsilon \,\psi ^{ab}(\eta)\;,\;\; g^{ab}_{(\epsilon )}(\eta)\in \mathcal{M}et(\Sigma )\;,\forall \epsilon \in [0,1]\;,
\end{equation}
and
\begin{equation}
f_{(\epsilon )}(\eta):=f(\eta)+\,\epsilon\,\phi (\eta)\;,
\end{equation}
of the fiducial backward flow $\eta \mapsto (g_{ab}(\eta),\, f(\eta))$. A standard computation, (see \emph{e.g.} \cite{Glickenstein}, Lemma 5.3), provides
\begin{eqnarray}
&&D\,\mathcal{F}[{g(\eta)};f(\eta)]\circ \left(\psi ^{ab}(\eta),\phi (\eta)\right):=\left.\frac{d}{d\epsilon }\,\mathcal{F}[{g_{(\epsilon )}(\eta)};f_{(\epsilon )}(\eta)]\right|_{\epsilon =0}=\\
\nonumber\\
&&-\int_{\Sigma }\,\psi ^{ab}(\eta )\left(\mathcal{R}_{ab}(\eta )+\nabla_a \nabla_b\, f(\eta)\right){\rm
e}^{-f(\eta)}\,d\mu _{{g(\eta)}}\nonumber\\
\nonumber\\
&&+\int_{\Sigma }\,\left(\frac{\Psi (\eta )}{2}-\phi (\eta) \right)\,\left(2\triangle f(\eta) -|\nabla f(\eta)|^2 +\mathcal{R}(\eta)\right){\rm
e}^{-f(\eta)}\,d\mu _{{g(\eta)}}\;,\nonumber
\end{eqnarray}

\noindent where we have set $\Psi (\eta ):=\psi ^{ab}(\eta )\,g_{ab}(\eta)$. By considering variations $\phi(\eta)$ preserving the volume form ${\rm
e}^{-f(\eta)}\,d\mu _{{g(\eta)}}$, (\emph{i.e.}, by choosing $\phi(\eta)\equiv \frac{\Psi (\eta )}{2}$), we get
\begin{eqnarray}
\label{lambdavar}
&&\left.\frac{d}{d\epsilon }\,\mathcal{F}[{g_{(\epsilon )}(\eta)};f_{(\epsilon )}(\eta)]\right|_{\epsilon =0}=\\
\nonumber\\
&&-\int_{\Sigma }\,\psi ^{ab}(\eta )\left(\mathcal{R}_{ab}(\eta )+\nabla_a \nabla_b\, f(\eta)\right){\rm
e}^{-f(\eta)}\,d\mu _{{g(\eta)}}\nonumber\;.
\end{eqnarray}
Let us restrict the variation (\ref{lambdavar})  to perturbations $\psi^{ab}(\eta)$ solution of the $L^2(\Sigma_{\eta}\times [0,\beta^*],\,e^{-f(\eta)\,d\mu_{g(\eta)}})$--conjugated linearized Ricci flow (\ref{transVol2}), (see case (i) of Lemma \ref{prop2}),
\begin{equation} 
\begin{tabular}{l}
$\left(\frac{\partial }{\partial \eta }-\Delta_{L}+\,2\nabla ^{i}f(\eta)\,\nabla_{i} \right){\psi}^{ab}(\eta)=0\;,$\\
\\
${\psi}^{ab}(\eta=0)={\psi}^{ab}_{*}\;.$%
\end{tabular}
\;   \label{transVol3}
\end{equation}  
In this case we have 
\begin{theorem}
Let us consider the set of bilinear forms
\begin{equation}
\Psi _{\bot }\doteq \left\{{\psi}_*^{ab}\in C^{\infty }(\Sigma,\otimes ^2 T\Sigma)\,:\,\psi_*^{ab}\,{\rm e}^{-f}\in Ker\;\delta _{\overline{g}} \right\}\;,
\end{equation}
which are $L^2(\Sigma,\,e^{-f}\,d\mu_{g})$--orthogonal to $Im\;\delta^* _{\overline{g}}$. Let $\{\eta\mapsto {\psi}^{ab}(\eta)\,:{\psi}^{ab}(\eta=0)\in \Psi _{\bot } \}$ be the pencil of parabolic flows solution of (\ref{transVol3}) with initial data varying in $\Psi _{\bot }$, \emph{i.e.},
\begin{equation} 
\begin{tabular}{l}
$\frac{\partial }{\partial \eta }\,{\psi}^{ab}(\eta)=\Delta_{L}\,{\psi}^{ab}(\eta)-\,2\nabla ^{i}f(\eta)\,\nabla_{i} {\psi}^{ab}(\eta)\;,$\\
\\
${\psi}^{ab}(\eta=0)={\psi}^{ab}_{*}\in \Psi _{\bot } \;.$%
\end{tabular}
\;   \label{transVol4}
\end{equation}  
Then,  the corresponding variation $g^{ab}_{(\epsilon )}(\eta):=g^{ab}(\eta)+\,\epsilon \,\psi ^{ab}(\eta)$ of the fiducial backward Ricci flow
$\eta\mapsto g_{ab}(\eta)$, generates a constant  shift in the Perelman functional, \emph{i.e.},
\begin{eqnarray}
&&\left.\frac{d}{d\epsilon }\,\mathcal{F}[{g_{(\epsilon )}(\eta)},\,f(\eta)]\right|_{\epsilon =0}=-\int_{\Sigma }\,\mathcal{R}_{ab}(\eta )\,\psi ^{ab}(\eta )\,{\rm
e}^{-f(\eta)}\,d\mu _{{g(\eta)}}=\\
\nonumber\\
&&-\int_{\Sigma }\,\overline{\mathcal{R}}_{ab}\,\psi_* ^{ab}\,{\rm
e}^{-f}\,d\mu _{\overline{g}}\nonumber\;.
\end{eqnarray}
\end{theorem}

\begin{proof}
Along the flow (\ref{transVol3}) let us rewrite (\ref{lambdavar}) as
\begin{eqnarray}
\label{van}
&&\left.\frac{d}{d\epsilon }\,\mathcal{F}[{g_{(\epsilon )}(\eta)},\,f(\eta)]\right|_{\epsilon =0}=
-\int_{\Sigma }\mathcal{R}_{ab}(\eta )\,\psi ^{ab}(\eta ){\rm
e}^{-f(\eta)}\,d\mu _{{g(\eta)}}-\\
\nonumber\\
&&-\int_{\Sigma }[\delta^{*}_{g(\eta)}(w(\eta))]_{ab}\,\psi ^{ab}(\eta ){\rm
e}^{-f(\eta)}\,d\mu _{{g(\eta)}}\nonumber\;,
\end{eqnarray}
where we have set $w_k(\eta)\doteq \,\nabla_k\,f_{(1)}(\eta)$. According to proposition \ref{prop2}, since $\psi ^{ab}(\eta=0 )\in \Psi _{\bot }$, we have that  $\psi ^{ab}(\eta ){\rm
e}^{-f(\eta)}\in Ker\,\delta_{g(\eta)}$, $\forall \eta\in [0,\beta^*]$ and the last term in (\ref{van}) vanishes by $L^2(\Sigma,\,e^{-f}\,d\mu_{g})$--orthogonality. Thus, along (\ref{transVol3}) $\left.\frac{d}{d\epsilon }\,\mathcal{F}[{g_{(\epsilon )}(\eta)},\,f(\eta)]\right|_{\epsilon =0}$ reduces to $-\int_{\Sigma }\mathcal{R}_{ab}(\eta )\,\psi ^{ab}(\eta ){\rm
e}^{-f(\eta)}\,d\mu _{{g(\eta)}}$, which, again by proposition \ref{prop2}, is a conserved quantity.
\end{proof}
As an immediate consequence of this result we have that  the initial data set $\Psi _{\bot }$ can be used to parametrize the pencil of linear perturbations around a generic (\emph{i.e.}  non Ricci--flat solitonic) backward Ricci flow. In particular we have the following characterization of perturbed backward Ricci flow trajectories
\begin{lemma}
Let $\eta \mapsto g(\eta)$ denote a fiducial backward Ricci flow, and  assume that $g(\eta)$ is not a Ricci--flat soliton. Let
$\mathcal{P}[\Psi _{\bot };g(\eta)]\doteq \{\eta\mapsto \left({\psi}^{ab}(\eta),\,f(\eta)\right)\;:{\psi}^{ab}(\eta=0)\in \Psi _{\bot } \}$ be the corresponding pencil of parabolic flows solution of (\ref{transVol2}) with initial data varying in $\Psi _{\bot }$. A flow  $\eta\mapsto (\psi(\eta),f(\eta))$ with 
${\psi}^{ab}(\eta=0)\in \Psi _{\bot }$ is $\mathcal{F}[{g(\eta)},\,f(\eta)]$--energy increasing (decreasing)
\begin{equation}
\left.\frac{d}{d\epsilon }\,\mathcal{F}[{g_{(\epsilon )}(\eta)},\,f(\eta)]\right|_{\epsilon =0}\,>\,0\,,\;\;\;\;\; (<0)\;,
\end{equation}
if, for $\eta=0$,  $\int_{\Sigma }\,\overline{\mathcal{R}}_{ab}\,\psi ^{ab}\,{\rm
e}^{-f}\,d\mu _{\overline{g}}<0$, ($>0$).
\end{lemma}
\noindent
We can apply this result to the the backward evolution $\eta\mapsto (\overline{\varrho }(\eta),\,\overline{K}(\eta))$ of the (reference) matter density  and second fundamental form $(\overline{\varrho },\,\overline{K})\in \overline{\mathcal{C}}_{\overline{g}}(\Sigma )$ so as to obtain the following entropic characterization of the stability of Ricci flow conjugation around  a generic interpolating Ricci flow. 
\begin{theorem}
\label{capsuleFin}
For $\epsilon >0$ small enough and $0\leq\beta\leq\beta^*$, let $\Omega_{\epsilon }(g(\beta))$
\begin{equation}
:= \left\{\left. g(\beta)+\, h(\beta)\,\right|\; h\in \mathcal{T}_{(\Sigma,g(\beta))}\,\mathcal{M}et(\Sigma)\;,\;\| h(\beta)\| _{L^{2}(\Sigma,d\mu_{g(\beta)})}<\epsilon    \right\}\;,\nonumber
\end{equation}
denote the (affine) $\epsilon $--tubular neighborhood of the fiducial Ricci flow  $\beta\mapsto g(\beta)$ in $\mathcal{M}et(\Sigma)$. We 
assume that $\beta\mapsto g(\beta)$  is not a Ricci--flat soliton over $\Sigma\times [0,\beta^*]$. If $\overline{K}_{TT}$ is the trace--free and divergence--free part of $\overline{K}\in \overline{\mathcal{C}}_{\overline{g}}(\Sigma)$,\, then the reference flow $\eta\mapsto (g(\eta),\overline{\mathcal{C}}^{\;\sharp }(\eta))$ is  $\mathcal{F}[{g(\eta)},\,\overline{\varrho}(\eta)]$--energy decreasing (increasing) 
in  $\Omega_{\epsilon }(g(\beta))$, i.e. 
\begin{equation}
\left.\frac{d}{d\epsilon }\,\mathcal{F}[{g_{(\epsilon )}(\eta)},\,\overline{\varrho}(\eta)]\right|_{\epsilon =0}\,<\,0\,,\;\;\;\;\; (>0)\;,
\end{equation}
and the Ricci flow conjugation between the two  data sets ${\mathcal{C}}_{{g}}(\Sigma)$ and $\overline{\mathcal{C}}_{\overline{g}}(\Sigma)$ is 
$\mathcal{F}$--stable  \;(unstable) in the $\overline{K}$--direction
if for $\eta=0$ we have
\begin{equation}
\label{entrFin}
\mathfrak{F}(\overline{g},\,\overline{K}):=\int_{\Sigma }\,\left(\overline{\mathcal{R}}_{\;ab}\,\overline{K}_{TT}^{\;ab}+ \frac{1}{n}\,\overline{\mathcal{R}}\,tr_{\overline{g}}\,\overline{K}\right)\,d\mu _{\overline{g}}\,>\,0\,,\;\;(<0)\;.
\end{equation}
\end{theorem}
\begin{proof}
Along the flow $\eta\longmapsto (\overline{K}(\eta)\,,\overline{\varrho}(\eta))$ solution of the conjugate heat flow $\bigcirc ^{*}_d\,(\overline{K}(\eta)\,,\overline{\varrho}(\eta))=0$,\;with  $(\overline{K}(\eta=0)=\overline{K}_T\,,\overline{\varrho}(\eta=0)>0)$, both the divergence--free condition and the positivity condition are preserved. It immediately follows that $\psi(\eta):=\overline{\varrho}^{-1}(\eta)\,\overline{K}(\eta)$ is a solution of $\bigcirc ^{*}_{L,f}\,\psi^{ab}(\eta)=0$, with 
$f(\eta):=-\ln\,\overline{\varrho}(\eta)$, such that ${\psi}^{ab}(\eta=0)\in \Psi _{\bot }$. Then, the above lemma provides the stated result if we factorize $\overline{K}(\eta=0)=\overline{K}_T$ in its $TT$--part $\overline{K}_{TT}$ plus the trace $tr_{\overline{g}}\,\overline{K}$.  
\end{proof}
\noindent Along the same lines, (by exploiting  (\ref{bello2f})), one could discuss entropic stability of Ricci flow conjugation with respect to Perelman's shrinker functional $\mathcal{W}(g,\varrho ,\tau)$. The details of such an analysis will be discussed in a forthcoming paper.

\section{Conclusions}
\label{sezioneOtto}
The works \cite{buchert1,buchert2,buchert3}  describe a number of potential applications of
Ricci flow conjugation, such as producing averaged data for cosmological spacetimes,   computing backreaction terms to the constraint equations,
and in general giving (or rather trying to give) a sound mathematical basis to the challenging mathematical and physical problem of averaging in cosmology. The properties of Ricci flow conjugation discussed here indicates clearly that an averaging procedure based on  Ricci flow is mathematically feasible and when the averaging scale is not too large (\emph{i.e.} when $\eta\,|\overline{Rm}(\eta)|<<1$), such a procedure  corresponds, according to Theorem \ref{heatrepr},  to a form of local Gaussian averaging dressed with a rich spectrum of corrections terms of geometrical origin. Clearly, the nice geometrical properties of Ricci flow must come to terms with the intricacies of what should be considered as a  physically sound averaging technique in relativistic cosmology. Indeed, modern high precision cosmology calls into play delicate averaging issues \cite{ellis} ranging from frame effects, localized averaging over past light cone, multiscale averaging and the geometrical characterization of a corresponding distance ladder, just to mention a few \cite{rasanen, wiltshire}. Thus,  it is still an open problem to establish what the most appropriate averaging technique may be. In any case a general prescription for comparing (generalized) Einstein initial data sets seems, from the point of view of mathematical cosmology, a necessary step in such an averaging scenario and Ricci flow conjugation suggestes itself as a natural technique unifying in a unique geometrical framework several nice features: \; \emph{(i)} It sets a coherent averaging scale between matter and geometry which goes beyond a naive volume averaging; \; \emph{(ii)} It relates matter and geometrical averaging to energy conditions; \; \emph{(iii)} It provides a precise control over the entropic stability of the  relative matter--geometry fluctuations  between the given data sets. 
\bigskip

\subsection*{Aknowledgements}
The authors would like to thank T. Buchert and  G. F. R. Ellis for useful conversations at the initial stages of preparation of this paper.

\bibliographystyle{amsplain}

\end{document}